\documentclass{article}
\usepackage{graphicx} 
\usepackage[noadjust]{cite}
\usepackage{latexsym,amssymb,enumerate,amsmath,epsfig,amsthm,authblk,dsfont,fullpage}
\usepackage{bm}
\usepackage{multirow}
\usepackage{setspace,color,soul}
\usepackage{tikz}
\usepackage{algorithm,algpseudocode}
\usepackage{ulem}
\usetikzlibrary{shapes,arrows}

\usepackage{tikz-cd}
\usepackage{accents}

\usepackage{appendix}
\usepackage{graphicx,subfigure}
\usepackage{makecell}
\usepackage{tablefootnote} 
\usepackage{epstopdf}
\usepackage{grffile}
\usetikzlibrary{decorations.pathreplacing}
\usepackage{wrapfig}
\usepackage{lineno}

\usepackage{hyperref}
\hypersetup{
	colorlinks=true,
	linkcolor=blue,
	citecolor=blue,
	filecolor=blue,      
	urlcolor=blue,
}
\allowdisplaybreaks

\newcommand{\commentout}[1]{}

\DeclareMathOperator*{\argmax}{arg\,max}

\newcommand{\resid}{\text{resid}}
\newcommand{\bc}{\mathbf{c}}
\newcommand{\be}{\mathbf{e}}

\newcommand{\bb}{\mathbf{b}}
\newcommand{\wT}{\widetilde{T}}
\newcommand{\by}{\mathbf{y}}
\newcommand{\bx}{\mathbf{x}}
\newcommand{\bz}{\mathbf{z}}
\newcommand{\bv}{\mathbf{v}}
\newcommand{\bd}{\mathbf{d}}

\newcommand{\bA}{\mathbf{A}}
\newcommand{\bB}{\mathbf{B}}
\newcommand{\bF}{\mathbf{F}}

\newcommand{\spark}{\text{spark}}
\newcommand{\proj}{\text{proj}}
\newcommand{\pull}{Q}

\newcommand{\gsup}{\text{supp}_{\text{G}}}
\newcommand{\supp}{\text{supp}}
\newcommand{\rank}{\text{rank}}

\newtheorem{example}{Example}

\newtheorem{theorem}{Theorem}
\newtheorem{proposition}{Proposition}

\newtheorem{lemma}{Lemma}
\theoremstyle{definition}
\newtheorem{definition}{Definition}[section]

\numberwithin{equation}{section}

\newcommand\keywords{%
\noindent\textbf{Keywords:}\ }

\title{Group Projected Subspace Pursuit for Block Sparse Signal Reconstruction: Convergence Analysis and Applications }
\author{Roy Y. He\thanks{Department of Mathematics, City University of Hong Kong, Kowloon Tong, Hong Kong. Email: royhe2@cityu.edu.hk. The work of Roy Y. He was partially supported by CityU 7200779.}, 
Haixia Liu\thanks{School of Mathematics and Statistics  \& Institute of Interdisciplinary Research for Mathematics and Applied Science \& Hubei Key Laboratory of Engineering Modeling and Scientific Computing, Huazhong University of Science and Technology, Wuhan, Hubei, China. Email: liuhaixia@hust.edu.cn. The work of H.X. Liu was supported in part by NSFC 11901220 and Hubei Key Laboratory of Engineering Modeling and Scientific Computing.}, 
Hao Liu\thanks{Department of Mathematics, Hong Kong Baptist University, Kowloon Tong, Hong Kong. Email: haoliu@hkbu.edu.hk. The work of Hao Liu was partially supported by National Natural Science Foundation of China  12201530, HKRGC ECS 22302123, HKBU 179356.}

}
\begin{document}
\date{}
\maketitle

\begin{abstract}
In this paper, we present a convergence analysis of the Group Projected Subspace Pursuit (GPSP) algorithm proposed by He et al. \cite{he2023group} (Group Projected subspace pursuit for IDENTification of variable coefficient differential equations (GP-IDENT), {\it Journal of Computational Physics}, 494, 112526) and extend its application to general tasks of block sparse signal recovery. We prove that when the sampling matrix satisfies the Block Restricted Isometry Property (BRIP) with a sufficiently small Block Restricted Isometry Constant (BRIC), GPSP exactly recovers the true block sparse signals. When the observations are noisy, this convergence property of GPSP remains valid if the magnitude of the true signal is sufficiently large.   GPSP selects the features by subspace projection criterion (SPC) for candidate inclusion and response magnitude criterion (RMC)  for candidate exclusion. We compare these criteria with counterparts of other state-of-the-art greedy algorithms. Our theoretical analysis and numerical ablation studies reveal that SPC is critical to the superior performances of GPSP, and that  RMC can enhance the robustness of feature identification when observations contain noises. We test and compare GPSP with other methods in diverse settings, including heterogeneous random block matrices, inexact observations, face recognition, and PDE identification. We find that GPSP outperforms the other algorithms in most cases for various levels of block sparsity and block sizes, justifying its effectiveness for general applications.

\vspace{0.5cm}

\end{abstract}
\section{Introduction}
Recent advances in technology facilitate collecting and storing large amounts of high-dimensional data at low cost. Typically, each dimension represents a feature such as temperature, rating, or pricing, and many applications seek models to predict the quantity of interest based on simple combinations of these features.  To achieve high accuracy while avoiding issues of over-fitting as well as redundancy, it is a common practice to focus on only a few significant features. Feature selection is a challenging problem that has been extensively studied both theoretically and numerically.  Classical applications, such as DNA microarray analysis~\cite{kim2007sparse}, image processing~\cite{elad2010role},  portfolio selection~\cite{brodie2009sparse},  validate the long-lasting importance of this problem, and more recent advances in omics data analysis~\cite{vinga2021structured},  PDE identification \cite{schaeffer2017learning,kang2021ident,he2022robust,he2023much,he2023group,tang2023fourier,tang2023weakident} and sparsity-induced deep learning~\cite{molchanov2016pruning,molchanov2017variational,gale2019state,hoefler2021sparsity,menghani2023efficient}   prove that finding critical features will become more important in the era of big data. Providing a reliable and efficient algorithm for feature selection is highly desirable in many fields.

Mathematically, the task of feature selection can be formulated as an optimization problem~\cite{candes2006stable,candes2007sparsity}
\begin{align}
	\mathcal{P}_{\ell_0}:\quad \min \|\bc\|_0 \quad \mbox{ s.t. } \quad \by=\bA\bc.
	\label{eq.l0}
\end{align}
Here $\bA\in\mathbb{R}^{N\times M}$ with positive integers $N$ and $M$ is called a sampling matrix or a feature matrix,  the $\ell_0$-pseudonorm $\|\bc\|_0$ measures the sparsity of a signal $\bc\in\mathbb{R}^M$ by counting the number of non-zero entries, and $\by\in\mathbb{R}^N$ is called the observation vector. By solving~\eqref{eq.l0}, we seek a linear representation of $\by$ using the fewest columns, or features in $\bA$.  When non-zero elements appear in blocks, the sparsity describes the number of non-zero blocks, which is referred as block sparsity. The problem can be formulated as follows
\begin{align}
	\mathcal{P}_{\ell_p/\ell_0}:\quad \min \|\bc\|_{p,0} \quad \mbox{ s.t. } \quad \by=\bA\bc,
	\label{eq.l10}
\end{align}
where $p\geq 1$ is usually a positive integer, $\bA=[\bF_1,\dots,\bF_G]\in\mathbb{R}^{N\times GM}$ with $\bF_g\in\mathbb{R}^{N\times M}$ for $g=1,2,\dots,G$, $\bc=[\bc_1^\top,\dots,\bc_G^\top]^\top\in\mathbb{R}^{GM}$,   $\|\bc\|_{p,0}=\sum_{g=1}^G \mathcal{I}(\|\bc_g\|_p\neq0)$, and $\mathcal{I}$ is the indicator function that outputs $1$ when $\|\bc_g\|_p\neq0$ and $0$ otherwise.  Existence and uniqueness of solution for~\eqref{eq.l0} are established in~\cite{donoho2001uncertainty}, and the counterparts for~\eqref{eq.l10} can be found in~\cite{eldar2009robust,elhamifar2012block,donoho2003optimally}.

Problem~(\ref{eq.l0}) and~\eqref{eq.l10} are closely related and difficult to solve. A signal with $k$ non-zero blocks can be regarded as a signal with at most $k\times M$ non-zero entries, where $M$ denotes the size of each block. Leveraging the block sparsity pattern reduces the complexity. Indeed, the complexity of an exhaustive search for a $k$-sparse signal from solving $\mathcal{P}_{\ell_0}$ is 
$\mathcal{O}(\tbinom{M}{k}k^2N)$, and for a block $k$-sparse signal, this is $\mathcal{O}({\binom{G}{k}}(Mk)^2N)$. 
Moreover, we note that exactly solving~\eqref{eq.l0} and~\eqref{eq.l10} are both NP-hard~\cite{natarajan1995sparse, elhamifar2012block,kamali2013block,eldar2010block}. This means that neither of these problems has polynomial-time algorithms available. One approach to address this challenge is by convex relaxation~\cite{chen1994basis,donoho2006compressed,candes2005decoding}, and the computational complexity can still be high \cite{yang2011alternating}.

Another strategy is to assume that we know an upper-bound for the level of sparsity, e.g., $\|\bc\|_0\leq k$ for some positive integer $k$ much smaller than the dimension of $\bc$. Then we consider a sparsity-constrained least squares problem
\begin{align}
	\quad \min_{\bc} \|\bA\bc-\by\|_2^2 \quad \mbox{ s.t. } \quad \|\bc\|_0\le k.
	\label{eq.scls}
\end{align}
Similarly, for block sparse signal recovery, we consider
\begin{align}
	\quad \min_{\bc} \|\bA\bc-\by\|_2^2 \quad \mbox{ s.t. } \quad \|\bc\|_{p,0}\le k.
	\label{eq.blockscls}
\end{align}
In particular, a vector $\bc$ is called $k$-sparse if  $\|\bc\|_0\leq k$, and block $k$-sparse if $\|\bc\|_{p,0}\leq k$. If $k=\|\bc^*\|_0$, where $\bc^*$ is a solution to problem~\eqref{eq.l0}, then $\bc^*$ is also a solution to~\eqref{eq.scls}. Moreover, if $\bc^*$ is unique, i.e., the sparsest solution~\cite{donoho2003optimally}, then by solving~\eqref{eq.scls} with the correct $k$, we recover $\bc^*$. The relation between~\eqref{eq.l10} and~\eqref{eq.blockscls} is similar. Motivated by this observation, many algorithms are developed for finding $k$-sparse or block $k$-sparse signals that best fit the data, such as alternating direction methods~\cite{deng2013group} and proximal gradient methods~\cite{chen2012smoothing}. Among them, greedy algorithms are highly efficient and easy to implement.

Iterative greedy algorithms are popular techniques for addressing~\eqref{eq.scls} and~\eqref{eq.blockscls}. Some well-known methods for~\eqref{eq.scls} include orthogonal matching pursuit (OMP) \cite{tropp2007signal}, subspace pursuit (SP) \cite{dai2009subspace} and compressive sampling matching pursuit (CoSaMP) \cite{needell2009cosamp}. These algorithms find the support, i.e., indices with non-zero entries, of the unknown signal progressively by iterating different forms of candidate inclusion and exclusion.  With high efficiency, their effectiveness for finding the true sparse signals is also justified by different sufficient conditions involving the restricted isometry property (RIP)~\cite{candes2008restricted}.  For a sample matrix $\bA\in\mathbb{R}^{N\times M}$ and a positive integer $1\leq k\leq M$, RIP  specifies the bounds of eigenvalues of submatrices of $\bA$ with $k$ columns. This is quantified by a scalar restricted isometry constant (RIC),  denoted by $\delta_k\in(0,1)$; smaller $\delta_k$ implies that submatrices of $\bA$ with $k$ columns behave more similarly as an orthonormal transformation. 
Based on RIP, Liu and Temlyakov~\cite{liu2011orthogonal} proved that OMP recovers a $k$-sparse signal after $k$-iterations if the sampling matrix satisfies  $\delta_{k+1}<1/((1+\sqrt{2})\sqrt{k})$. This result was then sharpened by Wang and Shim~\cite{wang2012recovery} and Wen et al.~\cite{wen2016sharp}. For SP, Dai and Milenkovic~\cite{dai2009subspace} showed the convergence under the condition $\delta_{3k}<0.165$ which was then improved to $\delta_{3k}<0.4859$ by Song et al.~\cite{song2014improved}. Convergence analysis for CoSaMP was carried out in a series of works~\cite{needell2009cosamp,foucart2012sparse,song2013improved}. 

These aforementioned algorithms have been extended for recovering block sparse signals as well, including block orthogonal matching pursuit (BOMP) \cite{li2018new},  block subspace pursuit (BSP) \cite{kamali2013block} and block compressive sampling matching pursuit (BCoSaMP) \cite{zhang2019recovery}, which are the block versions of OMP, SP and CoSaMP, respectively. The concept of RIP was extended to block RIP (BRIP)~\cite{candes2005decoding} for developing sufficient conditions of convergence, and they are closely related to the magnitude of the counterpart of RIC, called block Restricted Isometry Constant (BRIC).
Indexed by two positive integers $M$ and $k$, a BRIC $\delta_{M,k}$  corresponding to a block matrix $\bA$ with block size $M$ indicates the difference between any $k$ sub-blocks of $\bA$ and an orthonormal transformation.  Wang et al.~\cite{wang2011analysis} showed that BOMP can exactly recover the block $k$-sparse signal in $k$ steps if $\delta_{M,k+1}<1/(2\sqrt{k}+1)$, which was then improved by Wen et al. to $\delta_{M,k+1}<1/\sqrt{k+1}$ in \cite{wen2019sharp}. Kamali et al.~\cite{kamali2013block} proved that BSP converges to the true block $k$-sparse signal if $\delta_{M,3k}<0.1672$. As for BCoSaMP, Zhang et al.~\cite{zhang2019recovery} showed the convergence when $\delta_{M,4k}<0.5$ is satisfied for the sampling matrix.  For sufficient conditions bounding RIC and BRIC from above, it is often desirable when the respective upper bounds are closer to $1$; it allows exact signal recovery for more general sample matrices.  However, it is in general difficult to evaluate the effectiveness of those sufficient conditions involving constants with different indices.

The Group Projected Subspace Pursuit (GPSP) proposed in~\cite{he2023group} was inspired by noticing limitations of state-of-the-art block-sparse signal reconstruction algorithms when applied to the identification of differential equations (DEs) with varying coefficients \cite{kang2021ident,schaeffer2017learning,he2022robust,tang2023weakident,tang2023fourier}. The objective is to identify active features governing the dynamics of observed data. Each feature corresponds to a possible differential operator depending on space and time, and the total number of true features is assumed to be small.  We note that GPSP, as an extension of SP, is applicable for general block sparse regression problems. Similar to the aforementioned greedy algorithms, GPSP selects candidates by considering highly correlated features.  Instead of measuring the correlation between one vector and a block of vectors by accumulating inner products, GPSP distinguishes itself by considering the projection to the column space spanned by the block of vectors. Figure~\ref{fig_illustrate_inner_proj} illustrates this key difference.  Fu et al.~\cite{fu2014block} proposing BOMPR also noticed the importance of projection in block feature selection. BOMPR follows the same strategy of BOMP where candidates are not filtered once being selected; whereas GPSP incorporates the shrinking stage as BSP but with different criterion. With these algorithmic discrepancies, the convergence results established for BOMPR or BSP cannot be immediately applied, and such analysis was missing from~\cite{he2023group}. In this paper, we investigate both theoretical and numerical aspects of GPSP to validate its effectiveness as a block-sparse signal reconstruction algorithm for general tasks.
\begin{figure}
	\centering
	\begin{tabular}{cc}
		(a)&(b)\\
		\includegraphics[width=0.3\textwidth]{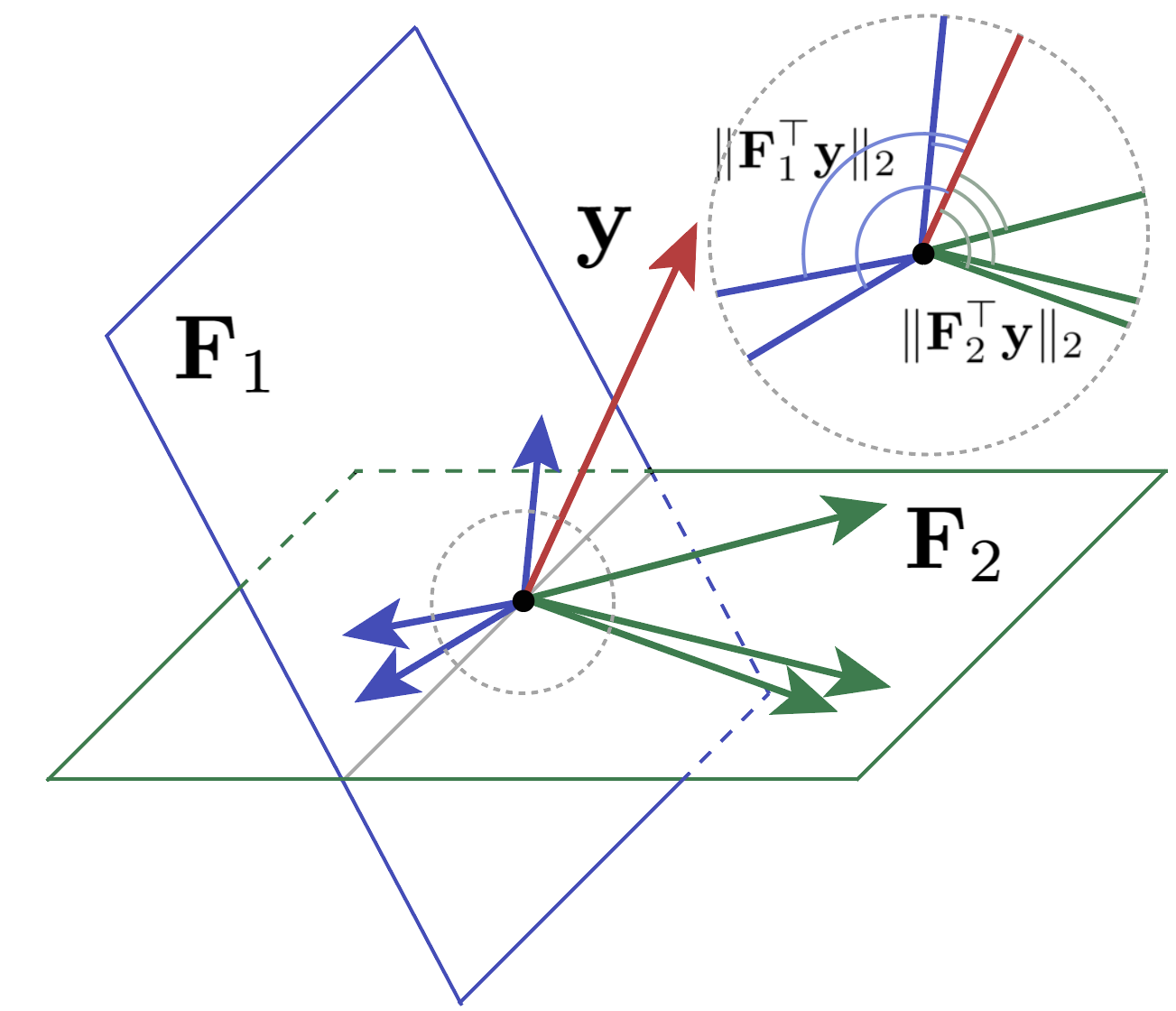}&
		\includegraphics[width=0.3\textwidth]{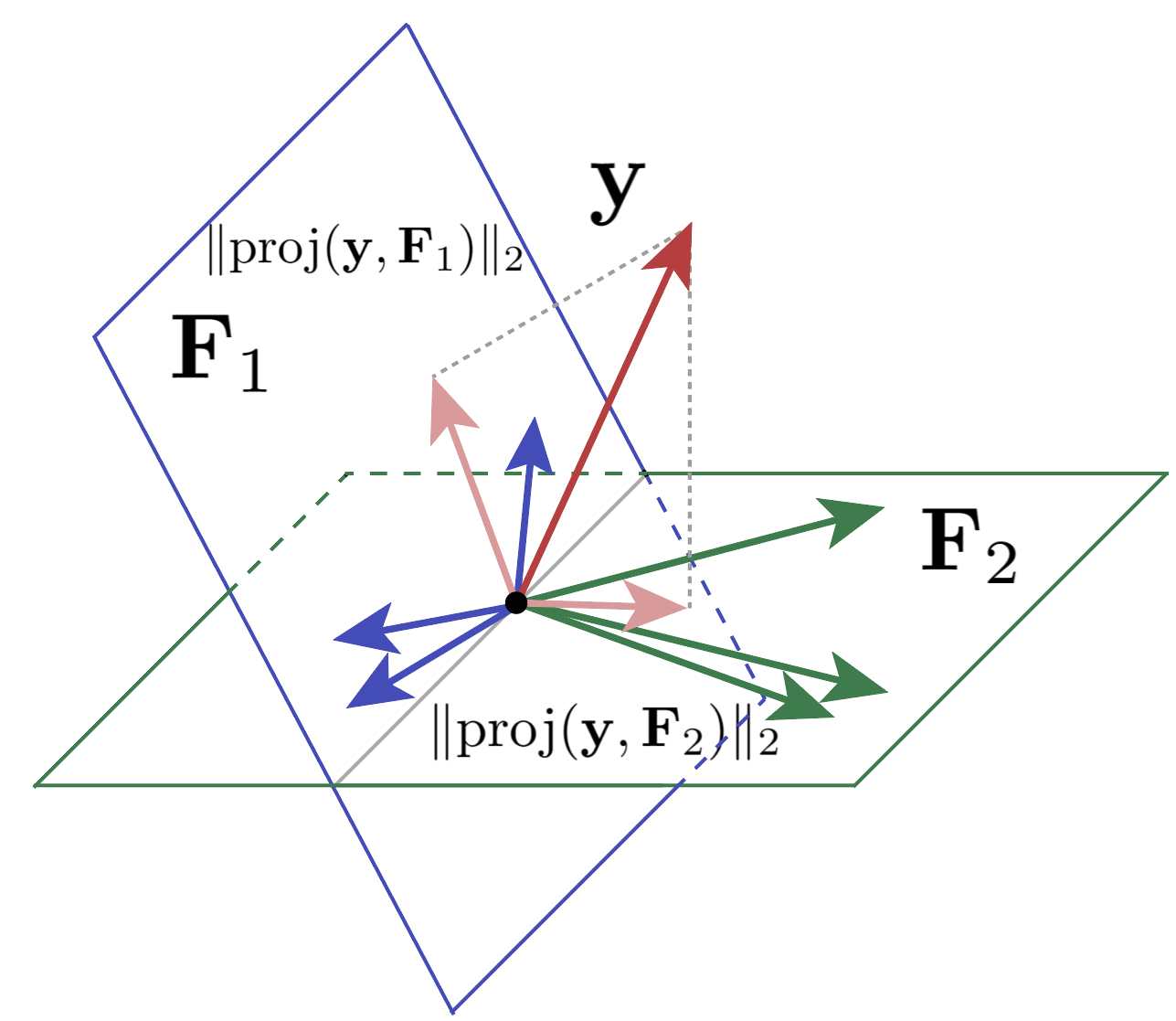}
	\end{tabular}
	\caption{Two ways of measuring correlations between a vector $\by$ and subspaces spanned by $\bF_1$ and $\bF_2$. (a) Accumulated inner product. This criterion is used in BOMP~\cite{li2018new}, BCoSaMP~\cite{zhang2019recovery}, and BSP~\cite{kamali2013block}. (b) Projection to subspaces. This criterion is used in BOMPR~\cite{fu2014block} and GPSP~\cite{he2023group}.}\label{fig_illustrate_inner_proj}
\end{figure}

Theoretically, we establish the convergence and robustness of GPSP. Firstly, to validate GPSP as an effective greedy algorithm for block-sparse signal reconstruction,  we show the exact recovery property of GPSP when the measurements are accurate and the sampling matrix satisfies BRIP with sufficiently small BRIC. See Theorem~\ref{thm_main}.  Secondly, to justify the robustness of GPSP, we derive an error bound on the estimated coefficients for inaccurate measurements and prove that signals with large magnitudes can be recovered even if the data is slightly perturbed. See Theorem~\ref{thm::perturb}.  Due to the different criteria for feature selection and removal, our proofs follow similar strategies as in~\cite{dai2009subspace,kamali2013block} but with new techniques for estimating changes of recovered signals during the iterations. Such analogous conditions developed for other greedy algorithms are sufficient but not necessary for exact recovery, and quantifying the BRIC for general matrices is NP-hard. In practice, it is more illuminating to compare different schemes in their algorithmic frameworks and via numerical experiments.

Algorithmically, we note that GPSP share a similar structure with many state-of-the-art greedy algorithms, including BOMP~\cite{li2018new}, BOMPR~\cite{fu2014block}, BCoSaMP~\cite{zhang2019recovery}, and BSP~\cite{kamali2013block}. They are all iterative schemes, and during each of their respective iteration, they involve strategically expanding the pool of candidate features. For some of these algorithms, the expanding stage is followed by a shrinking stage. The expanding stage aims at including the correct features while the shrinking stage refines the selection by removing less important ones. These algorithms take different criteria for feature inclusion and exclusion. We systematically investigate such discrepancies and classify them according to the ways of quantifying correlations between residuals and matrix blocks. We also provide systematic ablation studies to compare different combinations of criteria.

We  present various numerical experiments to justify the superior performance of GPSP over existing methods for block sparse regression under general settings. To our best knowledge, GPSP was shown to outperform BSP in~\cite{he2023group} for identifying PDEs with varying coefficients, and many of the existing algorithms were compared only in standard settings where each block of the sampling matrix follows an identical Gaussian distribution. In this work, we compare GPSP with state-of-the-art greedy algorithms under diverse conditions. In particular, we consider reconstructions with sampling matrices consisting of heterogeneous Gaussian blocks, where different blocks may have distinct means and variances.   To examine the robustness of these algorithms against data perturbations, we conduct a series of experiments with different noise levels as well as different models of noise. Moreover, we extend the comparisons to the contexts of face recognition and PDE identification. These experiments confirm that GPSP is consistent, robust, and accurate in diverse applications with complicated block features.

To sum up, our contributions in this work include
\begin{itemize}
	\item Convergence analysis of Group Projected Subspace Pursuit (GPSP)~~\cite{he2023group} that showed superior performances in identifying PDEs with varying coefficients. We characterize a sufficient condition for exact recovery of GPSP, and we prove an error bound for reconstruction when the observational data is perturbed. 
	\item Systematic comparison of the algorithmic structures between GPSP and other state-of-the-art greedy algorithms. Our study focuses on comparing different criteria involved in expanding and shrinking stages. We claim that GPSP outperforms the others particularly when the sampling matrix has heterogeneous variances.
	\item Comprehensive comparison between GPSP and other state-of-the-art greedy algorithms for block-sparse signal reconstruction in diverse settings. We compare their performances when feature blocks are heterogeneous and when the noise has different distributions. Moreover, we compare these methods in PDE identification and image processing to further validate the effectiveness of GPSP.  
\end{itemize}

This paper is structured as follows: In Section \ref{sec.problem}, we formally state the block sparse regression problem while fixing notations used in this paper. We then review some known properties about BRIP and BRIC closely related to our analysis. In Section \ref{sec.algorithm}, after recalling the GPSP algorithm, we present our main theoretical results. In particular, in Section~\ref{sec:main1}, we show the exact recovery property of GPSP (Theorem~\ref{thm_main}), and in Section~\ref{sec:main2}, we characterize the reconstruction distortion due to data perturbation (Theorem~\ref{thm::perturb}). We discuss the outlines of proofs for these results in Section~\ref{sec::proof1} and Section~\ref{sec::proof2}. In Section~\ref{sec::comparison}, we systematically compare the algorithmic differences between GPSP and other greedy algorithms by focusing on the criteria for feature selection and removal. The effectiveness of GPSP is demonstrated through systematic experiments in Section \ref{sec::numerical}. We conclude this paper in Section \ref{sec.conclusion} and collect the proof details in the Appendix.

\section{Preliminary}
\label{sec.problem}
In this section, we introduce notations and related concepts used in this paper. We also  review some existing results for characterizing sufficient conditions for the exact recovery of block-sparse signal reconstruction.
\subsection{Notations and definitions}
Throughout this paper, we use bold small letters to denote vectors, e.g. $\bv$ and $\bc$, and bold capital letters to denote matrices, such as $\bA$ and $\bB$. For a scalar $a\in\mathbb{R}$, we use $\lfloor a \rfloor$ to denote the largest integer that is no greater than $a$. For a vector $\bv=[v_1,\dots,v_N]^{\top}\in\mathbb{R}^N$, we denote the set of indices of its non-zero elements as $\text{supp}\,\bv=\{i~|~v_i\neq 0, i=1,2,\dots,N\}$.
\begin{definition}[$k$-sparse vector] For any integer $N\geq 2$, a vector $\bv\in\mathbb{R}^N$ is $k$-sparse for $k\leq N$, if the number of non-zero elements in $\bv$, denoted by $|\text{supp}\,\bv|$, does not exceed $k$.
\end{definition}
By this definition, a $k$-sparse vector can have fewer non-zero elements than $k$. 
In this paper, we consider block matrices with even block sizes.  Let $G,M$ be positive integers. A $G$-block matrix is defined by $\bA=[\bF_1,\dots,\bF_G]\in\mathbb{R}^{N\times GM}$, where $\bF_g\in \mathbb{R}^{N\times M}$ for $g=1,2,\dots,G$, and $M$ denoting the number of columns in each block. For a nonempty subset of indices $T\subset\{1,2,\dots, G\}$ with size $|T|$, we denote $\bA_T\in\mathbb{R}^{N\times |T|M}$  as the submatrix of $\bA$ obtained by concatenating matrices from $\{\bF_g\in\mathbb{R}^{N\times M}~|~g\in T\}$ horizontally in an increasing order of the block indices. For any  $\bv\in\mathbb{R}^{GM}$, the group structure of $\bA$ induces a natural partition such that $\bv=[\bv_1^{\top},\dots, \bv_G^{\top}]^{\top}$ with $\bv_g\in\mathbb{R}^{M}$ for $g=1,2,\dots,G$. Similarly, we use $\bv_T\in\mathbb{R}^{|T|M}$ to denote the subvector of $\bv$ formed by stacking vectors in $\{\bv_g
\in\mathbb{R}^M~|~g\in T\}$ together in an increasing order of the block indices. We  use $\gsup \bv=\{g~|~\|\bv_g\|_1 \neq 0, g=1,2,\dots,G\}$ to denote the set of indices of blocks with non-zero entries.
\begin{definition}[Block $k$-sparse vector] For any positive integers $G,M$ such that $GM\geq 2$, a vector $\bv=[\bv^\top_1,\bv^\top_2,\dots, \bv_G^\top]^\top\in\mathbb{R}^{GM}$  with the given block pattern is block $k$-sparse for $k\leq G$, if the number of blocks with non-zero entries, denoted by $|\gsup\,\bv|$, does not exceed $k$.
\end{definition}
Analogously, a block $k$-sparse vector can have fewer non-zero blocks than $k$. For any vector $\bv$ and matrix $\bA$, we introduce notations
\begin{align}
	\text{proj}(\bv,\bA):=\bA\bA^\dagger\mathbf{v}=\bA(\bA^\top\bA)^{-1}\bA^\top\mathbf{v}\;,~\resid(\bv, \bA):= \bv-\text{proj}(\bv,\bA),
\end{align}
for frequently used operators, where $^\dagger$ is the Moore-Penrose pseudo-inverse.

\subsection{Review of block restricted isometry property}
In this section, we review some existing concepts used in our theoretical analysis.  There are many notions developed for characterizing sufficient conditions for the exact recovery of sparse signals, such as Block Coherence (BC)~\cite{eldar2010block}, Mutual Subspace Incoherence (MSI)~\cite{ganesh2009separation}, and Cumulative Subspace Coherence (CSC)~\cite{elhamifar2012block}. In this work, we develop our analysis based on the Block Restricted Isometry Property (BRIP)~\cite{eldar2009robust} which generalizes the well-known Restricted Isometry Property (RIP)~\cite{candes2005decoding}.

\begin{definition}[Block Restricted Isometry Property (BRIP)]  A matrix $\bA\in\mathbb{R}^{N\times GM}$ has the BRIP with parameters $(k,\delta_M)$ where $\delta_{M}\in(0,1)$, if for any $T\subseteq\{1,2,\dots, G\}$ with $|T|\leq k$, 
	\begin{align}
		(1-\delta_M)\|\bx_T\|_2^2\leq \|\bA_T\bx_T\|_2^2\leq (1+\delta_M)\|\bx_T\|_2^2
		\label{eq_BRIP_condition}
	\end{align}
	holds for all $\bx_T\in\mathbb{R}^{M|T|}$. The block Restricted Isometry Constant (BRIC) $\delta_{M,k}$ is the infimum of all parameters $\delta_{M}$ for which~\eqref{eq_BRIP_condition} holds.
\end{definition}
 
We review several properties of BRIP and BRIC that are important for subsequent analysis, and their proofs can be found in~\cite{eldar2009robust,eldar2010block,kamali2013block}. Intuitively, submatrices of a block matrix $\bA$ with small BRIC are close to  orthonormal matrices. This is made precise in the following lemma, which directly follows from the definition. 

\begin{lemma}[Equations (69) and (70) in \cite{eldar2009robust}]\label{prop4}
	Let $\bA\in\mathbb{R}^{N\times GM}$ be a block matrix, and $\mathcal{T}_k$ be the collection of index sets with $k$ elements. For any positive integer $k\leq G$, define 
	$$\overline{\sigma}_k = \max_{T\in\mathcal{T}_k}\{\sigma_{\max}(\bA_{T})\},\ \underline{\sigma}_k = \min_{T\in\mathcal{T}_k}\{\sigma_{\min}(\bA_{T})\},$$
	where $\sigma_{\max}(\cdot)$ and $\sigma_{\min}(\cdot)$ denote the maximal and minimal singular values, respectively.  Then
	$$1-\delta_{M,k}\leq \underline{\sigma}^2_k \leq \overline{\sigma}^2_k \leq 1+\delta_{M,k}.$$
	
\end{lemma}

There is an ordering property associated with  BRIC that allows simple comparisons. This property becomes especially useful when simplifying expressions involving BRIC with mixed indices.
\begin{lemma}[Lemma 1 in~\cite{kamali2013block}]\label{prop5}
	For a fixed $M$ and any $1\leq k_1\leq k_2$ we have
	$$\delta_{M,k_1}\leq \delta_{M,k_2}.$$
\end{lemma}

Using Lemma~\ref{prop4}, BRIC can  be used to bound the norms of various important linear operators. In particular, for a block matrix $\bA$, the BRIC is associated with the correlation between disjoint blocks measured by inner products.
\begin{lemma}[Lemma 2 in~\cite{kamali2013block}]\label{lemma_inner}
	Let $I,J\subseteq\{1,2,\dots,G\}$ with $I\cap J=\varnothing$. Suppose $\bA$ satisfies BRIP with constant $\delta_{M,|I|+|J|}\leq 1$, then for any  $\bc\in\mathbb{R}^{M|J|}$, we have 
	\begin{align}
		\|\bA^\top_I\bA_J\bc\|_2\leq \delta_{M,|I|+|J|}\|\bc\|_2.\label{eq_lemma_inner}
	\end{align}
\end{lemma}
From Lemma~\ref{lemma_inner}, we can also deduce the following bounds for the subspace projection and the norm of the corresponding residual.
\begin{lemma}[Lemma 3 in~\cite{kamali2013block}]\label{lemma_proj}
	Let $I,J\subseteq\{1,2,\dots,G\}$ be two disjoint sets and $\bA$ satisfy BRIP with $\delta_{M,|I|+|J|}\leq 1$. Then for any $\by\in\text{span}(\bA_I)$, we have
	\begin{align}
		\|\text{proj}(\by,\bA_J)\|_2\leq \frac{\delta_{M,|I|+|J|}}{1-\delta_{M,\max\{|I|,|J|\}}}\|\by\|_2\label{eq_lemma_proj}
	\end{align}
	and 
	\begin{align}
		\left(1-\frac{\delta_{M,|I|+|J|}}{1-\delta_{M,\max\{|I|,|J|\}}}\right)\|\by\|_2\leq \|\text{resid}(\by,\bA_J)\|_2.
		\label{eq_lemma_resid}
	\end{align}
\end{lemma}
In the following sections, we  employ BRIP to prove a sufficient condition for the convergence of GPSP. Specifically, we  show that if a block matrix has a sufficiently small BRIC, then GPSP is guaranteed to exactly recover the underlying true block sparse signal. If the data is contaminated by noise, the recovery distortion is bounded by the perturbation with a scaling factor involving  BRIC.
\section{GPSP and Main Results}
\label{sec.algorithm}
The Group Projected Subspace Pursuit (GPSP) algorithm was first proposed in \cite{he2023group} and used in \cite{he2023much} for identifying PDEs with varying coefficients.
As its algorithmic design is independent from the particular form associated with the PDE identification, GPSP can also be used for general block-sparse signal reconstruction. Algorithm~\ref{alg_GPSP} presents the details for GPSP\footnote{In~\cite{he2023group}, the metric for initialization and expansion was defined by $|\proj(\bv,\bF_g)|/(\|\bv\|_2\|\proj(\bv,\bF_g)\|_2)$ for $\bv=\by$ or $\by_r^{l-1}$, which equals $\|\proj(\bv,\bF_g)\|_2/\|\bv\|_2$. In this paper, we present this equivalent form for convenience.}. It shares similar structures with state-of-the-art greedy algorithms, such as BSP~\cite{kamali2013block}, and we compare  them in Section~\ref{sec::comparison}. 

In this section, our focus is on proving the convergence and robustness of GPSP, thus theoretically justifying its validity in recovering block-sparse signals for general purposes. In Section~\ref{sec:main1}, we present our main result (Theorem~\ref{thm_main}) stating a sufficient condition for the exact recovery property of GPSP. In Section~\ref{sec:main2}, we analyze the behaviors of GPSP when the observation vector is perturbed by noise. We outline the major steps for proving Theorem~\ref{thm_main} and Theorem~\ref{thm::perturb} in Section~\ref{sec::proof1} and Section~\ref{sec::proof2}, respectively,  and delegate the details in the Appendix. We note that the existing proof for BSP in \cite{kamali2013block} cannot be trivially applied to analyze GPSP as different criteria are used for selecting and removing features; new techniques need to be developed.
In the following, we denote the set of group indices selected by GPSP as $\widehat{T}$. The corresponding reconstruction $\widehat{\bc}$ is obtained by setting entries associated with blocks indexed by $\widehat{T}$ via
\begin{align}
	\widehat{\bc}_{\widehat{T}}=\bA_{\widehat{T}}^\dagger\by
\end{align}
and the other elements as zeros. 
\begin{algorithm}[t!]
	\begin{algorithmic}[1]
		\Require{Feature system $(\mathbf{A}$, $\by)$, specified level of group sparsity $k\geq 1$, maximal number of iterations $\text{Iter}_{\max}\geq 1$.
		}
		
		\State Set $l=0$.
		\State Set 
		$T^l = \{k$ largest indices of 
		$\|\proj(\by,\bF_g)\|_2, g=1,2,\dots,G\}$. \Comment{Initialization}
		\State Set $\by_r^l = \resid(\by, \bA_{T^l})$, $\bA_{T^l}$ 
		concatenates $\{\bF_g\}_{g\in T^l}$ horizontally.
		
		\For{$l=1,\dots,\text{Iter}_{\max}$}
		\State$\wT^l = T^{l-1}\cup\{
		k$ largest indices of $\|\proj(\by_r^{l-1},\bF_g)\|_2, g=1,2,\dots,G\}$.\Comment{Expanding}
		\State  Compute $\mathbf{x}^l_p = \bA_{\wT^l}^\dagger\by$.
		\State Set  $T^{l}=\{
		k$ largest indices of 
		$\|\bF_g\bx^l_p[g]\|_2,~g\in \wT^l\}$, where $\bx^l_p[g]$ is the subvector of $\bx^l_p$ corresponding to the $g$-th group.\Comment{Shrinking}
		\State Compute $\by_r^l = \resid(\by, \bA_{T^l})$.
		\If{$\|\by_r^l\|_2 \geq \|\by_r^{l-1}\|_2$}
		\State Set $T^l = T^{l-1}$ and terminate.
		\EndIf
		\EndFor
		
		\Return The optimal group indices $T^l$ and the estimated coefficient $\bA_{T^l}^\dagger\by$.
	\end{algorithmic}
	\caption{Group Projected Subspace Pursuit (GPSP) \cite{he2023group}} 
	\label{alg_GPSP}
\end{algorithm}

\begin{figure}[t!]
	\centering
	\begin{tabular}{c@{\hspace{2pt}}c@{\hspace{2pt}}c}
		(a)&(b)&(c)\\
		\includegraphics[width=0.33\textwidth]{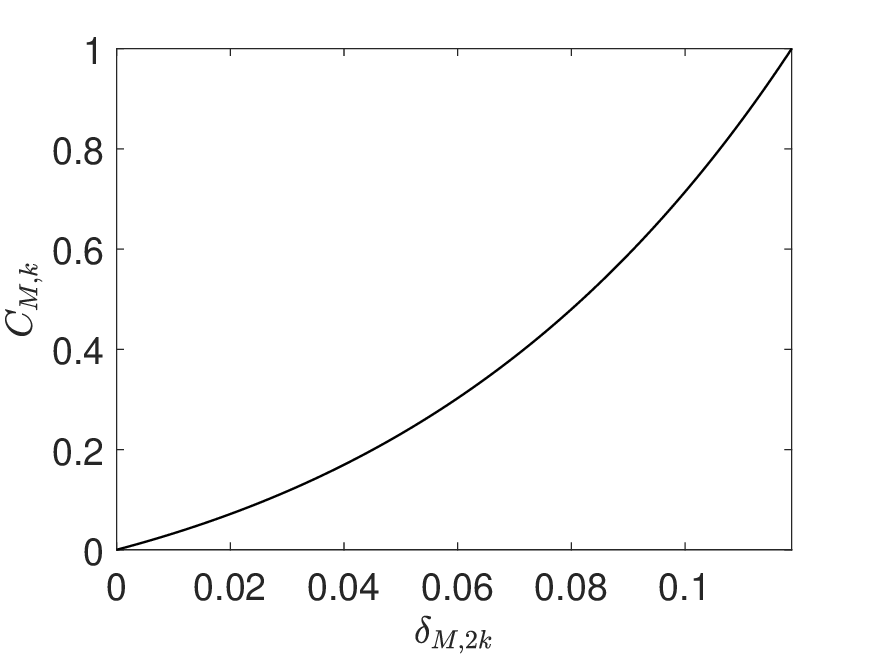}&
		\includegraphics[width=0.33\textwidth]{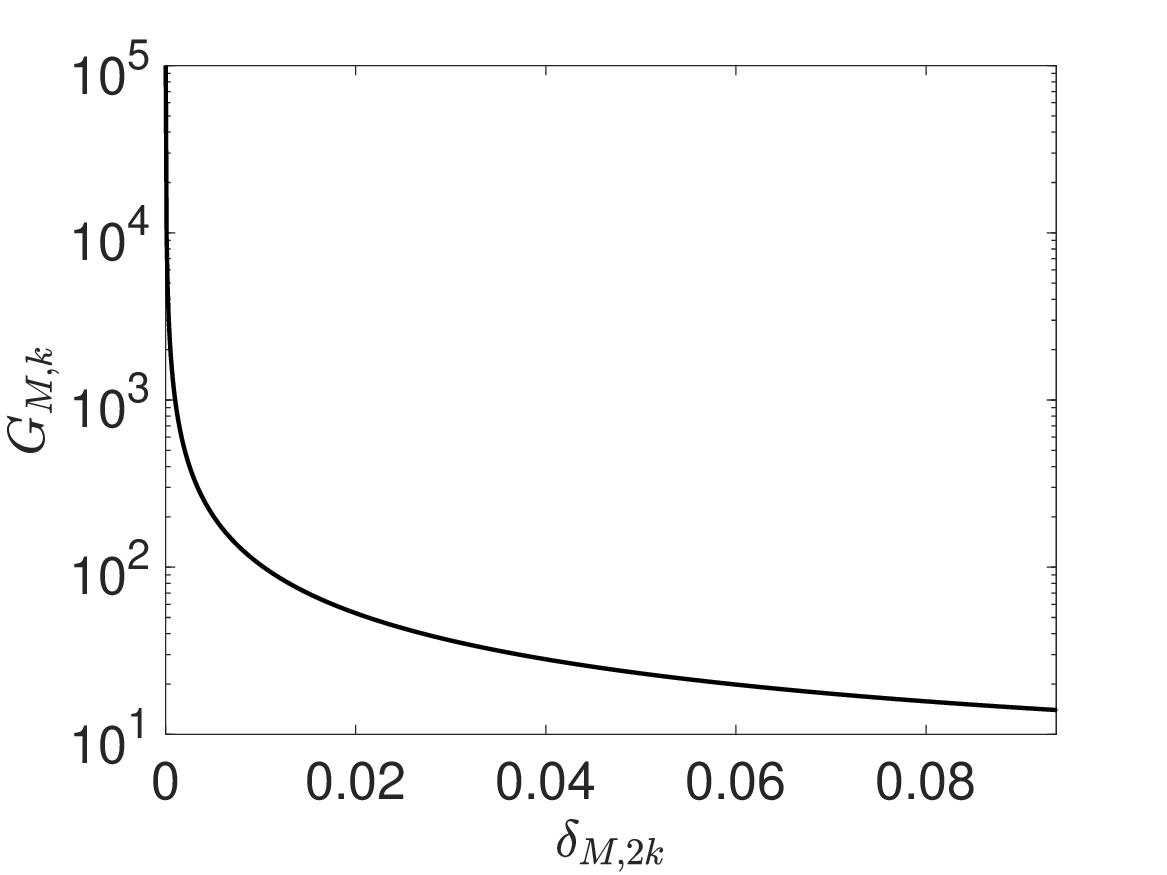}&
		\includegraphics[width=0.33\textwidth]{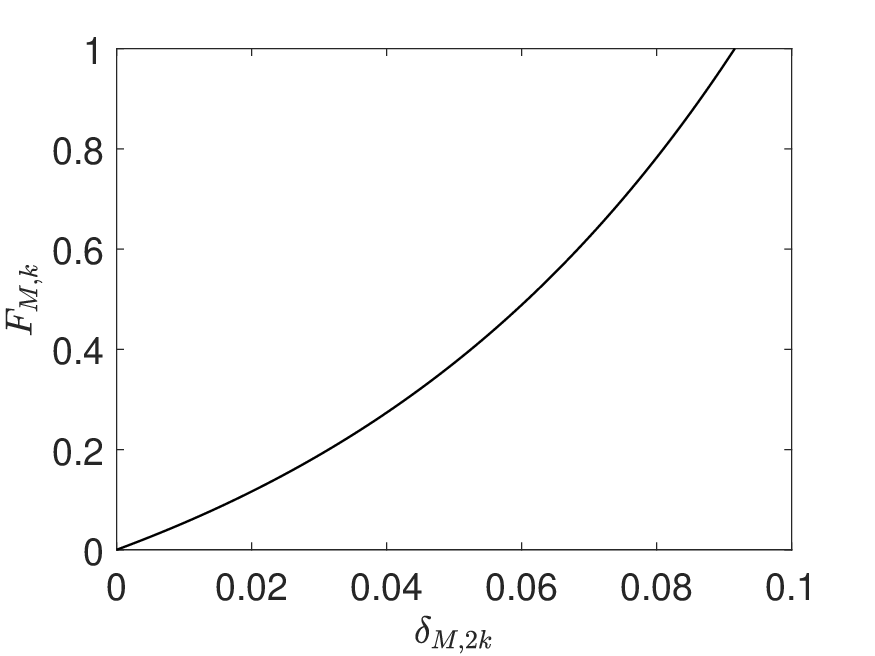}
	\end{tabular}
	\caption{(a) Plot of $C_{M,k}$~\eqref{eq::main1} for different values of BRIC $\delta_{M,2k}$. A sufficient condition for GPSP to converge is that $C_{M,k}<1$. (b) Plot of $G_{M,k}$~\eqref{eq::main_perturb}, which is the scaling factor for controlling the recovery distortion by the perturbation. (c) Plot of $F_{M,k}$~\eqref{eq::main2} for different values of BRIC $\delta_{M,2k}$. When $F_{M,k}<1$, GPSP shrinks the residual if the perturbation is sufficiently small. See Theorem~\ref{thm.inaccurate2}.  }\label{fig::condition}
\end{figure}

\subsection{Convergence of GPSP}\label{sec:main1}
Our main result  gives a sufficient condition for the exact recovery property of GPSP. 
\begin{theorem}\label{thm_main}
Suppose $\bc^*\in\mathbb{R}^{GM}$ is a block $k$-sparse vector, i.e., $\|\bc^*\|_{1,0}\leq k\leq G$.  Let $\by=\bA\bc^*$ where $\bA\in\mathbb{R}^{N\times GM}$. Then GPSP  (Algorithm~\ref{alg_GPSP}) with the data $\bA$ and $\by$ converges to the exact solution $\bc^*$ if the sampling matrix $\bA$ satisfies the BRIP~\eqref{eq_BRIP_condition} with BRIC $\delta_{M,2k}$ and
\begin{align}
	C_{M,k}:=\frac{\delta_{M,2k}(1+\delta_{M,2k})(3-\delta_{M,2k})(-\delta_{M,2k}^2+2\delta_{M,2k}+1)}{(1-\delta_{M,2k})^6}<1.
	\label{eq::main1}
\end{align}
\end{theorem}
Theorem \ref{thm_main} is proved in Section~\ref{sec::proof1}. 
Theorem \ref{thm_main} provides a sufficient condition for the convergence of GPSP in terms of the BRIC $\delta_{M,2k}$ of the block matrix $\bA$. In particular, for block matrices with sufficiently small BRIC $\delta_{M,2k}$, GPSP is guaranteed to converge. Figure~\ref{fig::condition} (a) plots the graph of $C_{M,k}$ in~\eqref{eq::main1} against $\delta_{M,2k}$; it increases as $\delta_{M,2k}$ gets larger before it reaches the value of $1$. Using the bisection method, we deduce that GPSP converges when  $\delta_{M,2k}$ is  bounded by approximately $0.1188$. To guarantee an exact recovery of a block $k$-sparse signal, we note that the BRIC in~\eqref{eq::main1} is related to the block sparsity level $2k$. This is consistent with the analysis of the uniqueness of block sparse representations. See Appendix~\ref{sec::unique_sparse} for details. We also remark that a similar sufficient condition is available for BSP~\cite{kamali2013block}, which is guaranteed to converge when $\delta_{M,3k}<0.1672$. Although the upper-bound constant is higher than $0.1188$, their condition is associated with the BRIC of order $3k$, and the relation $\delta_{M,2k}\leq \delta_{M,3k}$ makes it difficult to judge which condition is more relaxed. Instead, we shall compare GPSP with BSP and other related algorithms in Section~\ref{sec::comparison} and Section~\ref{sec::numerical} via various examples and experiments.

\subsection{Recovery distortion with perturbed measurements}\label{sec:main2}
When the measurements contain an additive perturbation, we have the following characterization of the recovery distortion in terms of the BRIC of the sampling matrix and the magnitude of the perturbation. 
\begin{theorem}\label{thm::perturb}
Suppose $\bc^*\in\mathbb{R}^{GM}$ is a block $k$-sparse vector, i.e., $\|\bc^*\|_{1,0}\leq k\leq G$.  Let $\by=\bA\bc^*+\be$ where $\bA\in\mathbb{R}^{N\times GM}$ satisfies the BRIP with BRIC $\delta_{M,2k}$, and $\be\in\mathbb{R}^{GM}$ is a perturbation. Then
\begin{align}
	\|\bc^*-\widehat{\bc}\|_2< G_{M,k}\|\be\|_2\label{eq::main_perturb}
\end{align}
with $G_{M,k}=\frac{1+2\delta_{M,2k}}{\delta_{M,2k}(1-\delta_{M,2k})}$ whenever $\delta_{M,2k}<0.0916$. 
\end{theorem}
Theorem \ref{thm::perturb} is proved in Section \ref{sec::proof2}. 
Theorem \ref{thm::perturb} is a generalization of the result in~\cite{dai2009subspace} (Theorem 9) for GPSP. It states that the distortion of the recovery is bounded by the magnitude of the observational perturbation multiplied by a factor. As shown in Figure~\ref{fig::condition} (b), the factor $G_{M,k}$ in~\eqref{eq::main_perturb} decreases as $\delta_{M,2k}$ increases.   When $\delta_{M,2k}$ approaches the bound $0.0916$, $G_{M,k}$ is approximately $13.9825$; when $\delta_{M,2k}$ approaches $0$ from right, $G_{M,k}$ grows to infinity. This indicates that mild correlations among blocks can be helpful for handling noise. Theorem~\ref{thm::perturb} justifies the robustness of GPSP against noise for sampling matrices with sufficiently small BRIC, and the reconstruction error reduces to $0$ when the perturbation vanishes. In particular, the relative $L_2$-error $\|\bc^*-\widehat{\bc}\|_2/\|\bc^*\|_2<\alpha\%$ whenever  $\|\be\|_2<13.9825^{-1}\times \alpha\%\|\bc^*\|_2 =0.0715\times \alpha\%\|\bc^*\|_2$; hence, GPSP has a higher tolerance for the noise if the true signal $\bc^*$ has  a  larger magnitude than that of the noise.

\subsection{Outline of the proof for Theorem~\ref{thm_main}}\label{sec::proof1}

Our strategy of  proving the convergence of GPSP is inspired by~\cite{dai2009subspace,kamali2013block}.  As stated in Algorithm \ref{alg_GPSP}, each iteration of GPSP contains two stages: a dictionary expanding stage (line 5) and a dictionary shrinking stage (line 7). Hence, there are three key criteria towards the establishment of Theorem~\ref{thm_main}: (i) correct features need to be included during the expansion (Theorem~\ref{thm1}); (ii) correct features need to survive during the shrinking (Theorem~\ref{thm2}); and (iii) the fitting errors need to reduce (Theorem~\ref{thm3}).  

We start by studying the change of the magnitude of the blocks of $\bc^*$ that are not selected by GPSP during the expansion (line 5 of Algorithm~\ref{alg_GPSP}).
\begin{theorem}\label{thm1}
Suppose $\bA\in \mathbb{R}^{N\times GM}$ satisfies the BRIP condition with BRIC $\delta_{M,2k}$, and the exact solution $\bc^*$ is block $k$-sparse. During the $l$-th iteration of GPSP, it holds that
\begin{align}
	\|\bc^*_{T^*-\widetilde{T}^l}\|_2\leq \beta_k\|\bc^*_{T^*-T^{l-1}}\|_2,
\end{align}
where
\begin{align}
	\beta_{k} = \frac{\delta_{M,2k}(1-\delta_{M,k}+\delta_{M,2k})}{(1-\delta_{M,k})^2}\left(\sqrt{\frac{2(1+\delta_{M,1})}{1-\delta_{M,1}}}+1\right).
\end{align}
\end{theorem}
Theorem \ref{thm1} is proved in Appendix~\ref{app:thm1}.
This result says that during the expansion, the coefficient magnitude of the unselected true blocks at the current iteration is controlled by that of the true blocks missed from the previous iteration. Next, we  study the shrinking stage.
\begin{theorem}\label{thm2} Under the conditions of Theorem~\ref{thm1} and during the $l$-th iteration of GPSP, it holds that
\begin{align}
	\|\bc^*_{T^*-T^l}\|_2\leq\mu_k\|\bc^*_{T^*-\wT^l}\|_2,
\end{align}
where
\begin{align}
	\mu_k=1+ \frac{\sqrt{2}\delta_{M,2k}}{1-\delta_{M, 2k}}\left(1+\sqrt{\frac{1+\delta_{M,1}}{1-\delta_{M,1}}}\right).
\end{align}
\end{theorem}
Theorem \ref{thm2} is proved in Appendix~\ref{app:thm2}.
This result says that during the shrinking, the coefficient magnitude of the unselected true blocks at iteration $l$ is controlled by that of the true blocks missed by the expanded selection $\wT^l$. 
Based on Theorem \ref{thm1} and \ref{thm2}, our next result compares the sizes of the residuals between consecutive iterations.

\begin{theorem}\label{thm3}
Under the conditions of Theorem~\ref{thm1} and during the $l$-th iteration of GPSP, it holds that
\begin{align}
	\|\by_r^l\|_2\leq \rho_k\|\by_r^{l-1}\|_2,
\end{align}
where
\begin{align}
	\rho_k = \frac{\mu_k\beta_k\sqrt{1-\delta^2_{M,k}} }{1-\delta_{M,k}-\delta_{M,2k}}
\end{align}
and the constant $\beta_k,\mu_k$ are defined in Theorem \ref{thm1} and \ref{thm2}, respectively.
\end{theorem}

Theorem \ref{thm3} is proved in Appendix~\ref{app:thm3}.
Theorem \ref{thm3} states that by running GPSP (Algorithm \ref{alg_GPSP}), the $\ell_2$ norm of the residual is controlled by that of the residual in the previous iteration. Notice that if $\rho_k<1$, the size of the residual decays in each iteration of GPSP. 
Hence, a sufficient condition for the convergence of GPSP is  $\rho_k<1$. Using Lemma~\ref{prop5}, we have 
\begin{align*}
\beta_{k} &= \frac{\delta_{M,2k}(1-\delta_{M,k}+\delta_{M,2k})}{(1-\delta_{M,k})^2}\left(\sqrt{\frac{2(1+\delta_{M,1})}{1-\delta_{M,1}}}+1\right)\leq\frac{\delta_{M,2k}(1+\delta_{M,2k})}{(1-\delta_{M,2k})^2}\left(\sqrt{\frac{2(1+\delta_{M,2k})}{1-\delta_{M,2k}}}+1\right)\\
&=\frac{\delta_{M,2k}(1+\delta_{M,2k})}{(1-\delta_{M,2k})^2}\times\frac{\sqrt{2(1-\delta^2_{M,2k})}+1-\delta_{M,2k}}{1-\delta_{M,2k}}\leq \frac{\delta_{M,2k}(1+\delta_{M,2k})}{(1-\delta_{M,2k})^2}\times\frac{3-\delta_{M,2k}}{1-\delta_{M,2k}}\\
=&\frac{\delta_{M,2k}(1+\delta_{M,2k})(3-\delta_{M,2k})}{(1-\delta_{M,2k})^3}
\end{align*}
and
\begin{align*}
\mu_k&=1+ \frac{\sqrt{2}\delta_{M,2k}}{1-\delta_{M, 2k}}\left(1+\sqrt{\frac{1+\delta_{M,1}}{1-\delta_{M,1}}}\right)\leq 1+ \frac{\sqrt{2}\delta_{M,2k}}{1-\delta_{M, 2k}}\left(1+\sqrt{\frac{1+\delta_{M,2k}}{1-\delta_{M,2k}}}\right)\\
&= 1+ \frac{\sqrt{2}\delta_{M,2k}}{1-\delta_{M, 2k}}\times\frac{1-\delta_{M,2k}+\sqrt{1-\delta^2_{M,2k}}}{1-\delta_{M,2k}}\leq\frac{(1-\delta_{M,2k})^2+2\delta_{M,2k}(2-\delta_{M,2k})}{(1-\delta_{M,2k})^2}\\
&=\frac{-\delta_{M,2k}^2+2\delta_{M,2k}+1}{(1-\delta_{M,2k})^2}.
\end{align*}
Therefore, we arrive at
$$\rho_k
\leq \frac{\delta_{M,2k}(1+\delta_{M,2k})(3-\delta_{M,2k})(-\delta_{M,2k}^2+2\delta_{M,2k}+1)}{(1-\delta_{M,2k})^6}$$
and the right-hand side of the inequality above is exactly $C_{M,k}$ defined in~\eqref{eq::main1}, which proves  Theorem~\ref{thm_main}.

\subsection{Outline of the proof for Theorem~\ref{thm::perturb}}\label{sec::proof2}

The proof for Theorem~\ref{thm::perturb} is analogous to \cite[Proof for Theorem 2]{dai2009subspace}. With the measurement perturbation, we shall first bound the recovery distortion as follows.
\begin{theorem}\label{thm.inaccurate}
Suppose $\bc^*\in\mathbb{R}^{GM}$ is a block $k$-sparse vector, i.e., $\|\bc^*\|_{1,0}\leq k\leq G$, whose block support is $T$.  Let $\by=\bA\bc+\be$ where $\bA\in\mathbb{R}^{N\times GM}$ satisfies the BRIP with BRIC $\delta_{M,2k}$ and $\be\in\mathbb{R}^{GM}$ is a perturbation. Then
\begin{align}
	\|\bc^*-\widehat{\bc}\|_2\leq \frac{1+\delta_{M,2k}-\delta_{M,k}}{1-\delta_{M,k}}\|\bc^*_{T^*-\widehat{T}}\|_2 + \frac{1}{\sqrt{1-\delta_{M,k}}}\|\be\|_2.
\end{align}
\end{theorem}

Theorem \ref{thm.inaccurate} is proved in Appendix~\ref{app:inacc}.
The error bound in Theorem~\ref{thm.inaccurate} consists of two components: the magnitude of the true coefficients associated with the features missing from $\widehat{T}$, and the magnitude of the perturbation. In particular, when the features are correctly chosen, i.e., $\widehat{T}=T^*$, the recovery distortion reduces to $0$ as the perturbation vanishes.

Next, we prove that during the iteration of GPSP, the magnitude of the true coefficients associated with the features missing from the $l$-th iteration can be controlled by that from the previous iteration plus the magnitude of the perturbation.

\begin{theorem}\label{thm.inaccurate2}
Under the conditions of Theorem~\ref{thm.inaccurate} and during the $l$-th iteration of GPSP, it holds that
\begin{align}
	\|\bc^*_{T^*-T^l}\|_2
	\leq D_{M,k}\|\bc^*_{T^*-T^{l-1}}\|_2+E_{M,k}\|\be\|_2,
	\label{eq::thm_inacc2}
\end{align}
where $D_{M,k}$ and $E_{M,k}$ are constants independent of  $l$. In particular
\begin{align}
	D_{M,k}&=\frac{6\sqrt{2}\delta^2_{M,2k}(1+\delta_{M,2k})(2-\delta_{M,2k})}{(1-\delta_{M,2k})^5},\\
	E_{M,k}&=\frac{12\delta_{M,2k}(1+\delta_{M,2k})^2+(2-\delta_{M,2k})(1-\delta_{M,2k})^3}{(1-\delta_{M,2k})^5}.
\end{align}
Furthermore, if  the BRIC $\delta_{M,2k}$ of $\bA$ satisfies 
\begin{align}
	F_{M,k}=\frac{\delta_{M,2k}}{1-\delta_{M,2k}}+\frac{\sqrt{1+\delta_{M,2k}}D_{M,k}+\delta_{M,2k}\left(\sqrt{1+\delta_{M,2k}}E_{M,k}+2\right)}{\sqrt{1-\delta_{M,2k}}}<1,
	\label{eq::main2}
\end{align}
then $\|\by_r^l\|_2<\|\by_r^{l-1}\|_2$ whenever $\|\be\|_2\leq \delta_{M,2k}\|\bc^*_{T^*-T^{l-1}}\|_2$ for the $l$-th iteration.
\end{theorem}

Theorem \ref{thm.inaccurate2} is proved in Appendix~\ref{sec::proof_inaccurate2}.
Theorem~\ref{thm.inaccurate2}  states a sufficient condition for the decay of the residuals during one iteration of GPSP. Under the requirement of Theorem~\ref{thm.inaccurate2}, for sufficiently small perturbations, GPSP always decreases the size of residuals during each iteration. In case of large perturbations,  GPSP can still keep decaying the residuals during each iteration if the true signal has a sufficiently large $\ell_2$ norm. Theorem~\ref{thm.inaccurate2}  characterizes sufficient conditions for the convergence of GPSP when the measurements are perturbed.

Figure~\ref{fig::condition} (c) shows the plot of $F_{M,k}$. To finish the proof for Theorem~\ref{thm::perturb}, using the bisection method, we find that when $\delta_{M,2k}<0.0916$, the inequality~\eqref{eq::main2} holds. We then note that if $\delta_{M,2k}<0.0916$, Theorem~\ref{thm.inaccurate2} implies that $\|\be\|_2> \delta_{M,2k}\|\bc^*_{T^*-T^{l-1}}\|_2$ when GPSP terminates, i.e., $\|\by_r^l\|_2\geq \|\by_r^{l-1}\|$. By Theorem~\ref{thm.inaccurate}, we deduce the relation~\eqref{eq::main_perturb}, thus proving Theorem~\ref{thm::perturb}.
\section{Relations to Other Greedy Algorithms}\label{sec::comparison}

\begin{table}[t!]
\centering
\begin{tabular}{c||c|c|c|}
	\hline
	{\bf Algorithm}&{\bf Initialization}&{\bf Expanding}&{\bf Shrinking}\\\hline
	BOMP~\cite{eldar2010block}   &$\varnothing$&$\arg\max_{g}\|\bF_g^\top\by_r^{l-1}\|_2$ & -  \\\hline
	BOMPR~\cite{fu2014block}   &$\varnothing$&$\arg\min_{g}\|\by_r^{l-1}-\proj(\by_r^{l-1},\bF_g)\|_2$ & -  \\\hline
	BCoSaMP~\cite{zhang2019recovery}&$\varnothing$&$\arg\text{top}^{2k}_g\|\bF_g^{\top}\by_r^{l-1}\|_2$&$\arg\text{top}^k_g\|\bx_p^l[g]\|_2$\\\hline
	BSP~\cite{kamali2013block}&$\arg\text{top}^k_g\|\bF_g^{\top}\by\|_2$&$\arg\text{top}^k_g\|\bF_g^{\top}\by_r^{l-1}\|_2$&$\arg\text{top}^k_g\|\bx_p^l[g]\|_2$\\\hline
	GPSP~\cite{he2023group}&$\arg\text{top}^k_g\|\proj(\by,\bF_g)\|_2$&$\arg\text{top}^k_g\|\proj(\by_r^{l-1},\bF_g)\|_2$&$\arg\text{top}^k_g\|\bF_g\bx_p^l[g]\|_2$\\\hline
\end{tabular}
\caption{Comparison among different greedy algorithms for block-sparse signal recovery in terms of their respective operations during initialization, expanding, and shrinking. These stages are parallel with the those marked in Algorithm~\ref{alg_GPSP}. Here $\bF_g$ is the $g$-th block, $\by_r^{l-1}$ is the residual in the previous iteration, and $\arg\text{top}^k_g$ means taking the indices with the top $k$ values. }
\label{tab.diff}
\end{table}
There are various greedy algorithms developed for recovering block sparse signals from given observations. While from different perspectives, most of these algorithms share a common pattern: after initialization, they proceed iteratively by expanding and shrinking the pool of candidate features.

In the initialization stage, some algorithms may start with none, e.g., BOMP~\cite{eldar2010block}, BOMPR~\cite{fu2014block}, and BCoSaMP~\cite{zhang2019recovery}, and others have initial guesses, e.g., BSP~\cite{kamali2013block} and GPSP~\cite{he2023group}. During the expansion stage, they select candidates according to different criteria. Algorithms like BOMP and BOMPR select one candidate at a time, whereas BCoSaMP, BSP, and GPSP select multiple candidates. During the shrinking stage, BOMP and BOMPR do not exclude chosen candidates; BCoSaMP, BSP, and GPSP refine their selections by removing chosen candidates according to certain rules. Table~\ref{tab.diff} summarizes these discrepancies among  BOMP, BOMPR, BCoSaMP, BSP, and GPSP for a systematic comparison.

Based on Table~\ref{tab.diff}, it is clear that there are two types of criteria in both expanding and shrinking stages of the aforementioned algorithms. During the expanding stage, the correlations between the residual vector with block features, i.e., matrices with multiple columns,  are measured by (i) inner product criterion (IPC), used by BOMP, BCoSaMP, and BSP, or (ii) subspace projection criterion (SPC), used by BOMPR and GPSP.   During the shrinking stage, the importance of a feature is measured via (i) regression coefficient criterion (RCC), used by BCoSaMP and BSP, or (ii) response magnitude criterion (RMC), used by GPSP. BOMP and BOMPR do not remove candidate features once they are selected.

In the following, we discuss these two sets of criteria with simple examples and illustrate their effects. We conduct comprehensive numerical comparisons among these algorithms in Section~\ref{sec::numerical}.

\subsection{Feature inclusion criteria: IPC vs. SPC}\label{subsec::inner_proj}
For all the greedy algorithms summarized in Table~\ref{tab.diff}, correlations between a vector $\bd\in\mathbb{R}^N$ ($\by$ during the initialization or $\by_r^{l-1}$ during each iteration) and individual blocks are measured during the respective expansion stage. 

For IPC, the correlation is measured by the $\ell_2$ norm of the inner product: $\|\bF_g^{\top}\bd\|_2$, where $\bF_g\in\mathbb{R}^{N\times M}$. Assume that the blocks are normalized so that each column has unit $\ell_2$-norm. Then $\|\bF_g^\top\bd\|_2/\|\bd\|_2$ is the $\ell_2$-norm of the vector of cosine similarities between $\bd$ and each column of $\bF_g$, i.e., 
\begin{align}
\|\bF_g^\top\bd\|_2=\sqrt{\sum_{m=1}^M\cos^2\theta_{g,m}}\times\|\bd\|_2\label{eq_inner_score}
\end{align}
where $\theta_{g,m}$ is the angle between $\bd$ and the $m$-th column of $\bF_g$.  For SPC, the correlation is measured by the $\ell_2$-norm of the projection of $\bd$ to the column space of each block
\begin{align}\|\proj(\bd,\bF_g)\|_2.\label{eq_project_score}
\end{align} 
Let $\mathcal{V}_g$ denote the column space of $\bF_g$, and suppose $\|\proj(\bd,\bF_g)\|_2>0$. It is clear that 
\begin{align*}
\proj(\bd,\bF_g)/\|\proj(\bd,\bF_g)\|_2=\argmax_{\bx\in \mathcal{V}_g, \|\bx\|_2=1}\left|\bx^\top\bd\right|.
\end{align*}
This implies that~\eqref{eq_inner_score} yields large values if all the columns in $\bF_g$ are located close to the projection of $\bd$ to $\mathcal{V}_g$. In case $\text{dim}(\mathcal{V}_g)=1$,  then the correlation based on~\eqref{eq_inner_score} and projection differ by a multiplicative factor $\sqrt{M}$. If $\text{dim}(\mathcal{V}_g)>1$, then $\mathcal{V}_g\cap \bd^{\perp}\neq \varnothing$, and we have
\begin{align}
\min_{\bx\in \mathcal{V}_g, \|\bx\|_2=1}\left|\bx^\top\bd\right|/\|\bd\|_2=0\;,\label{eq_range}
\end{align}
where $\bd^{\perp}$ is the orthogonal space of $\bd$.
Therefore, with the subspace $\mathcal{V}_g$ fixed and if its dimension is greater than $1$, we may rotate columns of $\bF_g$ inside $\mathcal{V}_g$ and deduce the range of possible values for~\eqref{eq_inner_score} as follows
\begin{align}
0<\|\bF_g^\top\bd\|_2< \sqrt{M}\|\proj(\bd,\bF_g)\|_2\;.\label{eq_relation}
\end{align}
In other words,  IPC depends on the configuration of columns of $\bF_g$ in $\mathcal{V}_g$ whenever $\text{dim}(V_g)>1$, whereas SPC is invariant to this variation; see Figure \ref{fig.criteria} for an illustration.

Based on the discussion above, we deduce that both IPC and SPC can ignore blocks whose column spaces are almost orthogonal to $\bd$; however, algorithms using the IPC are prone to miss correct features during the expansion stage. On one hand, since~\eqref{eq_inner_score} is bounded above by the magnitude of the projection of $\bd$ to $\mathcal{V}_g$, as shown in~\eqref{eq_relation},  if the correlation measured via SPC is low, then the value computed by IPC is also low. On the other hand, even if the correlation measured by projection is high,  depending on the configuration of the columns in $\bF_g$, the score by IPC~\eqref{eq_inner_score} can be low. Consequently, if some columns of a block are close to the orthogonal complement of $\bd$,  the corresponding feature is more likely to be ignored by IPC during the expansion stage. We present a simple example to illustrate this point.

\begin{figure}
\centering
\begin{tabular}{cc}
	(a)&(b)\\
	\includegraphics[width=0.25\textwidth]{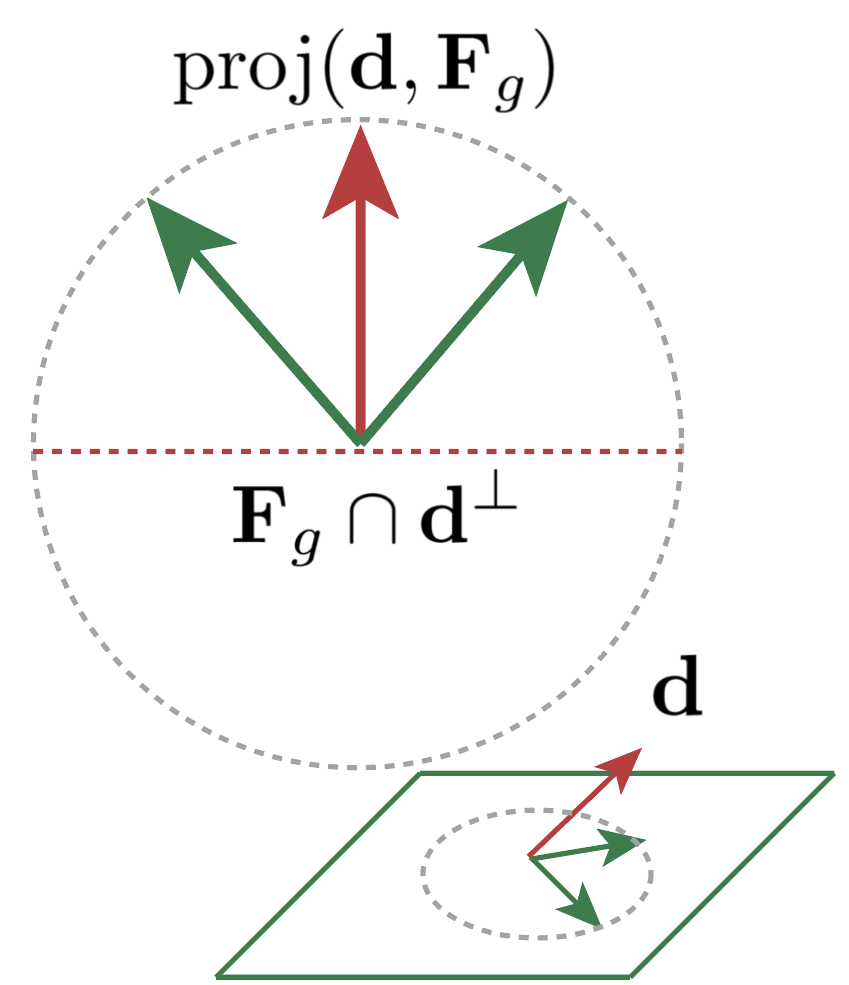}&
	\includegraphics[width=0.25\textwidth]{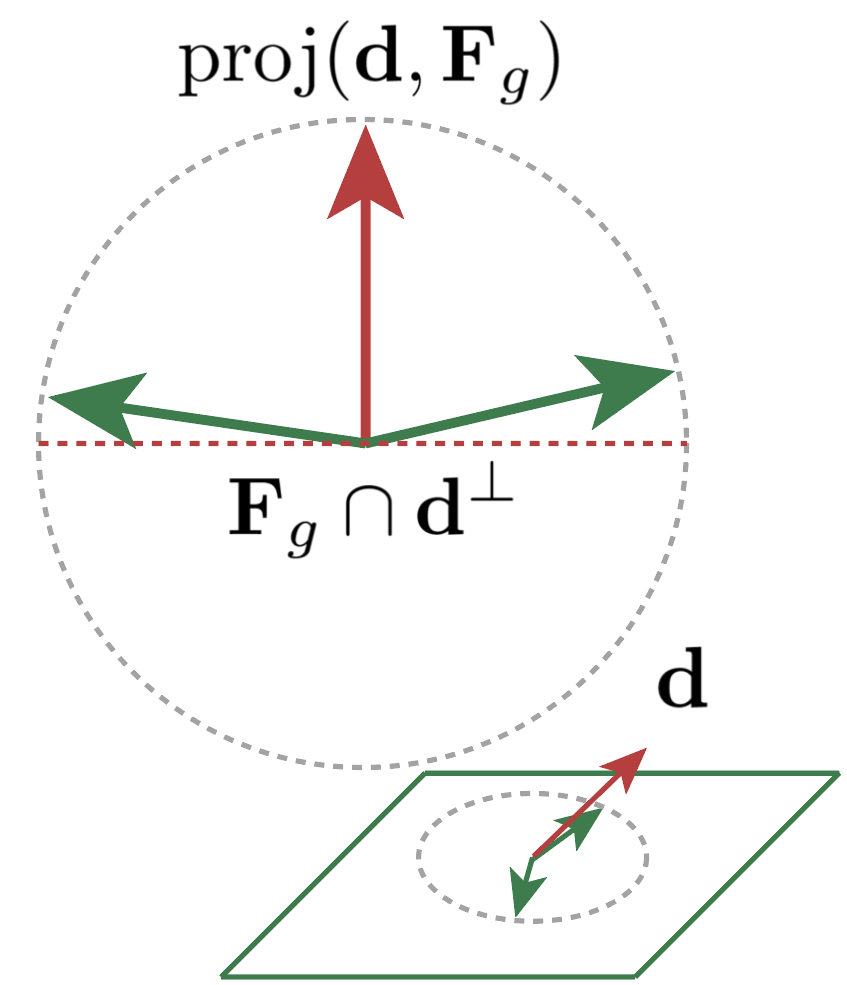}
\end{tabular}
\caption{Influence of the configuration of columns in a block on the correlation measured by inner-product~\eqref{eq_inner_score}. The subspaces spanned by the columns of $\bF_g$ (green arrows) in  (a) and (b) are identical. Since the column vectors are closer to $\bF_g\cap\bd^\perp$, the score by IPC~\eqref{eq_inner_score} for (b) is lower than that for (a); the block $\bF_g$ is considered inferior in (b) compared to (a). In contrast, the score by SPC~\eqref{eq_project_score} remains the same for (a) and (b). }
\label{fig.criteria}
\end{figure}

\begin{example}\label{ex.1} 
Let $\bA=\begin{bmatrix}
	\bF_1 & \bF_2 & \bF_3
\end{bmatrix} $ with
\begin{align*}
	\bF_1=\begin{bmatrix}
		\frac{1}{\sqrt{2}} & 0 \\ 0 & 1 \\ \frac{1}{\sqrt{2}} & 0
	\end{bmatrix}, 
	\ \bF_2=\begin{bmatrix}
		\frac{1}{\sqrt{2}} & 0 \\ \frac{1}{\sqrt{2}} & 1 \\ 0 &0
	\end{bmatrix}, 
	\ \bF_3=\begin{bmatrix}
		1 & 0\\ 0 & 0 \\ 0 &1
	\end{bmatrix} \quad \mbox{ and } \quad
	\bd=\begin{bmatrix}
		\frac{1}{\sqrt{2}} \\ 4 \\ \frac{1}{\sqrt{2}}
	\end{bmatrix}.
\end{align*}
In this case, the unique solution  to $\bA\bc=\bd$ with block-sparsity one is $\bc^*=\begin{bmatrix}
	\bc_1 & \bc_2 & \bc_3
\end{bmatrix}^{\top}$ with $\bc_1=\begin{bmatrix}
	1 & 4
\end{bmatrix}$ and $\bc_2=\bc_3=\mathbf{0}$.
If we use~\eqref{eq_inner_score} to quantify the correlation between $\bd$ and individual blocks, we have
\begin{align*}
	\|\bF_1^{\top}\bd\|_2=\sqrt{17}\approx 4.12, \quad \|\bF_2^{\top}\bd\|_2=\sqrt{\left(\frac{1}{2}+2\sqrt{2}\right)^2+16}\approx 5.20, \quad \|\bF_3^{\top}\bd\|_2=1.00.
\end{align*}
Hence, algorithms such as BOMP and BSP will take the second block as the optimal initial candidate. In contrast, if we use~\eqref{eq_project_score}, we observe that
\begin{align*}
	\|\proj(\bd,\bF_1)\|_2\approx4.12, \quad \|\proj(\bd,\bF_2)\|_2\approx4.06, \quad \|\proj(\bd,\bF_3)\|_2\approx1.00.
\end{align*}
This indicates that algorithms such as BOMPR and GPSP prefer the correct first block as the optimal initial candidate.
\end{example}

The example above shows a typical impact caused by configurations of columns within individual blocks. Note that the column space of $\bF_3$ is relatively far from $\by$, thus both criteria rule out the third block. Since the second dimension of $\bd$ has a large value, and the first column of $\bF_1$ has zero in the corresponding dimension, the angle between them is large; this puts $\bF_1$ in an inferior position by~\eqref{eq_inner_score}. Although the second block cannot be the true feature, as its third row contains only zeros, the presence of the non-zero second row makes it  preferable according to~\eqref{eq_inner_score}. In the case of BOMP, this immediately leads to a wrong choice of feature.
In the following, we compare different criteria for feature exclusion.

\subsection{Feature exclusion criteria: RCC vs. RMC}\label{subsec::coef_response}

In the shrinking stages of the algorithms summarized in Table~\ref{tab.diff}, there are two types of criteria. 
Let $\bx\in\mathbb{R}^{MG}$ be a reconstruction during the iteration, i.e., $\bx_T=\left(\bA^\top_T\bA_T\right)^{-1}\bA^\top_T\by$  for some $T\subset\{1,2,\dots,G\}$ and entries associated with other blocks are filled with zeros. The corresponding approximation for the observed data $\by$ is $\bz = \bA\bx=\sum_{g=1}^G\bF_g\bx[g]=\sum_{g=1}^G\bz_g$ with $\bz_g=\bF_g\bx[g]$. To determine which blocks among $\{\bF_g\}_{g\in T}$ are more important, the criterion RCC, used by BCoSaMP and BSP, evaluates the $2$-norm of $\bx[g]\in\mathbb{R}^{M}$, $g=1,2,\dots,G$, the subvector of $\bx$ corresponding to the $g$-th block;  and the criterion RMC, used by GPSP, computes the $2$-norm of $\bz_g$ for $g=1,2,\dots,G$.

These two criteria are closely related. When the sampling matrix satisfies BRIP~\eqref{eq_BRIP_condition} with a small BRIC, these criteria can choose similar features to remove.   
By the BRIP~\eqref{eq_BRIP_condition} presumed for $\bA$, for each $g=1,2,\dots,G$, we have
\begin{align}
\left|\|\bz_g\|_2^2-\|\bx[g]\|_2^2\right|\leq \delta_{M}\|\bx[g]\|_2^2
\end{align}
for some $\delta_{M}\in(0,1)$. When $\delta_{M}$ is close to $0$,  RCC and RMC yield similar results. When $\delta_{M}$ is close to $1$, they may yield different features to remove.  In particular, for some fixed $g\in T$, if $\|\bx[g]\|_2^2$ is sufficiently small, then $\|\bz_g\|_2^2$ is also small; this implies that if BCoSaMP or BSP decides to remove the $g$-th feature from the pool, GPSP tends to do the same. If $\|\bx[g]\|_2^2$ is large, i.e., BCoSaMP or BSP prefers to reserve the $g$-th feature, while GPSP may find this feature to be unimportant. Overall, RMC is more stringent than RCC.

However, if the observation vector $\by$ contains perturbations, RCC and RMC can behave differently. When we compute $\bx$, since $\by$ is not exact, wrong features may be used to fit the perturbation. In particular, if the perturbation vector can be represented by some wrong features and has large angles with the columns of these features, $\bx[g]$ corresponding to these features can have very large magnitude, leading to the regression coefficient criteria by BCoSaMP and BSP failing to indicate the true features. This problem can be mitigated by the reconstructed response criterion  by GPSP. We demonstrate this case by the following example.

\begin{example}\label{ex.2}
Consider a scenario where the first two features are
\begin{align*}
	\bF_1=\begin{bmatrix}
		1 & 0\\ 0 & 1 \\ 0 & 0 \\ 0 &0
	\end{bmatrix}, \quad
	\bF_2= \begin{bmatrix}
		0 & 0\\ 0 & 0 \\ 1 & \frac{10}{\sqrt{101}}\\ 0 & \frac{1}{\sqrt{101}}
	\end{bmatrix},
\end{align*}
and $
\by=\begin{bmatrix}
	1 &1 & 0 &0
\end{bmatrix}^{\top}.$ The unique solution with block-sparsity one is $\bc^*[1]=[1 \quad 1]^{\top}$ and $\bc^*[g]=\mathbf{0}$ for $g>1$.
Assume at the shrinking stage of the $l$-th iteration, the features in the pool are $\bF_1,\bF_2$.
The two criteria in the shrinking stage give the same scores:
$$
\bF_1:\  \|\bx_p^l[1]\|_2=\|\bF_1\bx_p^l[1]\|_2=1.1412, \quad \bF_2: \ \|\bx_p^l[2]\|_2=\|\bF_2\bx_p^l[2]\|_2=0.
$$
Both of them select the correct feature.

Now consider a small perturbation of $\by$ and the observation vector is  $
\widetilde{\by}=\begin{bmatrix}
	1 &1 & 0 &0.1
\end{bmatrix}^{\top}$ instead. 
In this case, we have $\bx_p^l=\begin{bmatrix}
	1 & 1 & -2 & 2.0025
\end{bmatrix}^{\top}$. By using criterion RCC (BCoSaMP, BSP), we get
\begin{align*}
	\bF_1:\  \|\bx_p^l[1]\|_2=1.1412, \quad \bF_2: \ \|\bx_p^l[2]\|_2=2.8302.
\end{align*}
The second feature $\bF_2$ will be selected. This is because to fit the small perturbation which has a large angle with the two columns of $\bF_2$, one has to use coefficients with large magnitude, leading to a higher score of $\bF_2$ than that of $\bF_1$. 

By using criterion RMC (GPSP), we get
\begin{align*}
	\bF_1: \ \|\bF_1\bx_p^l[1]\|_2=1.4142, \quad \bF_2: \ \|\bF_2\bx_p^l[2]\|_2=0.1.
\end{align*}
In this case, even though $\bF_2$ has a large coefficient, they are used to reconstruct the small perturbation. Thus its corresponding reconstructed response has a small score. The correct feature $\bF_1$ will be selected. 
\end{example}

Finally, we note that RCC is a direct extension of the criterion used in CoSaMP~\cite{needell2009cosamp} and SP~\cite{dai2009subspace}, while RMC does not follow the same paradigm. We numerically compare these two criteria in Subsection~\ref{subsec:ablation} via a group of ablation studies.

\section{Numerical Experiments}\label{sec::numerical}
In this section, we present various experiments to validate the effectiveness of GPSP by conducting a series of comparison and ablation studies under diverse scenarios. We focus on comparing GPSP~\cite{he2023group} with BOMP~\cite{eldar2010block}, BOMPR~\cite{fu2014block}, BCoSaMP~\cite{zhang2019recovery}, and BSP~\cite{kamali2013block}.  We test these methods in practical cases where block features have different statistical properties with and without noise perturbations.  We investigate the algorithmic performances on random sampling matrices with non-Gaussian types.  Furthermore, we consider more complicated settings in face recognition and PDE identification to test GPSP and state-of-the-art methods. All the algorithms are implemented in Python and available online\footnote{ \url{https://github.com/RoyYuchenHe/BlockSparse}}.

\subsection{Heterogeneous Gaussian blocks}\label{sec_hetero}
In this subsection, we consider recovering block sparse signals with sampling matrices with heterogeneous Gaussian blocks.

We consider a block matrix $\bA=[\bF_1,\dots,\bF_G]\in\mathbb{R}^{N\times GM}$ with $N$ observations and $G$ blocks, each of which $\bF_g\in\mathbb{R}^{N\times M}$ is a Gaussian random matrix where each entry is independently and identically sampled from a normal distribution $\mathcal{N}(\mu_g,\sigma_g)$ with mean $\mu_g$ and standard deviation $\sigma_g\ne0$. For each $g=1,2,\dots, G$, the mean $\mu_g$ is generated from $\mathcal{N}(1,5)$, and $\sigma_g$ is generated from the absolute value of a Gaussian variable following $\mathcal{N}(1,5)$. For the true signals $\bc$, the first $k$ entries are simulated by a normal distribution $\mathcal{N}(\mu_c,1)$ and the remaining are zeros, where $\mu_c$ are generated from a normal distribution $\mathcal{N}(1,5)$. The corresponding observations are produced by $\bb=\bA\bc$. We fix $GM=1000$ and $N=400$. For $M=5,8,$ and $10$, we test BCoSaMP, BOMP, BOMPR, BSP, and GPSP on their performances in recovering signals with block sparsities varying from $1$ to $\frac{N}{2M}$ to guarantee the uniqueness of the solutions. In addition, for some experiments, we consider processing $\bA$ with column normalization so that each column has a unit norm.  For each setting, we repeat the experiments 100 times and record the number of experiments where block features are successfully identified.

Figure~\ref{fig:exact} shows the experiment results. Although the probability of exact recovery for each method reduces as the underlying signals have fewer zero blocks, we observe that GPSP outperforms other methods for varying levels of block sparsity and different block sizes. In particular, we find that BOMP has almost less than $50\%$ of success rate for all levels of sparsity, with and without column normalization. Such unsatisfying behaviors of BOMP were also reported in~\cite{fu2014block} when the blocks are redundant.  We note that BOMPR presents better recovery results than BOMP, further validating the effectiveness of using projection for feature selection, as discussed in Section~\ref{sec::comparison}. Generally better than BOMP, the success rates of BSP and BCoSaMP are improved by column normalization, although they are consistently lower than those of BOMPR and GPSP. For signals with single non-zero blocks, both GPSP and BOMPR identify the correct features in all experiments with and without column normalization.   For lower levels of sparsity ($2\sim 5$), both GPSP and BOMPR  keep their success rates above $50\%$.  For middle levels of sparsity ($6\sim 8$), we observe that the success rate of BOMPR drops below $50\%$ especially when the columns are not normalized; whereas the success rate of GPSP still maintains above $50\%$. For higher levels of sparsity (above $10$), the success rate of  GPSP slowly decreases while remaining above $20\%$, and the success rate of BOMPR stays below that of GPSP in all settings. 

We conclude from this set of experiments that GPSP consistently identifies the underlying block sparse signals with the highest success rate with varying sparsity levels when the sampling matrices contain heterogeneous Gaussian blocks. The heterogeneity of the blocks reflects discrepancies among group features, which is commonly expected in real applications.   

\begin{figure}[t!]
\centering
\begin{tabular}{c|c|c}
	\hline
	$M=5$& $M=8$ & $M=10$  \\\hline
	\multicolumn{3}{c}{Without column normalization}\\\hline
	\includegraphics[width=0.3\textwidth]{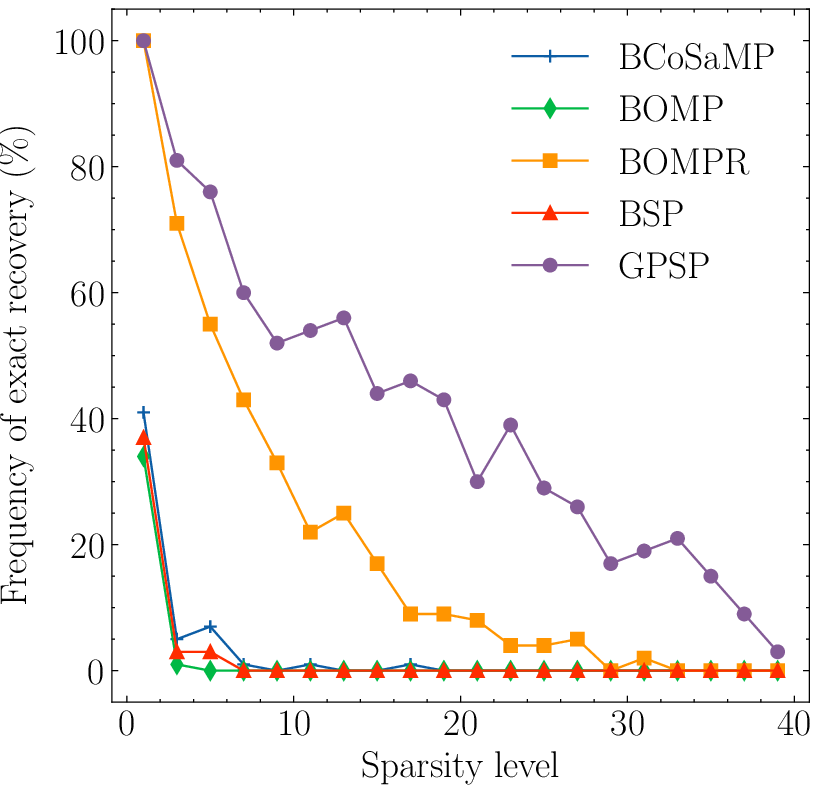}&
	\includegraphics[width=0.3\textwidth]{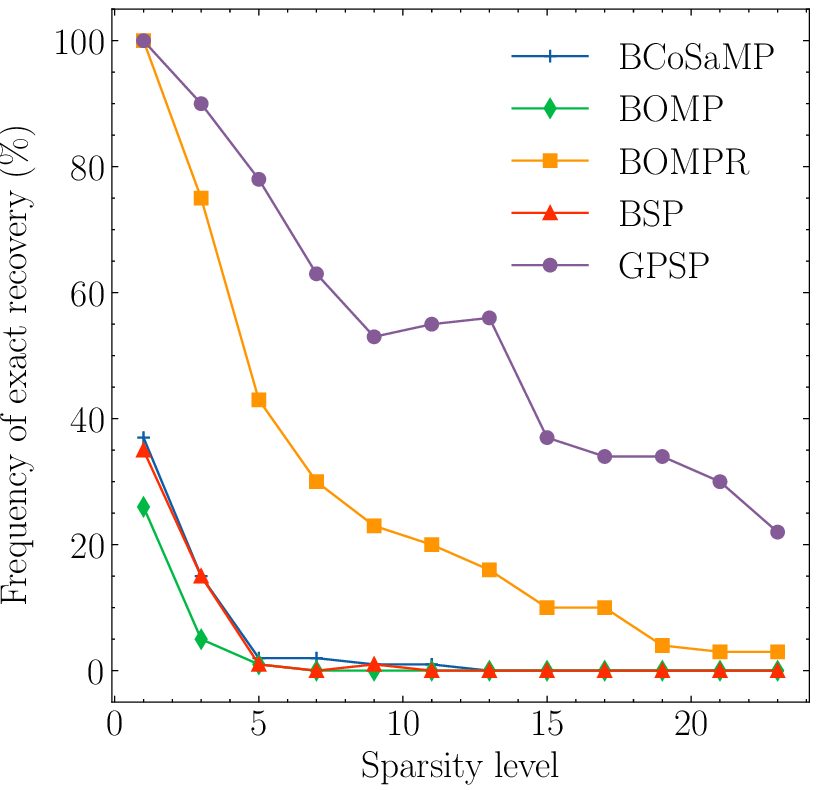}&
	\includegraphics[width=0.3\textwidth]{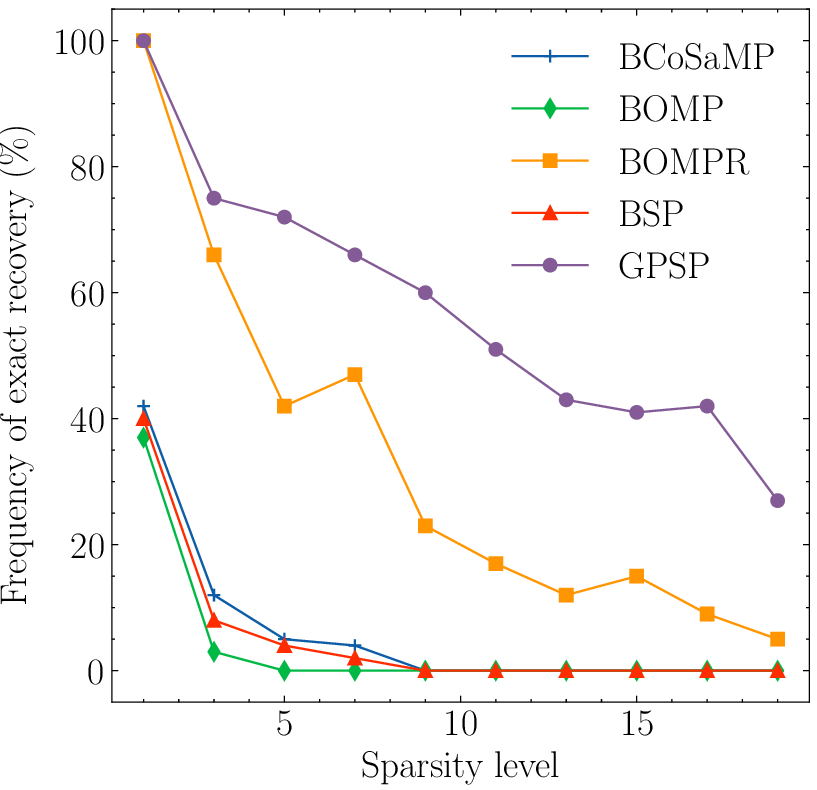}\\\hline
	\multicolumn{3}{c}{With column normalization}\\\hline
	\includegraphics[width=0.3\textwidth]{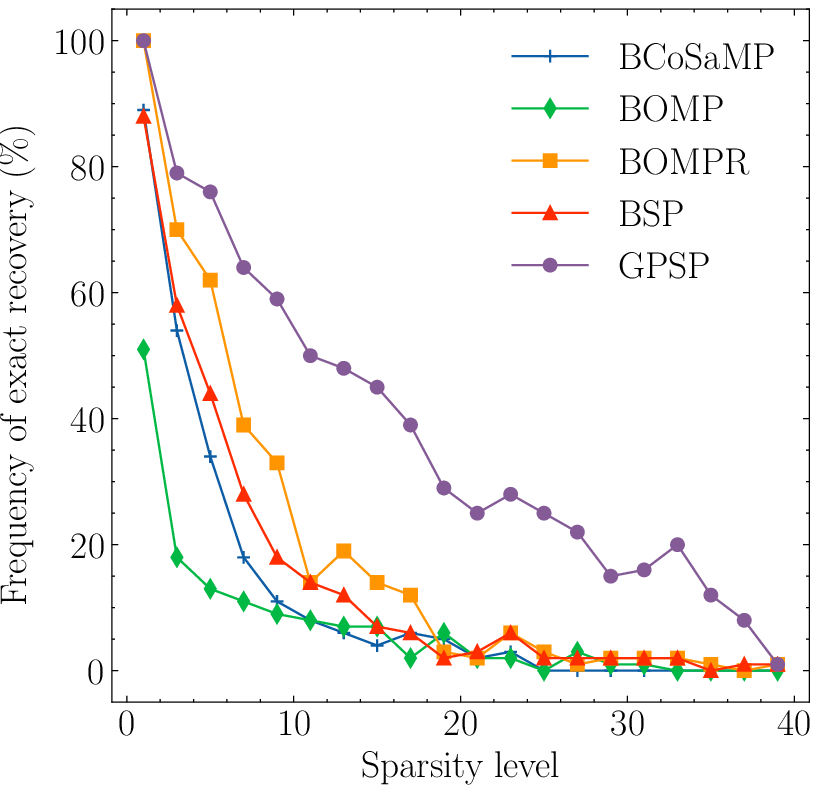}& 
	\includegraphics[width=0.3\textwidth]{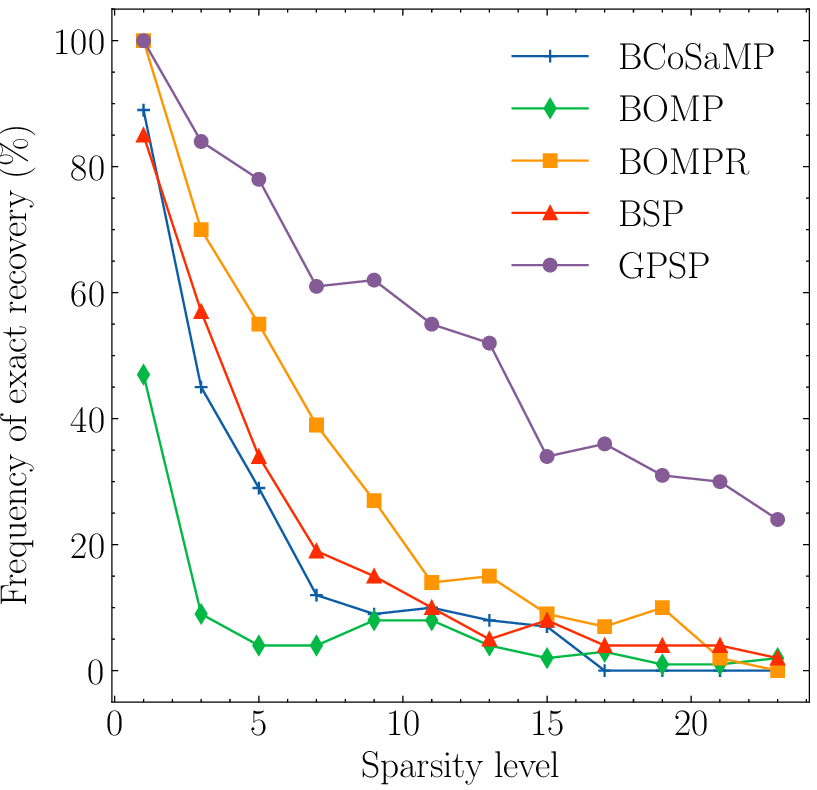}&
	\includegraphics[width=0.3\textwidth]{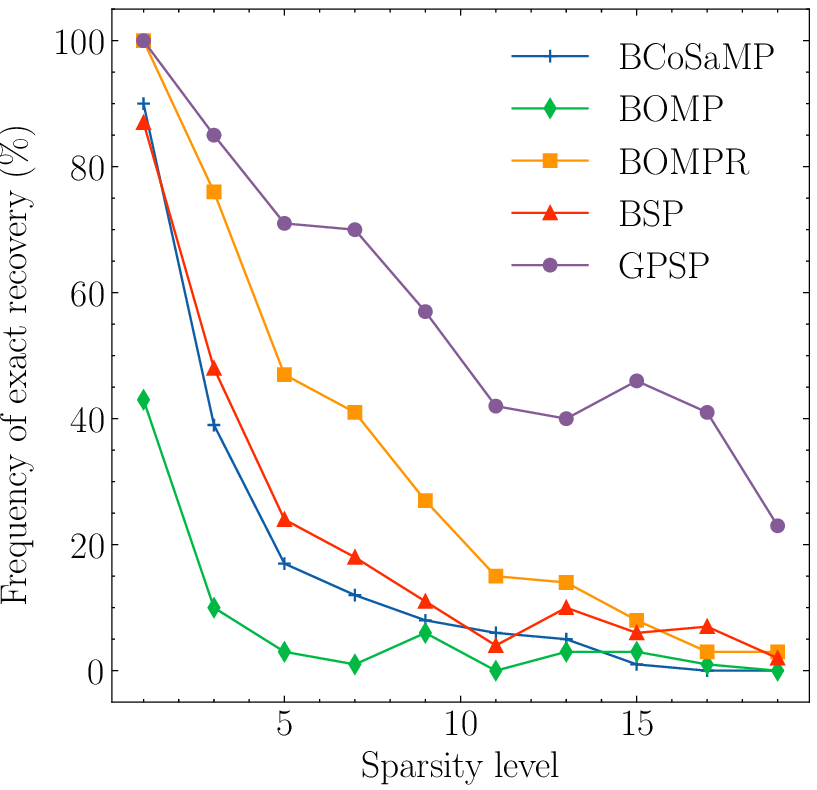}\\\hline
\end{tabular}
\caption{Comparison study on the frequency of exact recovery of GPSP, BCoSaMP, BOMP, BOMPR and BSP with random sampling matrices consisting of heterogeneous Gaussian blocks. From left to right, each column corresponds to the block size $M=5,8,10$, respectively. The first row shows the results without column normalization and the second row shows the results when the columns are normalized. In all the tested settings, GPSP consistently achieves the highest levels of accuracy.}
\label{fig:exact}
\end{figure}
\subsection{Reconstruction with inexact data}
We test and compare all methods when the data are inexact. After generating the sampling matrices $\bA\in\mathbb{R}^{N\times GM}$, the simulated true signals, and the corresponding responses $\bb$ as in subsection~\ref{sec_hetero}, we perturb entries of $\bb$ with independent Gaussian noises with mean $0$ and standard deviation $\sigma>0$. We repeat the experiments in subsection~\ref{sec_hetero} with the perturbed sampling matrices with column normalization, and Figure~\ref{fig:inexact} shows the results when $\sigma=0.2$ and $1.0$ separately. Compared to cases with clean data, we observe that the exact recovery rates of all methods are not significantly affected by the noise when the true signals have single non-zero blocks, and the success rates decrease for higher sparsity levels. Moreover, we find that the success rates of all the tested methods are not strongly compromised by the level of noise added to the observations. Consistent with the results without perturbation, GPSP still outperforms other methods in all settings, and BOMPR remains the second-to-the-best for most scenarios.  These results validate the robustness of GPSP when the observations are perturbed by noise, and the clear advantages of GPSP, as well as BOMPR over the other methods, further justify the effectiveness of subspace projection for feature selection in non-trivial settings.
\begin{figure}[t!]
\centering
\begin{tabular}{c|c|c}
	\hline
	$M=5$& $M=8$ & $M=10$  \\\hline
	\multicolumn{3}{c}{Data perturbation $\sigma=0.2$}\\\hline
	\includegraphics[width=0.3\textwidth]{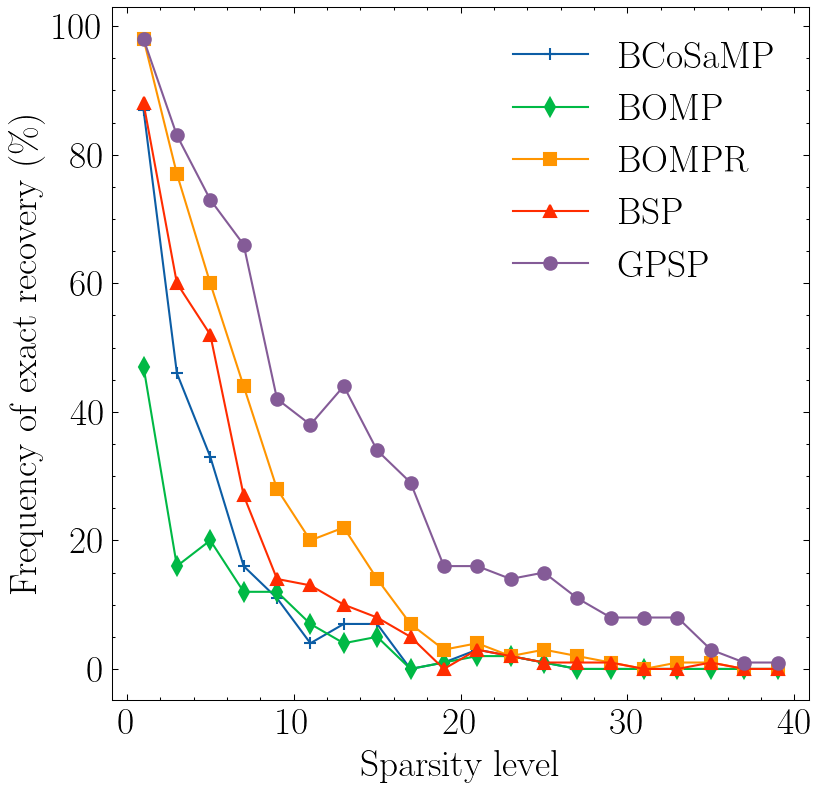}& 
	\includegraphics[width=0.3\textwidth]{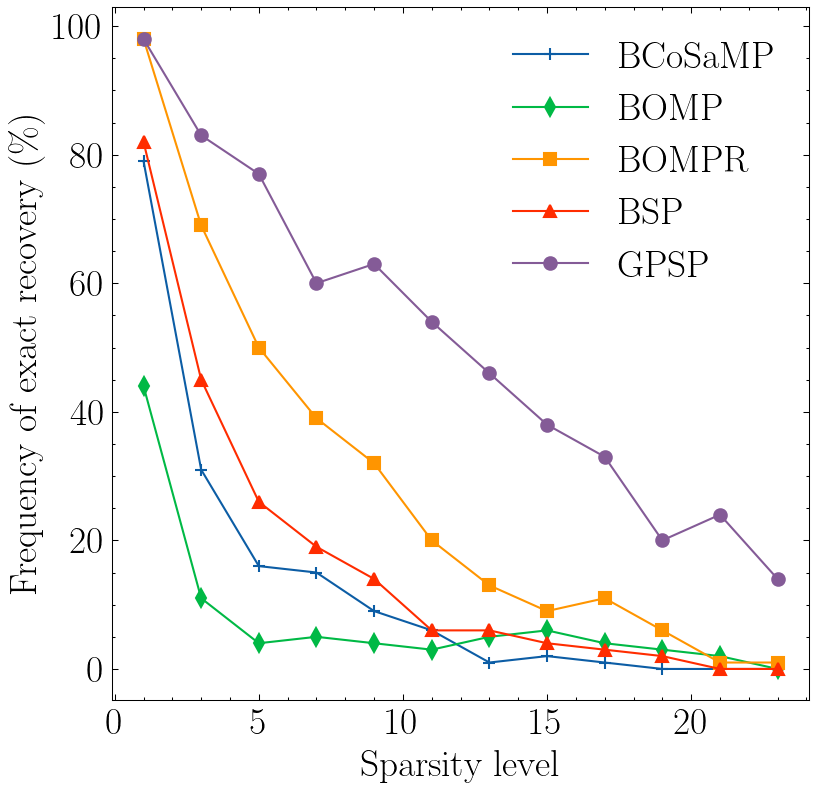}&
	\includegraphics[width=0.3\textwidth]{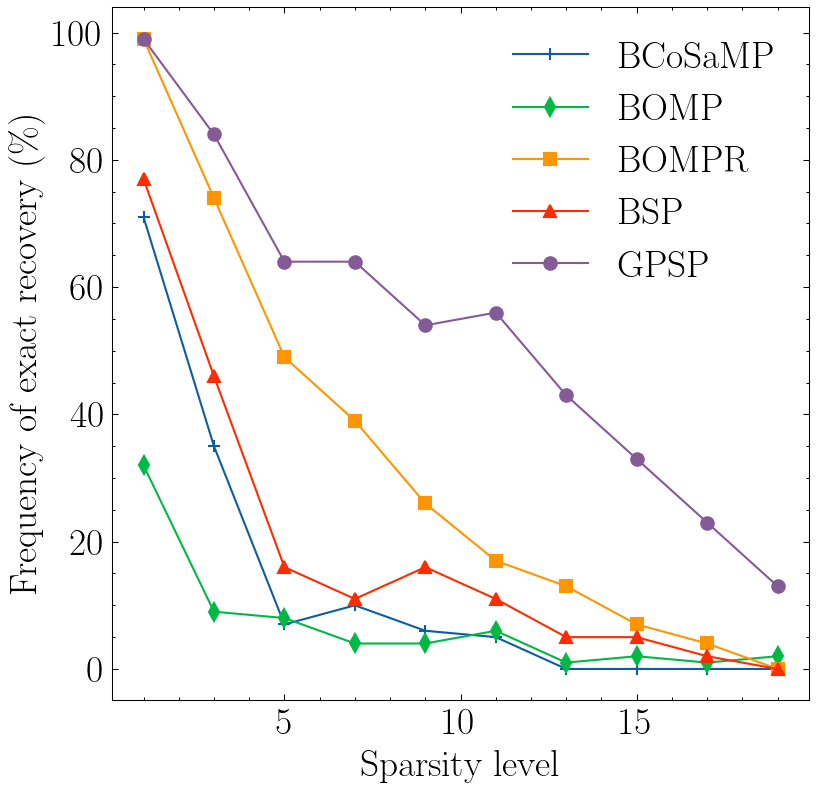}\\\hline
	\multicolumn{3}{c}{Data perturbation $\sigma=1.0$}\\\hline
	\includegraphics[width=0.3\textwidth]{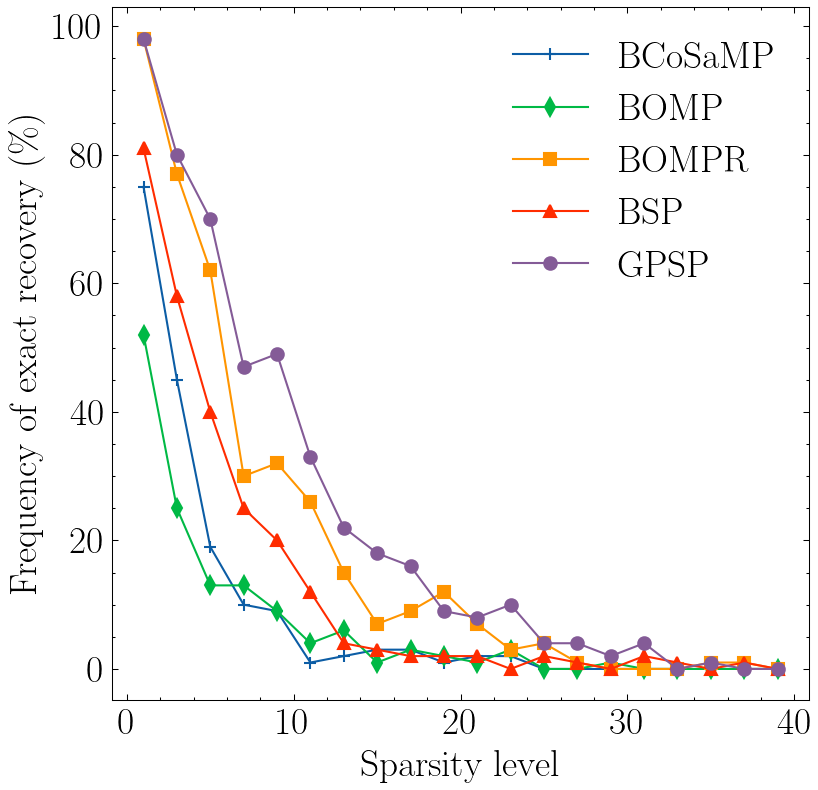}&
	\includegraphics[width=0.3\textwidth]{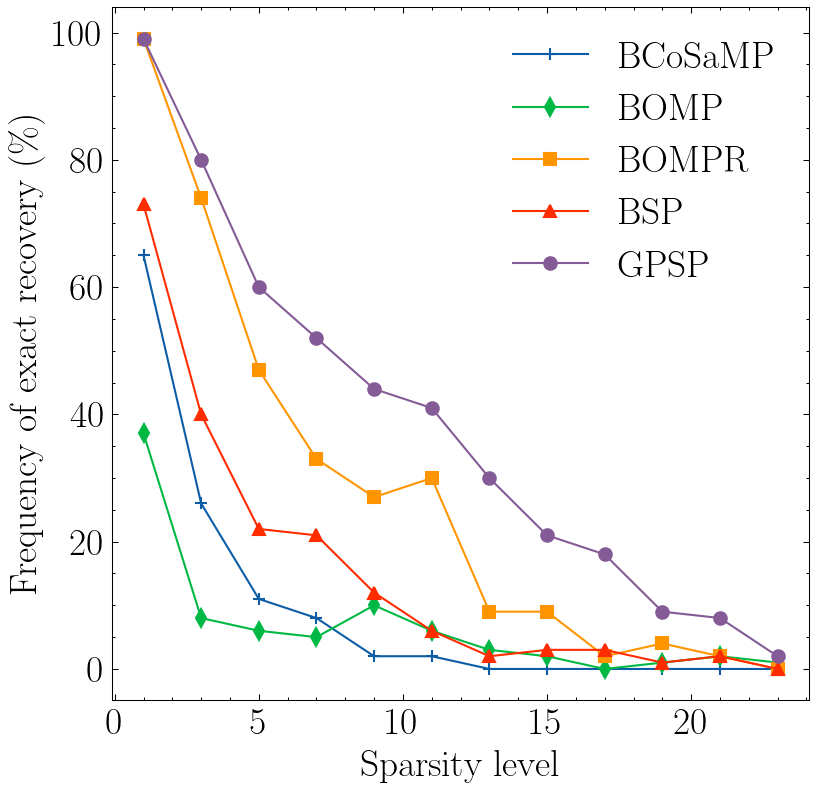}&
	\includegraphics[width=0.3\textwidth]{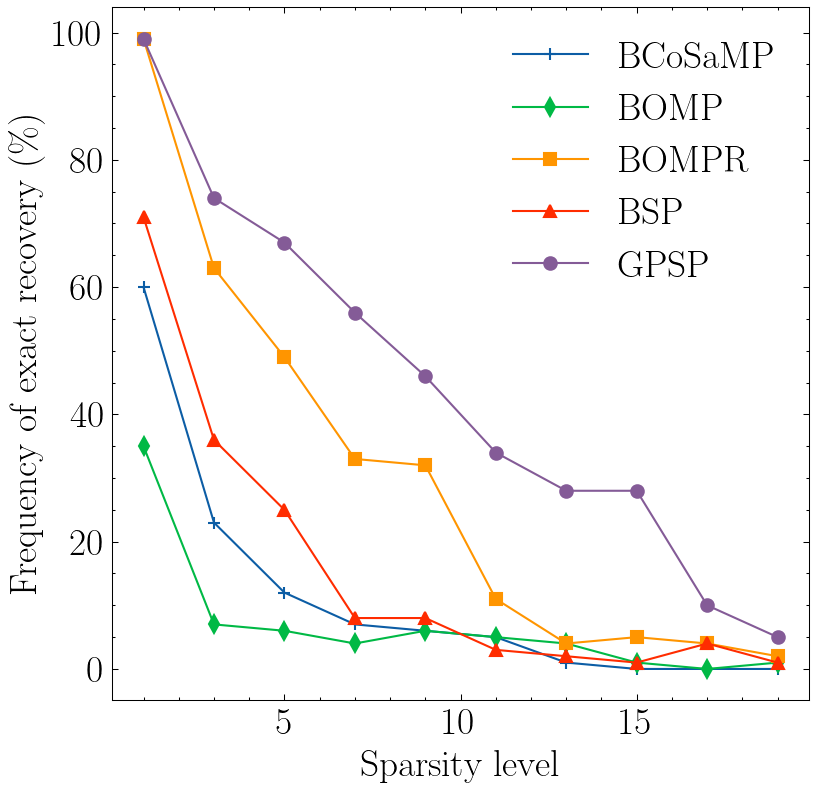}\\\hline
\end{tabular}
\caption{Comparison study on the frequency of exact recovery of GPSP, BCoSaMP, BOMP, BOMPR and BSP with random sampling matrices consisting of heterogeneous Gaussian blocks and inexact data. Columns are normalized. The observed responses are perturbed by additive Gaussian noises with mean $0.0$ and standard deviation $\sigma>0$. From left to right, each column corresponds to the block size $M=5,8,10$, respectively. The first row shows the results when $\sigma=0.2$, and the second rows shows the results when $\sigma=1.0$. Although compromised by the perturbations, in all the tested settings, GPSP consistently achieves the highest level of accuracy. 
}
\label{fig:inexact}
\end{figure}
\subsection{Other types of random sampling matrices}
To account for more general types of randomness, we consider  cases where entries of the sampling matrix are  Poisson or Bernoulli random variables. In the first scenario, each block of $\bA$ is a random matrix whose entries are independent Poisson variables, and the mean of each block is set to be the square of a Gaussian random variable with mean $5$ and standard deviation $20$. In the second scenario, every element in $\bA\in \mathbb{R}^{N\times GM}$ is either 0 or 1, and for each block, the probability of being 1 is determined by a uniform random variable from $[0,1]$. The true signals are generated from the normal distribution $\mathcal{N}(\mu_{\bx},1)$ where $\mu_{\bx}$ is another random variable following $\mathcal{N}(1,5)$. Figure \ref{fig:distribution} shows the results. 

For the Poisson distributed blocks, GPSP shows the highest exact recovery rate for most of the sparsity levels, and it is the only one whose success rates maintains above $50\%$ up to sparsity level $5$. For the middle levels of sparsity ($10\sim 15$), BCoSaMP outperforms other methods, but the margin decreased when the block size increased. We also observe that BOMPR is not as competitive as in the scenario with Gaussian random blocks; this implies that feature screening as employed in GPSP can be important.  For the binary distributed blocks, GPSP outperforms other algorithms, and its success rate remains above $80\%$ in most of the settings. Interestingly, BSP remains the second best when Bernoulli random blocks are considered, and it shows above $60\%$ accuracy for most cases. BCoSaMP shows comparable levels of accuracy for lower levels of sparsity, but it fails completely and immediately when there are more active features involved. It is consistent that BOMP does not show satisfying performances, and with the subspace projection instead of the inner-product rule, BOMPR significantly improves the results. However, we find that BOMPR is generally inferior to BSP or BCoSaMP in this scenario. 

From these experiments, we deduce that the performances of most of the tested greedy algorithms are highly affected by the types of randomness of the sampling matrices. Moreover, our results show that GPSP is robust to such variations, and it yields successful feature identifications for general levels of sparsity and block sizes.

\begin{figure}[t!]
\centering
\begin{tabular}{c|c|c}
	\hline
	$M=5$& $M=8$ & $M=10$  \\\hline
	\multicolumn{3}{c}{Poisson distribution}\\\hline
	\includegraphics[width=0.3\textwidth]{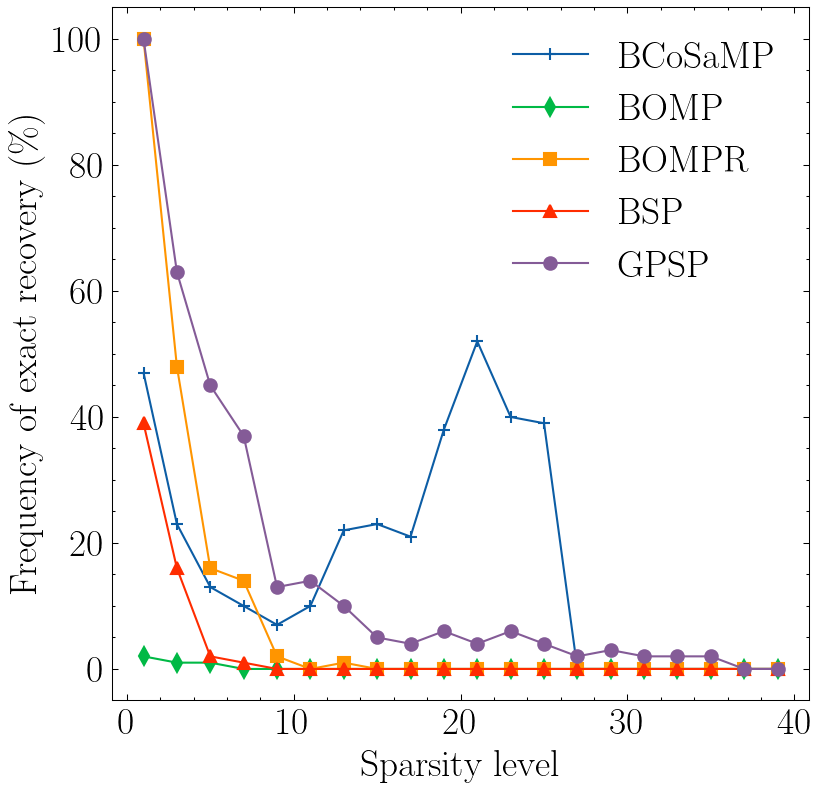}& 
	\includegraphics[width=0.3\textwidth]{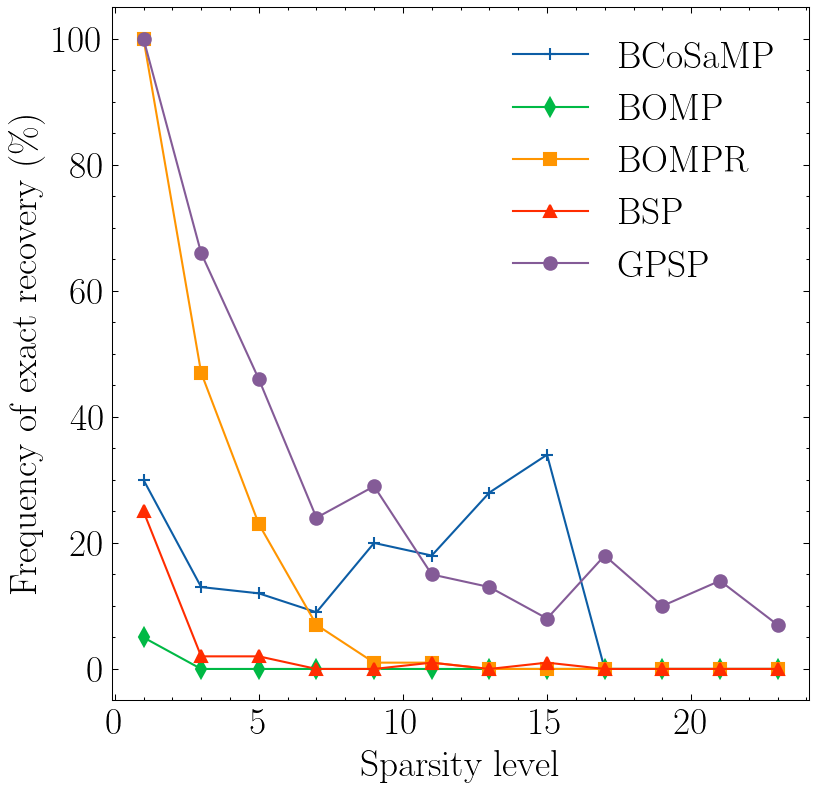}&
	\includegraphics[width=0.3\textwidth]{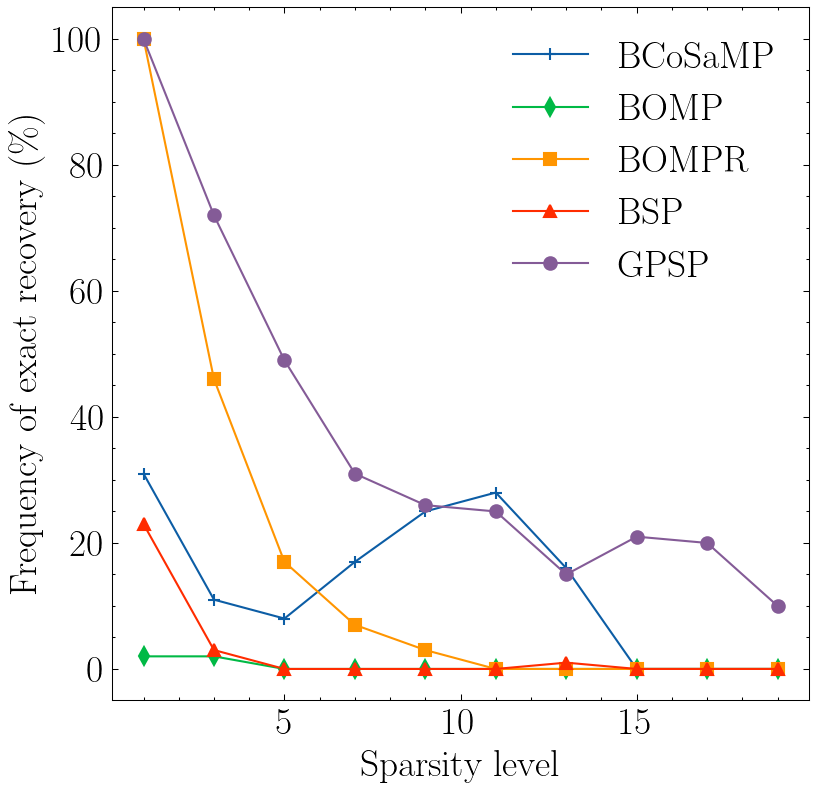}\\\hline
	\multicolumn{3}{c}{Bernoulli distribution}\\\hline
	\includegraphics[width=0.3\textwidth]{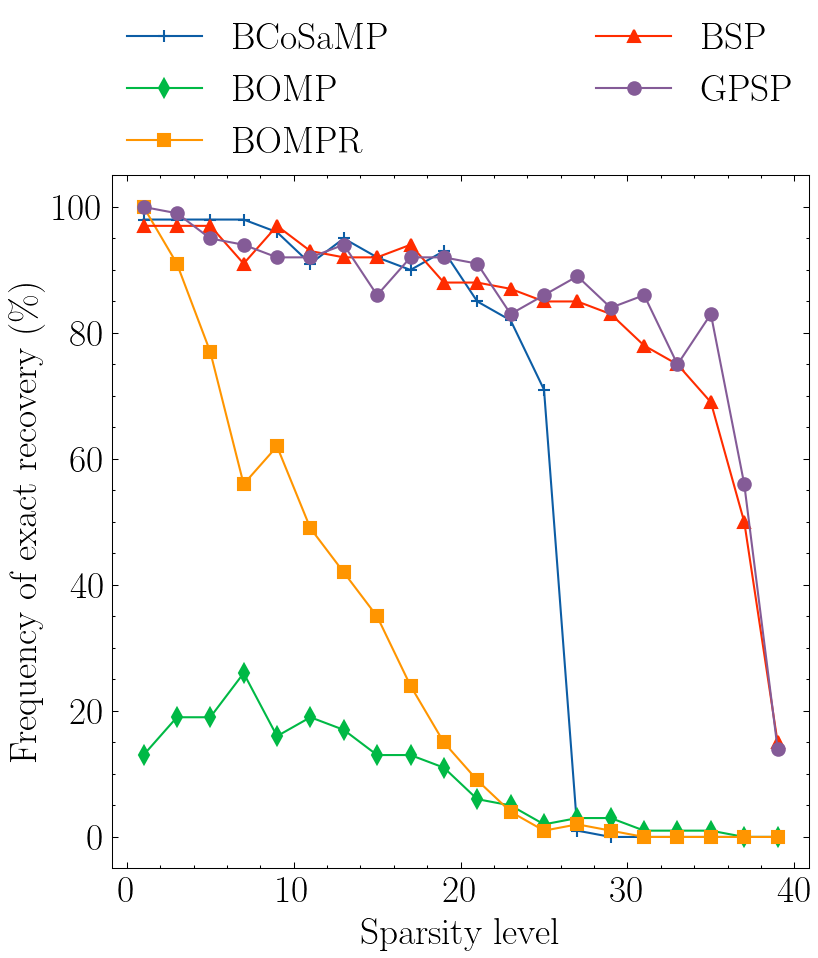}&
	\includegraphics[width=0.3\textwidth]{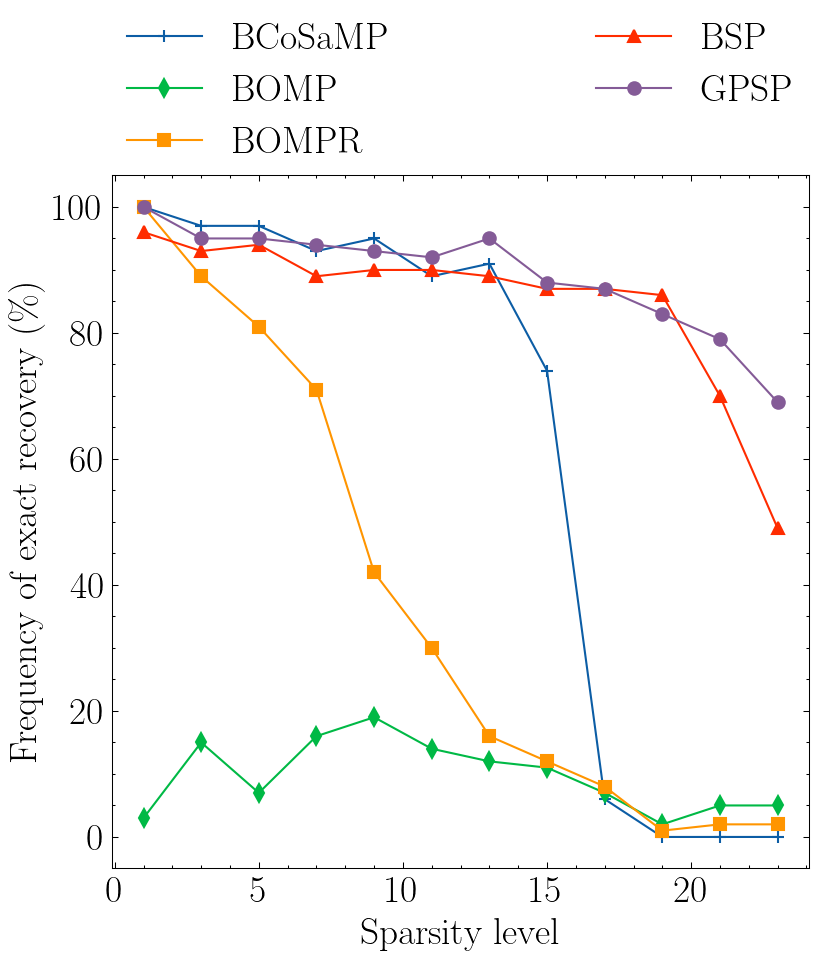}&
	\includegraphics[width=0.3\textwidth]{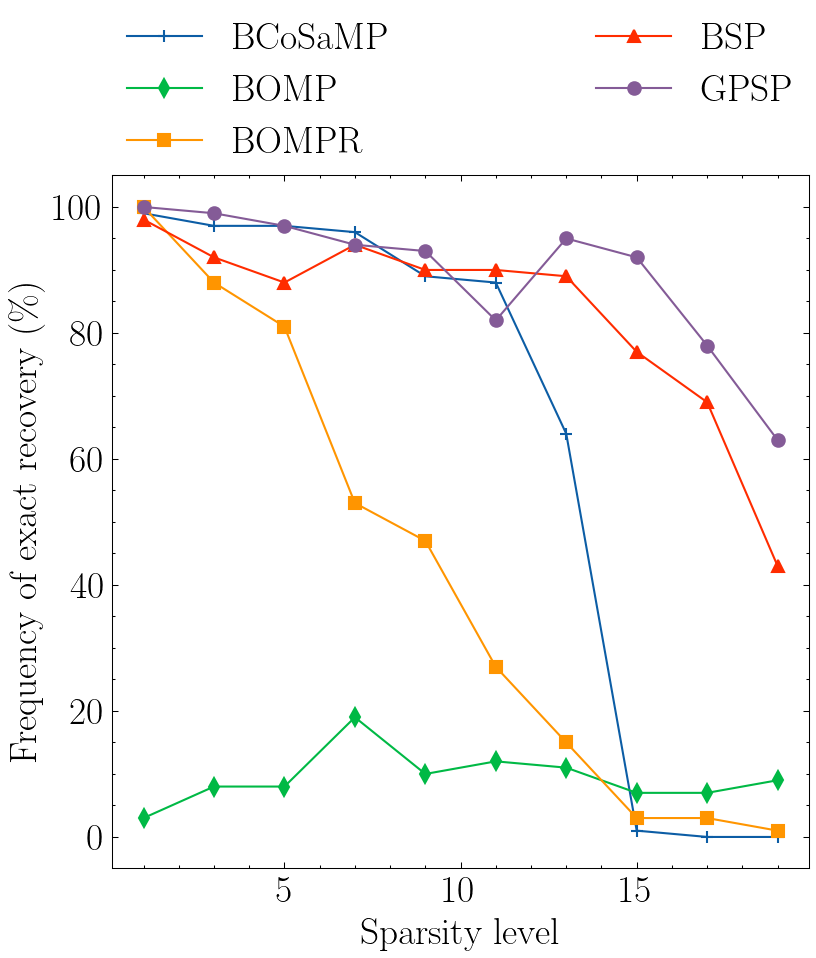}\\\hline
\end{tabular}
\caption{Other types of random sampling matrices. Comparison of GPSP with BCoSaMP, BOMP, BOMPR and BSP. From left to right, each column corresponds to $M=5,8,10$, respectively. The first row shows the result for sampling matrix following Poisson distribution and the second row shows results for sampling matrix following Bernoulli distribution.}
\label{fig:distribution}
\end{figure}

\subsection{Application to Face Recognition}
One of the applications of block sparse signal recovery in image processing is face recognition \cite{liu2017subspace,elhamifar2012block}. Given a face image, the goal is to identify the person from a list of subjects whose face images under different lighting conditions are archived. For each subject, the corresponding archived face images are grouped as a block feature, and our sampling matrix, or dictionary, comprises all the block features; and we can treat the test image as the response. Under the Lambertian assumption that each image is represented as a linear combination of other images of its subject in the dictionary,  the problem of face recognition can be addressed by finding the block, i.e., the subject, whose feature matches the test image, or equivalently, by reconstructing a signal with block sparsity one.

We evaluate and compare the performances of GPSP with other algorithms using the Extended Yale B dataset~\cite{lee2005acquiring} derived from~\cite{GeBeKr01}. This dataset consists of face images belonging to 38 human subjects. Each subject is represented by 64 face images, which are captured under various illuminating conditions. The first row of Figure \ref{tab:face} shows 8 images of a specific subject; they exemplify the variations in the dataset, including contrast change (the 2nd image), minor differences in facial expressions (the 3rd image), and extremely low lighting (the 5th image). For each subject, we randomly sampled $M$ images to form our training dataset of size $38\times M$, and we use the rest images for  evaluation. Upon stretching each image to a vector of length $32256\,(=192\times168)$,   we employ three strategies to reduce the dimension to $D=132$ as in~\cite{elhamifar2012block}: 
\begin{itemize}
\item \textbf{Principal Component Analysis (PCA)}. Project the face vectors onto the first $D$ principal components of the covariance matrix obtained from the training data.
\item \textbf{Random Projection (randomProj)}. Multiply the face vectors by a random projection matrix, which is constructed by generating i.i.d. entries from a Gaussian distribution with zero mean and variance $1/D$.
\item \textbf{Downsampling (downSample)}. Randomly sample $D$ entries from the image vector.
\end{itemize}

More specifically, for a random separation of the training and testing datasets, and with a fixed dimension reduction strategy, we construct a block matrix $\bA=[\bF_1,\dots,\bF_G]\in\mathbb{R}^{D\times GM}$, and for  $i=1,2,\dots,G$, each column of $\bF_i\in\mathbb{R}^{D\times M}$ is a dimension-reduced face image of the $i$-th subject.  Given a test image, we obtain a dimension-reduced vector $\bb\in\mathbb{R}^{D}$. Suppose the test image belongs to the $i$-th subject, and the predicted block feature is $j$, then the accuracy is measured by the ratio of occurrences of correct matches, i.e., $i=j$. We test BCoSaMP, BOMP, BOMPR, BSP, and GPSP with  $M=9, 18, 25, 32$ for increasing sizes of training datasets. 

In general, we observe that BCoSaMP, BOMP, and BSP have consistently lower levels of accuracy compared to BOMPR and GPSP. 
With smaller sizes of training datasets ($M=9$ and $M=18$), GPSP outperforms the other algorithms. When the size of the training dataset becomes larger ($M=25$ and $M=32$), BOMPR achieves slightly higher accuracy than GPSP when using PCA for the dimension reduction; and for the other two dimension reduction strategies, GPSP has higher accuracy than BOMPR. We also observe that the rates of exact recovery of both BOMPR and GPSP stay stable for the tested reduction strategies; they differ by at most $2\%$ and the results using downSample for reduction are always the worst. In contrast, the performances of BCoSaMP, BOMP and BSP are strongly affected by the choice of reduction techniques. Taking $M=9$ as the example, for BCoSaMP, the maximal difference is $17.1814\%$; for BOMP, it is $3.4459\%$; and for BSP, it is $18.4074\%$. Noticeably, these methods all use IPC for feature selection, and both BOMPR and GPSP use SPC. From these results, we conclude that GPSP is an effective algorithm for feature selection with complicated real data, and this superior performance is empowered by the SPC for candidate inclusion.  
\begin{table}
\centering
\begin{tabular}{c|c|ccccc}
	\multicolumn{7}{c}{\includegraphics[width=0.7\textwidth]{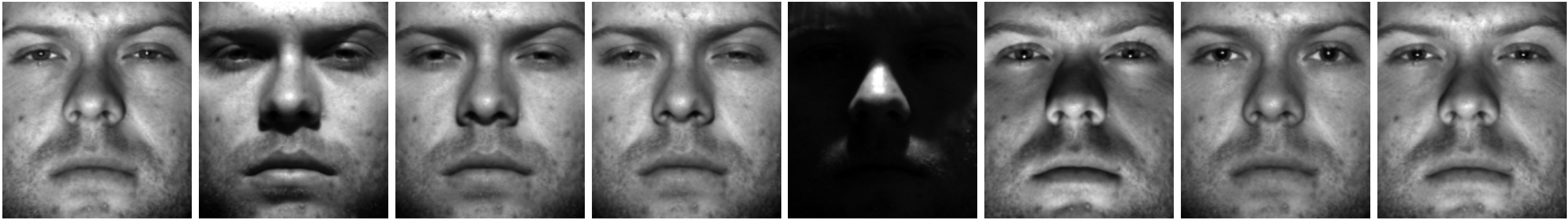}}\\\hline
	$M$&Mode&BCoSaMP& BOMP& BOMPR& BSP& GPSP\\\hline
	\multirow{3}{*}{9}&PCA&56.8629\%&6.0376\%&\underline{80.6129\%}&51.5396\%&\textbf{81.8726\%}\\
	&randomProj& 43.8851\%&3.2529\%&\underline{79.4788\%}&37.2056\%&\textbf{80.4537\%}\\
	&downSample&39.6815\%&2.5917\%&\underline{78.7597\%}&33.1322\%&\textbf{79.4546\%} \\\hline
	\multirow{3}{*}{18}&PCA&58.4855\%  & 6.8497\% & \underline{89.8324\%} & 60.5318\% & \textbf{89.9249\%} \\
	&randomProj& 35.7630\% & 2.9075\%& \underline{90.0983\%}& 34.8960\% &\textbf{90.4971\%}\\
	&downSample&33.8844\%&2.6879\%& \underline{88.1618\%}& 31.1156\%&\textbf{88.6127\%}\\\hline
	\multirow{3}{*}{25}&PCA&54.9180\%& 5.3552\%& \textbf{92.2473\%}& 60.6216\%& \underline{91.8989\%}\\
	&randomProj&29.3921\%& 2.5342\%&\underline{92.5751\%}&30.669\%& \textbf{92.8689\%} \\
	&downSample&31.2090\%& 2.5683\%&\underline{91.0861\%}& 30.5055\%& \textbf{91.3046\%} \\\hline
	\multirow{3}{*}{32}&PCA&48.6060\%& 5.6594\%& \textbf{92.9048\%}& 59.4658\%& \underline{92.3205\%}\\
	&randomProj&31.0935\%& 3.0301\%&\underline{93.8063\%}& 33.0050\%& \textbf{93.9065\%}\\
	&downSample&25.6845\%& 2.5876\%&\underline{91.0100\%}& 28.0885\%& \textbf{91.2355\%}\\\hline
\end{tabular}
\caption{Comparison study on the task of face recoginition using the Extended Yale B dataset~\cite{lee2005acquiring}. The first row shows sample images of a single subject. The table records the average values of frequency of exact recovery over 10 trials for different training size per subject $M=9,18,25,32$ and three dimension reduction strategies: principle component analysis (PCA), random projection (randomProj), and down-sampling (downSample). The highest levels of accuracy are bolded, and the second best are underlined.}\label{tab:face}
\end{table}
\subsection{Application to PDE Identification}
GPSP was first proposed in~\cite{he2022robust} for identifying PDEs with time-space dependent coefficients from trajectory data. It was shown that GPSP can find the underlying models for various important classes of PDEs with higher success rates than BSP. In this experiment, we extend the comparison by including also BOMP, BOMPR, and BCoSaMP. We consider the viscous Burgers' equation
\begin{align}
\begin{cases}
	u_t=a(x,t)uu_x+b(x,t)u_{xx},\, x\in [-1,1],&\ t>0,\\
	u(x,0)=\sin(6\pi (x+0.1))+ 1.5\sin(2\pi(x+0.1))\cos(2\pi(x-0.5)),& \ x\in[-1,1],\\
	u(0,t) = u(1,t),&\ t\geq 0
\end{cases}
\label{eq.pde}
\end{align}
over the computational domain $(x,t)\in[-1,1]\times [0,0.3]$, where 
\begin{align}
a(x,t)=0.8+0.2\phi(t)\cos(\pi x), \quad b(x,t)=0.02,
\end{align}
with $\phi(t)=0.5+0.5\tanh(10(t-0.15))$. Following the numerical settings in~\cite{he2023group}, we discretize the computational domain with spatial step 0.02 and temporal step $3\times 10^{-4}$ with which the PDE is solved. Then we downsample the data to spatial step 0.04 and temporal step $3\times 10^{-3}$, such that the data is over a $50\times 100$ grid.

The PDE~\eqref{eq.pde} has non-trivial behaviors. The coefficient $a(x,t)$ is visualized in Figure~\ref{fig:pde}~(a), and the corresponding solution obtained by the spectral method (See~\cite{he2022robust}) is shown in (b). It models a non-linear advection with a homogeneous diffusion on a circle. The speed of the advection is affected by the solution itself as well as the varying coefficient $a$. When the solution $u$ is positive, it moves from right to left; when it is negative, it moves from left to right. Moreover, the higher the magnitude of the solution is, the faster it moves locally. Consequently, solutions can form shocks without diffusion. For more details about the  properties of the PDE, we refer the readers to~\cite{evans2022partial}.   

We set up a dictionary containing candidate differential operators as features to identify the PDE from the solution data. We include partial derivatives of $u$ up to order 4 and products of no more than 3 terms, leading to 55 features in the dictionary. Since the coefficients may vary both in space and time, we approximate them using linear combinations of B-spline bases of order $2$. As a result, each feature is represented by a block matrix whose columns are associated with local bases. We use 7 bases for space and 8 bases for time, leading to a block size of 56. The details of the implementation are explained in~\cite{he2022robust}. Assuming that we know the sparsity is 2, we apply BOMP, BOMPR, BCoSaMP, BSP, and GPSP to find the features. Table~\ref{tab:pde} shows the identified features by these methods. In this experiment, only GPSP finds the correct features, giving the underlying viscous Burgers' equation. BOMP and BOMPR have identical results; they find the correct second-order differential operator, yet the non-linear part is wrong. BCoSaMP and BSP fail to identify any correct features, and the respective models are more complicated than the true equation. This comparison verifies the effectiveness of GPSP in the application of PDE identification.

\begin{figure}[t!]
\centering
\begin{tabular}{cc}
	(a)   & (b)  \\
	\includegraphics[width=0.4\textwidth]{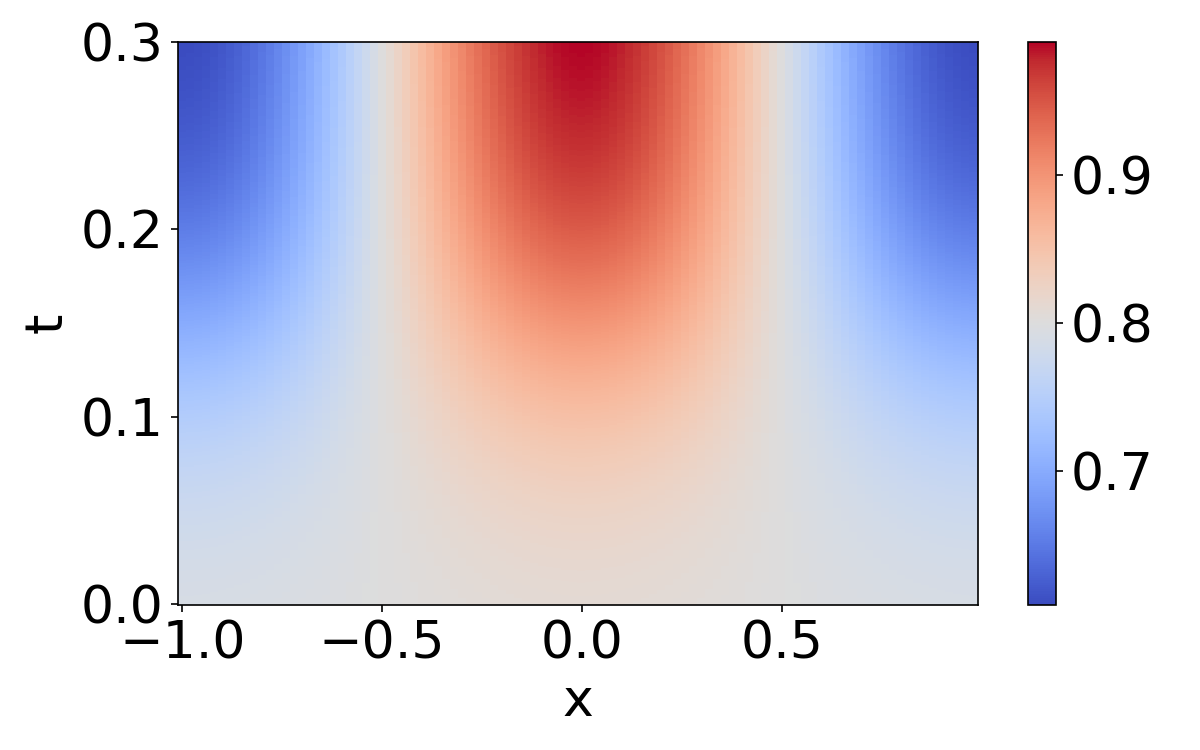} & 
	\includegraphics[width=0.4\textwidth]{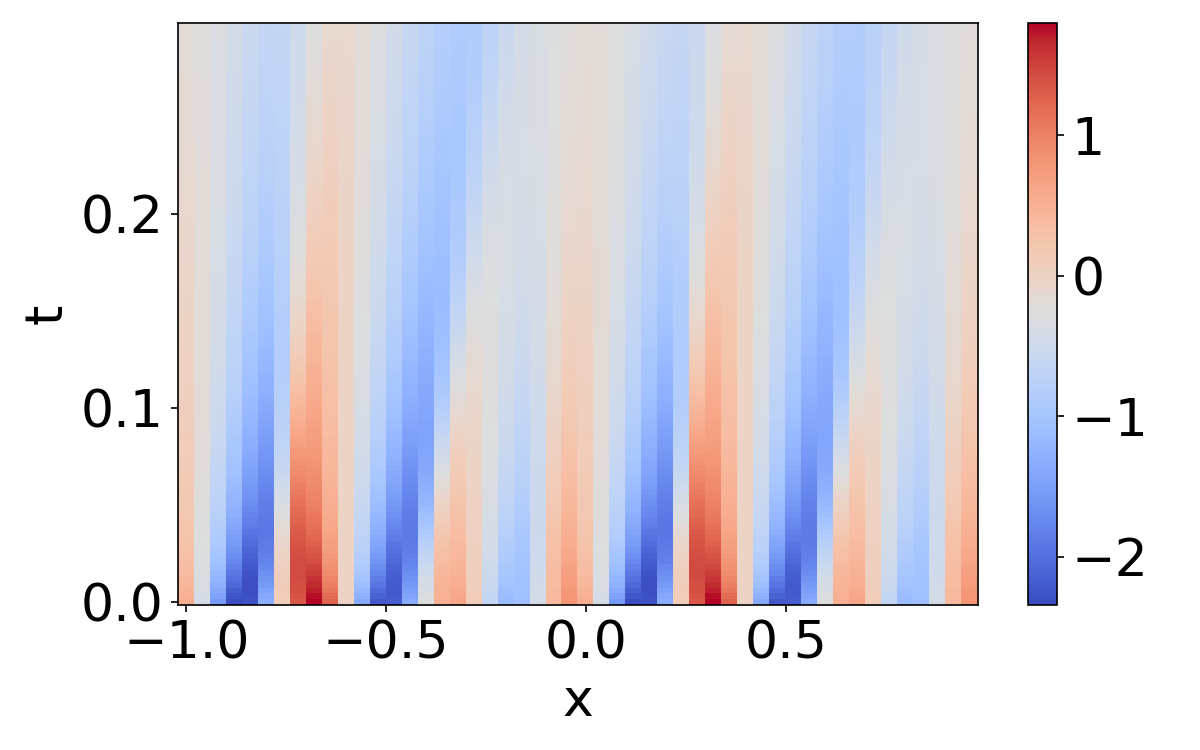}
\end{tabular}
\caption{Visualization of (a) coefficient $a(x,t)$ in (\ref{eq.pde}) and (b) solution data.}
\label{fig:pde}
\end{figure}

\begin{table}[t!]
\centering
\begin{tabular}{c|c}
	\hline
	Algorithm& Identified terms \\
	\hline
	BOMP& $uu_{xxx},u_{xx}$\\
	\hline
	BOMPR & $uu_{xxx},u_{xx}$\\
	\hline
	BCoSaMP & $u_{xx}^2u_{xxxx},u_{xx}u_{xxxx}^2$\\
	\hline
	BSP & $uu_{xxx},u_{x}^3$\\
	\hline
	GPSP & $\bm{uu_x,u_{xx}}$\\
	\hline
\end{tabular}
\caption{Identified PDE iterms by different algorithms. The correct terms are $uu_x,u_{xx}$.}
\label{tab:pde}
\end{table}
\subsection{Ablation study: criteria for feature inclusion and exclusion}
\label{subsec:ablation}
As discussed in Section~\ref{sec::comparison} and illustrated in Table~\ref{tab.diff}, GPSP and the other greedy algorithms mainly differ in their respective rules for the expansion and shrinking of the pools of candidates. For the expansion stage, there are two types of criteria: IPC for the inner product criterion and SPC for the subspace projection criterion. For the shrinking stage, there are also two types of criteria: RCC for regression coefficient criterion, and RMC for response magnitude criterion. Each criterion from different stages can be combined to form a greedy algorithm that iteratively updates the pool of candidates by expansion and shrinking.

We conduct a set of ablation studies to compare algorithms induced from different combinations of these criteria, which are SPC-RCC, IPC-RCC, SPC-RMC, and IPC-RMC. For example, the algorithm with SPC-RCC uses the projection-based criterion during the expansion stage and the regression-coefficient-based criterion during the shrinking stage. Testing in the same setting of heterogeneous Gaussian blocks with column-normalization (see subsection~\ref{sec_hetero}), we observe that the algorithm with SPC-RMC, which is identical to GPSP, consistently achieves the highest success rates for identifying the underlying signals with varying levels of block sparsity and different block sizes. In all settings, the algorithm with IPC-RCC, which is identical to BSP, has the lowest success rates, see Figure~\ref{fig:exact_ablation}. Comparing the curves of SPC-RMC with those of IPC-RMC, we notice significant improvements of the performances. This observation is consistent with the analysis in subsection~\ref{subsec::inner_proj} about the advantages of using subspace projection for robust feature identifications. We also notice that the success rates of the algorithm with SPC-RMC are always higher than that of the algorithm with SPC-RCC, and their differences are not as significant as those between SPC-RMC and IPC-RMC. On one hand, this comparison shows numerical evidence that RMC is a better option than RCC for correctly removing wrong features. On the other hand,  the relatively less significant improvements are compatible with our analysis in subsection~\ref{subsec::coef_response}, where the similarity between the magnitudes of regression coefficients and the corresponding responses are explained. Hence, we conclude that the major factor for the superior performances of GPSP is SPC for expanding the pool of candidates, and RMC  further improves the accuracy. 

To further investigate the role played by RMC, we conduct another set of ablation studies when the observations are perturbed by Gaussian noises of varying intensities ($\sigma=0.1,0.5$ and $1.0$). Figure~\ref{fig:inexact_ablation} shows the results. We find that the curves of success rates of SPC-RMC remain higher than those of SPC-RCC in all settings, and their differences are more prominent compared to those in the previous set of experiments shown in Figure~\ref{fig:exact_ablation}. This implies that the candidate filtering due to RMC is more effective and useful when the data is noisy; this is also compatible with the behaviors illustrated in Example~\ref{ex.2}. 

From these two sets of ablation studies, we conclude that GPSP, i.e., the greedy algorithm with SPC-RMC, outperforms the algorithms with other combinations of expansion and shrinking criteria. Furthermore, we identify SPC as the critical element for superior performances, and we find RMC to help eliminate wrong features especially when the observations are noisy.
\begin{figure}[t!]
\centering
\begin{tabular}{c|c|c}\hline
	$M=5$& $M=8$ & $M=10$  \\\hline
	\includegraphics[width=0.3\textwidth]{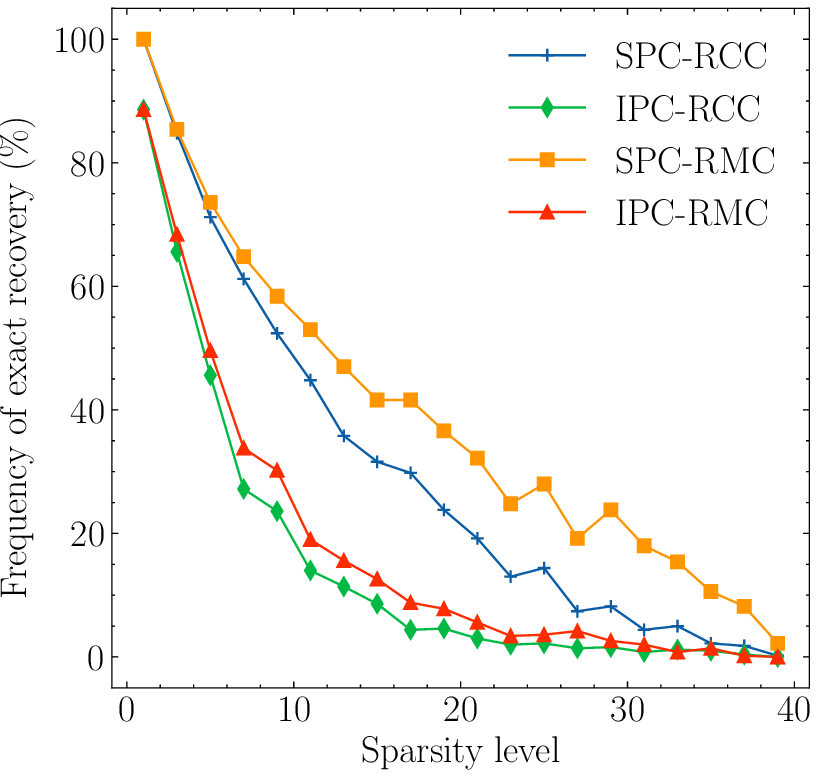}& 
	\includegraphics[width=0.3\textwidth]{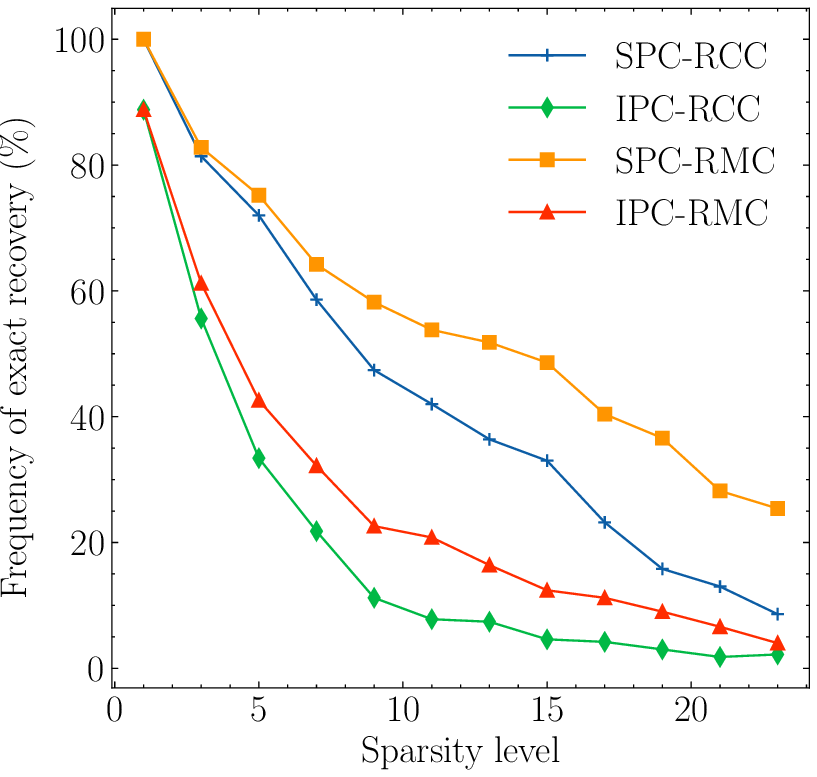}&
	\includegraphics[width=0.3\textwidth]{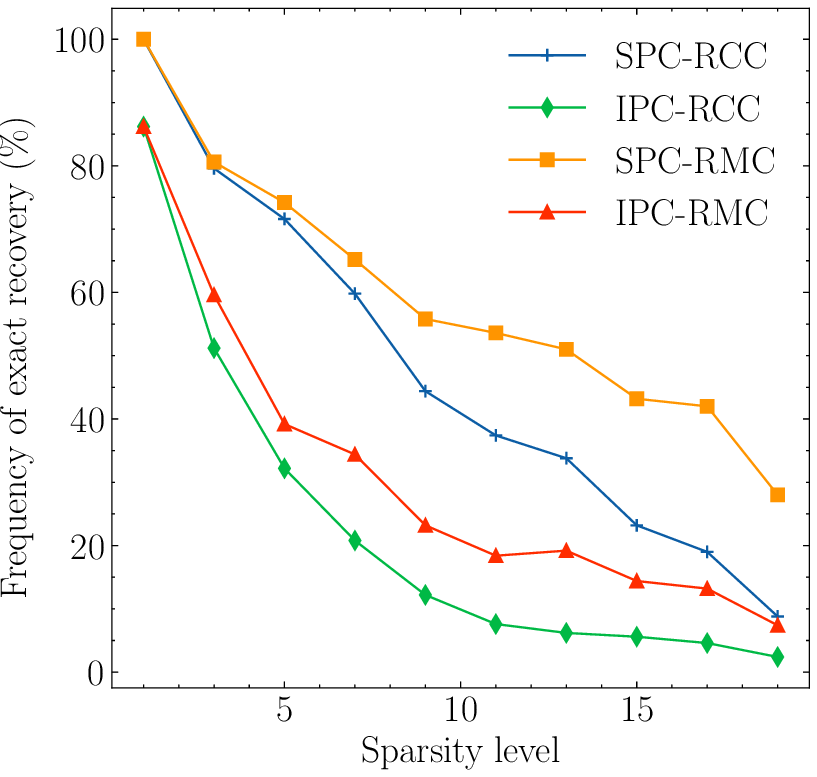}\\\hline
\end{tabular}
\caption{Ablation study of four combinations of two sets of criteria for candidate inclusion and exclusion using sampling matrices consisting of heterogeneous Gaussian blocks. From left to right, each column corresponds to block size $M=5,8,10$, respectively. In particular, the combination of IPC-RCC is BSP, and SPC-RMC corresponds to GPSP. }
\label{fig:exact_ablation}
\end{figure}
\begin{figure}[t!]
\centering
\begin{tabular}{c|c|c}
	\hline
	$M=5$& $M=8$ & $M=10$  \\\hline
	\multicolumn{3}{c}{Data perturbation $\sigma=0.1$}\\\hline
	\includegraphics[width=0.3\textwidth]{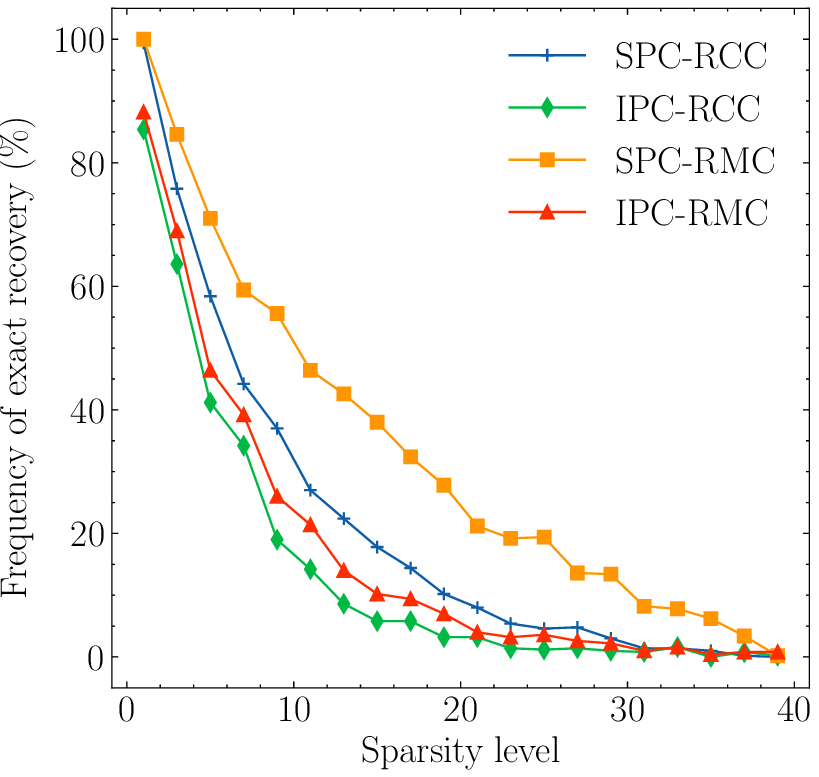}&
	\includegraphics[width=0.3\textwidth]{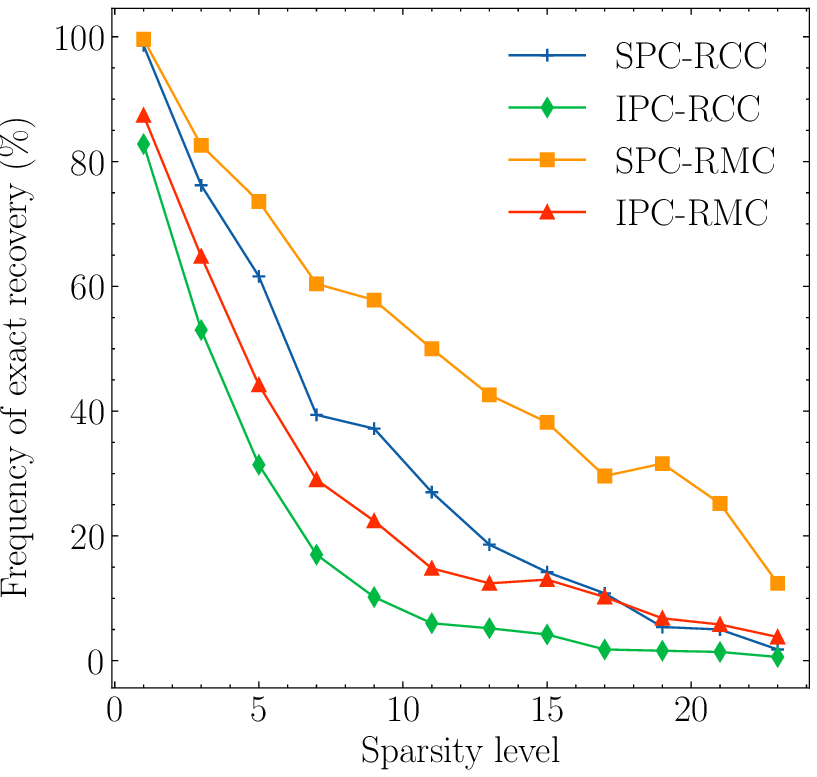}&
	\includegraphics[width=0.3\textwidth]{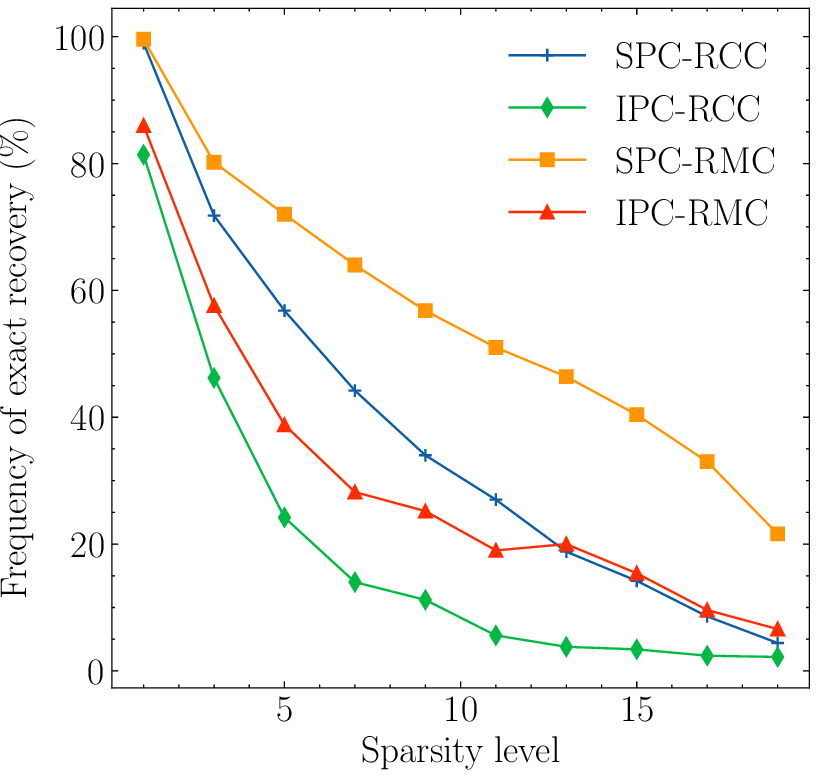}  \\\hline
	\multicolumn{3}{c}{Data perturbation $\sigma=0.5$}\\\hline
	\includegraphics[width=0.3\textwidth]{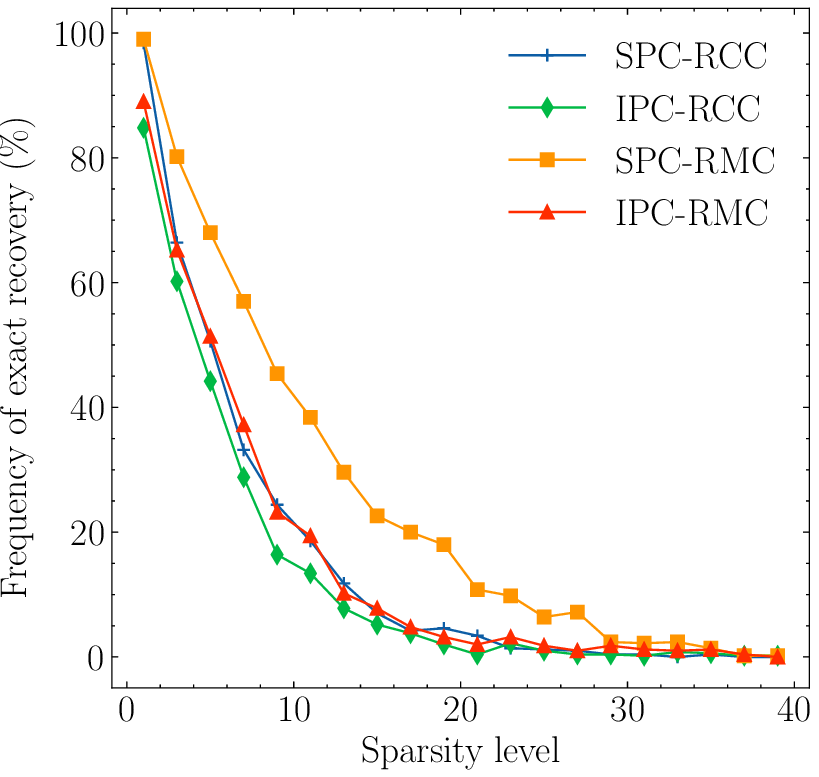}&
	\includegraphics[width=0.3\textwidth]{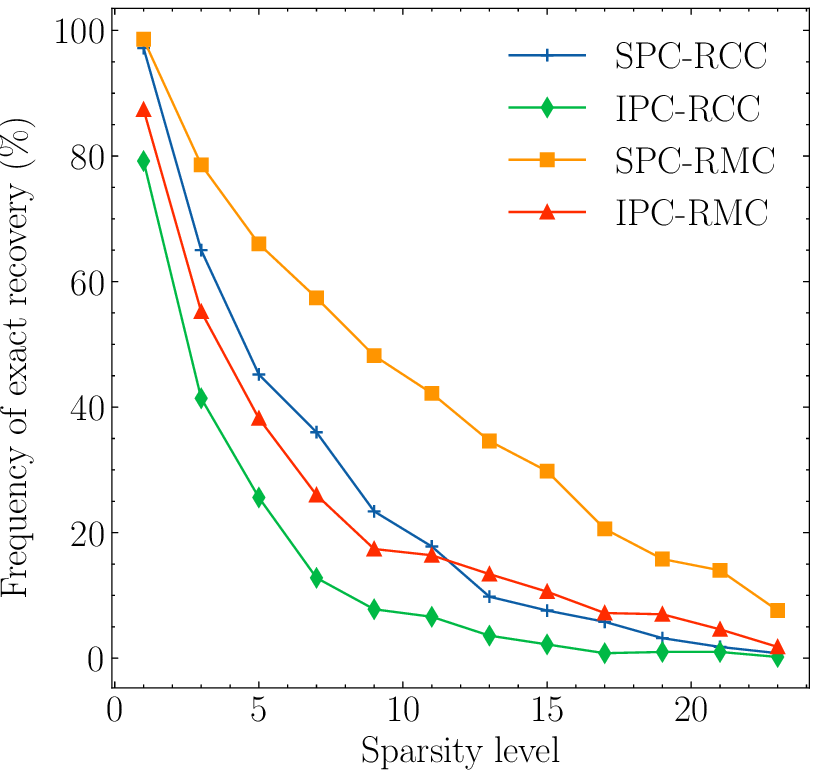}&
	\includegraphics[width=0.3\textwidth]{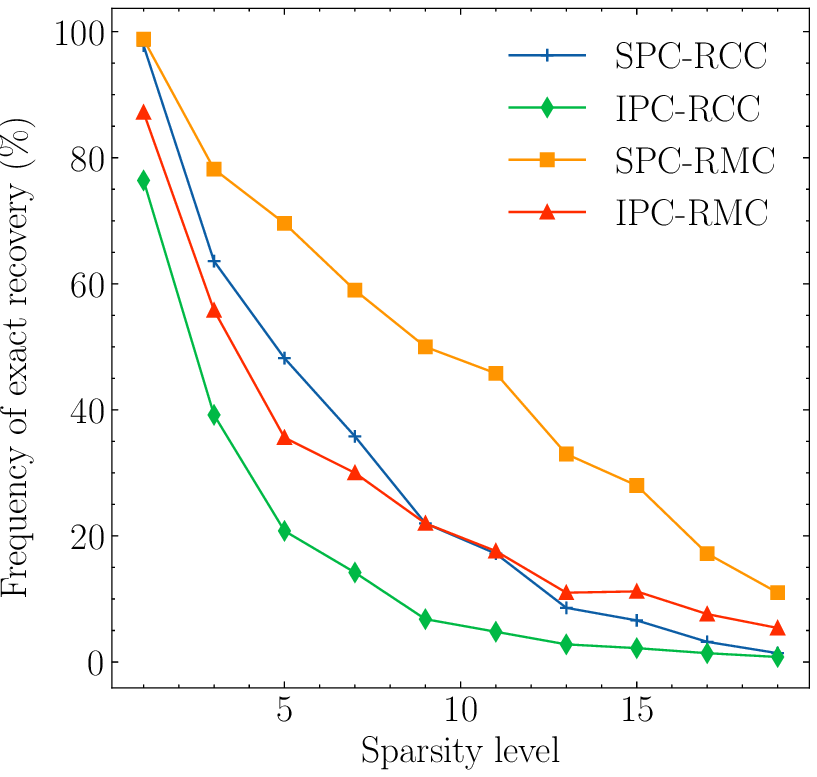} \\\hline
	\multicolumn{3}{c}{Data perturbation $\sigma=1.0$}\\\hline
	\includegraphics[width=0.3\textwidth]{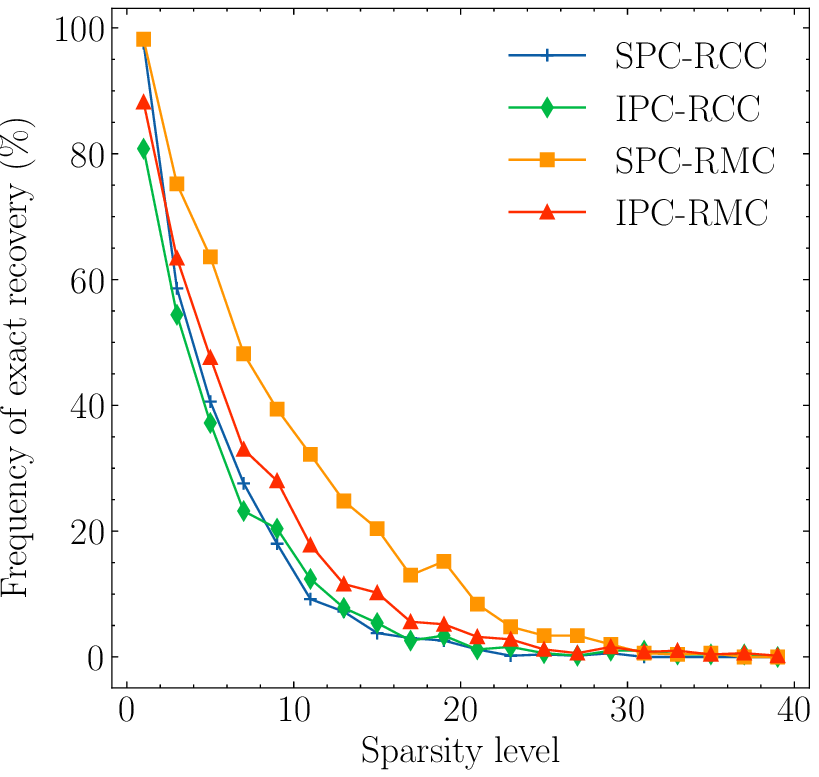}&
	\includegraphics[width=0.3\textwidth]{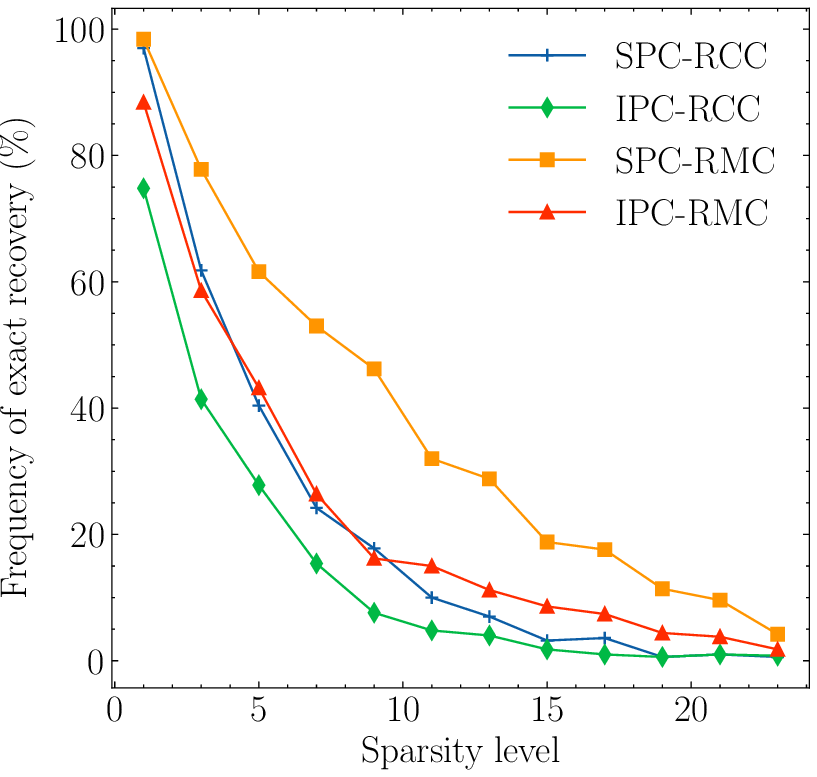}&
	\includegraphics[width=0.3\textwidth]{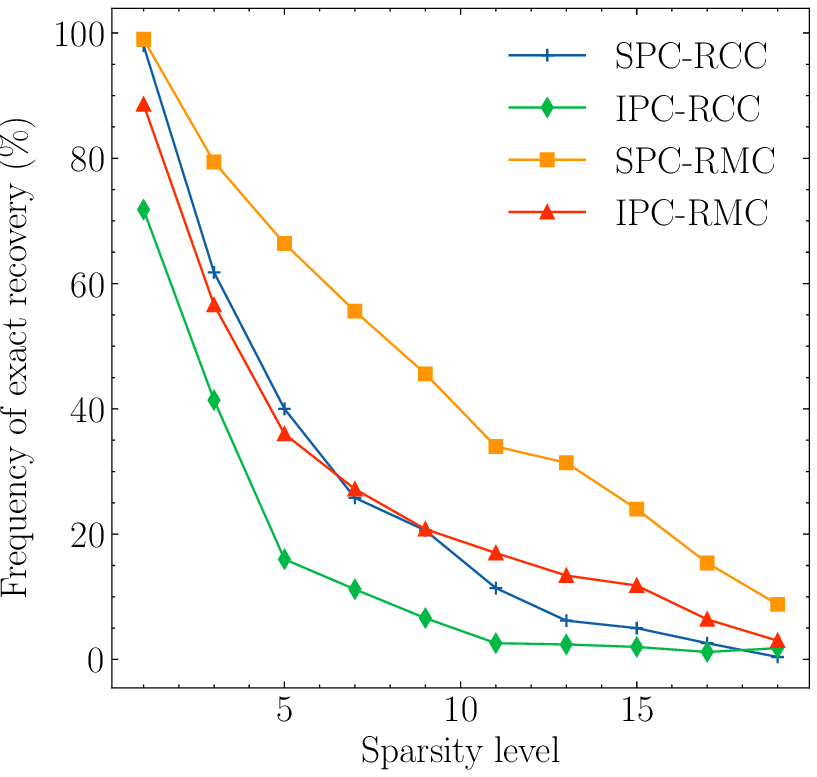}\\\hline
\end{tabular}
\caption{Ablation study of four combinations of two sets of criteria for candidate inclusion and exclusion using sampling matrices consisting of heterogeneous Gaussian blocks. Columns are normalized. The observed responses are perturbed by additive Gaussian noises with mean $0.0$ and standard deviation $\sigma>0$. From left to right, each column corresponds to $M=5,8,10$, respectively. From top to bottom, the rows correspond to $\sigma=0.1,0.5,1.0$. In particular, the combination of IPC-RCC is BSP, and SPC-RMC corresponds to GPSP. }
\label{fig:inexact_ablation}
\end{figure}

\section{Conclusion}
\label{sec.conclusion}
In this paper, we study the convergence analysis of the Group Projected Subspace Pursuit (GPSP) algorithm.
If the sampling matrix satisfies the BRIP condition with a sufficiently small BRIC, GPSP converges to the exact block-sparse solution. When the response signal contains noise, we derive an error bound for the reconstruction and show a sufficient condition for  GPSP to remain convergent. We conduct a comprehensive comparison of the criteria for feature selection and removal between GPSP and other methods. With both theoretical and numerical justifications, we show that the criteria of GPSP are more effective and robust, especially when columns of blocks are not orthogonal and the response signal contains noise. 
A series of numerical experiments are also conducted to verify the superior performances of GPSP under diverse settings and applications.
Tested with heterogeneous random blocks, exact and in-exact data, general types of randomness, face recognition, and PDE identification, GPSP shows the highest success rates of exact recovery for most cases. By a systematic ablation study, we find that the subspace projection criterion (SPC) is the main factor for the success of GPSP, and the response magnitude criterion (RMC) further enhances its performance especially when the data contains noise.  
\section*{Acknowledgement}
The work of Roy Y. He was partially supported by CityU 7200779. The work of H.X. Liu was supported in part by NSFC 11901220, HUST 2024JCYJ005 and Hubei Key Laboratory of Engineering Modeling and Scientific Computing. The work of Hao Liu was partially supported by National Natural Science Foundation of China  12201530, HKRGC ECS 22302123, HKBU 179356.
\appendix
\section*{Appendix}
\section{Uniqueness of block sparse representations}\label{sec::unique_sparse}
The uniqueness of  recovery is critical when reconstructing signals from sampled observations. This often requires that the kernel space of block matrix $\bA$ satisfies certain conditions.

One general characterization involves the compatibility between the kernel space of $\bA$ and the block patterns. Indeed, if $\bA\widetilde{\bc}=\bA\bar{\bc}$ for some $k$-sparse vectors $\widetilde{\bc}\neq \bar{\bc}$, then $\widetilde{\bc}-\bar{\bc}$ is a block $2k$-sparse vector in $\text{ker}(\bA)$; conversely, any non-zero block $2k$-sparse vector $\bc$ admits a decomposition $\bc=\widetilde{\bc}-\bar{\bc}$ where $\widetilde{\bc}\neq\bar{\bc}$, $|\gsup(\widetilde{\bc})|\leq k$ and $|\gsup(\bar{\bc})|\leq k$, and $\bc\in\text{ker}(\bA)$  implies $\bA\widetilde{\bc}=\bA\bar{\bc}$, thus proving the claim.  This observation is also made in~\cite{eldar2009robust}, and we restate it in the following proposition.

\begin{proposition}\label{prop1}
A block $k$-sparse signal $\bc$ is uniquely determined from $\by=\bA\bc$ if and only if the kernel of $\bA$ does not include any block $2k$-sparse vectors.
\end{proposition}

We discuss some implications of Proposition~\ref{prop1}. First, it indicates that columns within each block need to be linearly independent to ensure uniqueness, i.e., blocks are non-redundant. Second, we deduce that  the number of observations $N\geq \sup_{T:|T|\leq 2k}\text{dim}\left(\mathcal{P}_T\right)=2kM$ is necessary for unique recovery; this is an instantiation of the minimal sampling requirement for unique recovery from unions of subspaces~\cite{lu2008theory}. Third, the unique block $k$-sparse solution may have multiple $kM$-sparse representations. Indeed, block $k$-sparse vectors are $kM$ sparse vectors  but with particular placements of non-zero elements. The kernel of $\bA$ being absent of block $2k$-sparse vectors does not mean it does not include $2kM$-sparse vectors.  The number of possible distributions of non-zero elements for a $2kM$-sparse vector in $\mathbb{R}^{GM}$ is $\sum_{m=0}^{2kM}{\tbinom{GM}{m} }$, which is much greater than that of a block $k$-sparse vector $\sum_{s=0}^{2k}{\tbinom{G}{s}}(2^M-1)^s$. This confirms that leveraging the sparsity pattern yields significant dimension reduction.

Block sparse representations can also be unique for redundant blocks. The argument for proving Proposition~\ref{prop1}  remains valid  when focusing on full-rank submatrices of individual blocks, thus we have 
\begin{proposition}\label{prop2}
For a block matrix $\bA=[\bF_1,\dots,\bF_G]\in\mathbb{R}^{N\times GM}$, let $\overline{\bA}=[\overline{\bF}_1,\dots,\overline{\bF}_G]$ be the submatrix of $\bA$ where for each $g=1,\dots,G$, $\overline{\bF}_g$ is a submatrix of $\bF_g$ with $\text{rank}(\overline{\bF}_g)=\text{rank}(\bF_g)$.  A block $k$-sparse solution of $\by=\bA\bc$ is unique if and only if
the kernel of $\overline{\bA}$ does not include any block $2k$-sparse vectors.
\end{proposition}
A detailed proof can be found in~\cite{elhamifar2012block}, where another characterization of uniqueness concerning matrices with almost orthogonal blocks is available. We remark that both Proposition~\ref{prop1} and Proposition~\ref{prop2} provide a macroscopic perspective of the kernel space of $\bA$, and the sparsity level within each block is irrelevant.

Leveraging the characterization for unique sparse representations, we show a sufficient condition for unique block sparse representations. Introduced in the classical work by Donoho and Elad~\cite{donoho2003optimally}, the spark of $\bA$, denoted by $\spark(\bA)$, is defined as the smallest integer $m$ such that there exist $m$ columns from $\bA$ which are linearly dependent. If $A$ has full column rank, define $\spark(\bA)=\rank(\bA)+1$. Employing this concept within each block, we obtain the following result.
\begin{proposition}\label{prop3}
For a block matrix $\bA=[\bF_1,\dots,\bF_G]\in\mathbb{R}^{N\times GM}$ and any positive integer $1\leq k\leq G$, a block $k$-sparse vector $\bc=[\bc_1^\top,\dots,\bc_G^\top]^\top\in\mathbb{R}^{GM}$ is the unique solution of $\by=\bA\bc$ if 
\begin{align}
	|\supp(\bc_g)|\leq \frac{\spark(\bA)-1}{2G}\label{eq_sparse_spark}
\end{align}
for all $g=1,2,\dots,G$
\end{proposition}
\begin{proof}
For any $k=1,2,\dots,G$, if $\bc$ and $\bc'$ are two block $k$-sparse vectors satisfying~\eqref{eq_sparse_spark} and $\bA \bc = \bA\bc'$, then for any $g=1,2,\dots,G$, either $\bc_g=\bc_g'$, or $\bc_g-\bc_g'\neq \mathbf{0}$ and is at most $\lfloor(\spark(\bA)-1)/G\rfloor$-sparse. Hence, $\bc-\bc'$ is at most $(\spark(\bA)-1)$-sparse. By the definition of spark, this is impossible since $\bA(\bc-\bc')=\mathbf{0}$.  Thus we have $\bc=\bc'$.
\end{proof}

We note that~\eqref{eq_sparse_spark} provides a sufficient condition for unique solutions of any levels of block sparsity.  Unlike Proposition~\ref{prop1}, each block $\bF_g$ can have linearly dependent columns; as long as the true signal has sufficiently many zeros in each block, it is unique.

To sum up, the uniqueness of block sparse representation can be developed by restricting the kernel space of $\bA$  according to the blocks or the co-linearity patterns among columns within each block. 
\section{Proof of Theorem \ref{thm1}}\label{app:thm1}
We first define some notations:
For any non-empty $T\subset\{1,2,\dots,G\}$, we have $\bA_T\in \mathbb{R}^{N\times |T|M}$ and $\bv_T\in \mathbb{R}^{|T|M}$ for any $\bA\in\mathbb{R}^{N\times GM}$ and $\bv\in\mathbb{R}^{GM}$. We define the pull back operator: 
\begin{align*}
&\pull(\widetilde{\bB},T)\in \mathbb{R}^{N\times GM} \ \mbox{ with } \ (\pull(\widetilde{\bB},T))_{T(i)}=(\widetilde{\bB})_i \mbox{ for } i=1,...,|T|, \ \mbox{ and } \ (\pull(\widetilde{\bB},T))_{g'}=\mathbf{0} \ \mbox{ for } g'\notin T,\\
&\pull(\widetilde{\bv},T)\in \mathbb{R}^{GM} \ \mbox{ with } \ (\pull(\widetilde{\bv},T))_{T(i)}=(\widetilde{\bv})_i \mbox{ for } i=1,...,|T|, \ \mbox{ and } \ (\pull(\widetilde{\bv},T))_{g'}=0 \ \mbox{ for } g'\notin T,
\end{align*}
where $T(i)$ denotes the $i$-th element in $T$. Generally speaking, the pull back operator pulls a sub-matrix or a sub-vector to their original dimension using indices in $T$ by filling zeros. Given two sets of indices $T_1,T_2$, we use the notation $T_1-T_2=\{g: g\in T_1, \ g\notin T_2\}, \ T_1+T_2=T_1\cup T_2$.

\begin{proof}[Proof of Theorem \ref{thm1}]
Note that
\begin{align}
	\by_r^{l-1}&=\resid(\by,\bA_{T^{l-1}}) \nonumber\\
	&=\resid(\bA_{T^*-T^{l-1}}\bc^*_{T^*-T^{l-1}},\bA_{T^{l-1}})+\resid(\bA_{T^*\cap T^{l-1}}\bc^*_{T^*\cap T^{l-1}},\bA_{T^{l-1}})\nonumber\\
	&=\resid(\bA_{T^*-T^{l-1}}\bc^*_{T^*-T^{l-1}},\bA_{T^{l-1}})\nonumber\\
	&=\bA_{T^*-T^{l-1}}\bc^*_{T^*-T^{l-1}} -\text{proj}(\bA_{T^*-T^{l-1}}\bc^*_{T^*-T^{l-1}},\bA_{T^{l-1}})\nonumber\\
	&=\begin{bmatrix}
		\bA_{T^*-T^{l-1}}&\bA_{T^{l-1}}
	\end{bmatrix}\begin{bmatrix}
		\bc^*_{T^*-T^{l-1}}\\
		\bc_{p,T^{l-1}}^*
	\end{bmatrix} \nonumber\\
	&=\bA_{T^*\cup T^{l-1}}\bx_r^{l-1},\label{eq_4}
\end{align}
where notation $\bc_{p,T^{l-1}}^*=-(\bA^{\top}_{T^{l-1}}\bA_{T^{l-1}})^{-1}\bA^{\top}_{T^{l-1}}(\bA_{T^*-T^{l-1}}\bc_{T^*-T^{l-1}})$ is used in the second last equality, and notation $\bx_r^{l-1}=\bz^{l-1}_{T^*\cup T^{l-1}}$ with 
$$
\bz^{l-1}\in \mathbb{R}^{GM}, \quad \bz=\begin{cases}
	\bc^*_g & \mbox{ if } g\in T^*-T^{l-1},\\
	(\pull(\bc^*_{p,T^{l-1}},T^{l-1}))_g & \mbox{ if } g\in T^{l-1},\\
	\mathbf{0} & \mbox{ otherwise}
\end{cases}
$$
is used in the last equality. 
Denote $T^l_\Delta =\wT^l-T^{l-1}$. Since $|T_\Delta^l|=|T^*|$, we have
\begin{align}
	|T_\Delta^l-(T_\Delta^l\cap T^*)|=|T^*-(T_\Delta^l\cap T^*)|.
\end{align}
Moreover, since $T^*-\wT^l\subset T^*-(T_\Delta^l\cap T^*)$, we have
\begin{align}
	|T_\Delta^l-(T_\Delta^l\cap T^*)|=|T^*-(T_\Delta^l\cap T^*)|\geq |T^*-\wT^l|.\label{eq_th1_main_inequality}
\end{align}
Therefore, according to GPSP (line 5), we arrive at 
\begin{align}
	\sqrt{\sum_{g\in T_\Delta^l-T^*}\left(\|\proj(\by_r^{l-1},\bF_g)\|_2\right)^2}\geq \sqrt{\sum_{g'\in T^*-\wT^{l} }\left(\|\proj(\by_r^{l-1},\bF_{g'})\|_2\right)^2}.\label{eq_5}
\end{align}
Notice that in the $l$-th iteration, $T_\Delta^l$ consists of the newly selected $k$ groups with highest values of projection according to the line 5 in Algorithm~\ref{alg_GPSP}, and among which, the true groups are specified by indices in $T^l_{\Delta}\cap T^*$; thus the set $ T^l_{\Delta}-T^*$ on the left-hand side of (\ref{eq_5}) contains the indices of newly selected groups which are false. On the right-hand side of (\ref{eq_5}), $T^*-T^{l-1}$ contains the indices of true groups that are not included in the previous iteration. 
Observe that for any $\bv\in\mathbb{R}^N$, we have
\begin{align}
	\|\proj(\bv,\bF_g)\|^2_2=&\|\bF_g(\bF_g^\top\bF_g)^{-1}\bF_g^\top\bv\|^2_2= |\bv^{\top}\bF_g(\bF_g^{\top}\bF_g)^{-1}\bF^{\top}_g\bv| \nonumber\\
	\leq &\frac{1}{\lambda_{\min}(\bF_g^\top\bF_g)}\|\bF_g^\top\bv\|_2^2 
	\leq \frac{1}{1-\delta_{M,1}}\|\bF_g^\top\bv\|_2^2,
	\label{eq.proj.ineq}
\end{align}
where $\lambda_{\min}(\cdot)$ denotes the smallest singular value of a matrix, and the last inequality follows from the BRIP condition~\eqref{eq_BRIP_condition}. We thus deduce
\begin{align}
	&\sum_{g\in T_\Delta^l-T^*}\|\text{proj}(\bA_{T^*\cup T^{l-1}}(\bx_r^{l-1})_{T^*\cup T^{l-1}},\bF_g)\|^2_2 \nonumber\\
	\leq &\sum_{g\in T_\Delta^l-T^*}\frac{1}{1-\delta_{M,1}}\|\bF_g^\top\bA_{T^*\cup T^{l-1}}(\bx_r^{l-1})_{T^*\cup T^{l-1}}\|^2_2 \nonumber\\
	=&\frac{1}{1-\delta_{M,1}}\|\bA_{T_\Delta^l-T^*}^\top\bA_{T^*\cup T^{l-1}}(\bx_r^{l-1})_{T^*\cup T^{l-1}}\|^2_2 \nonumber\\
	\le&\frac{2}{1-\delta_{M,1}}(\|\bA_{T_\Delta^l-T^*}^\top\bA_{T^*}(\bx_r^{l-1})_{T^*}\|^2_2+\|\bA_{T_\Delta^l-T^*}^\top\bA_{ T^{l-1}-T^*}(\bx_r^{l-1})_{ T^{l-1}-T^*}\|^2_2) \nonumber\\
	\leq&\frac{2\delta_{M,2k}^2}{1-\delta_{M,1}}\left(\|(\bx_r^{l-1})_{T^*}\|^2_2+\|(\bx_r^{l-1})_{T^{l-1}-T^*}\|^2_2\right) 
	\le\frac{2\delta_{M,2k}^2}{1-\delta_{M,1}}\|\bx_r^{l-1}\|^2_2,\label{eq_upper}
\end{align}
where the first inequality follows from (\ref{eq.proj.ineq}), the second inequality follows from Lemma \ref{lemma_inner}.

For the right-hand side of (\ref{eq_5}), observe that for any $\bv\in\mathbb{R}^N$, we have
\begin{align}
	\|\proj(\bv,\bF_g)\|^2_2=&\|\bF_g(\bF_g^\top\bF_g)^{-1}\bF_g^\top\bv\|^2_2= |\bv^{\top}\bF_g(\bF_g^{\top}\bF_g)^{-1}\bF^{\top}_g\bv| \nonumber\\
	\geq& \frac{1}{\lambda_{\max}(\bF_g^\top\bF_g)}\|\bF_g^\top\bv\|_2^2 
	\geq \frac{1}{1+\delta_{M,1}}\|\bF_g^\top\bv\|_2^2,
	\label{eq.proj.ineq2}
\end{align}
where the last inequality follows from the BRIP condition~\eqref{eq_BRIP_condition}.
We have
\begin{align*}
	&\sqrt{\sum_{g\in T^*-\tilde T^l}\|\text{proj}(\by^{l-1}_r,\bF_g)\|^2_2}
	=\sqrt{\sum_{g\in T^*-\tilde T^l}|(\by^{l-1}_r)^{\top} \bF_g(\bF^\top_g\bF_g)^{-1}\bF^\top_g\by^{l-1}_r|}\\
	\ge&\frac{1}{\sqrt{1+\delta_{M,1}}}\sqrt{\sum_{g\in T^*-\tilde T^l}\|\bF^\top_g\by^{l-1}_r\|^2_2}
	=\frac{1}{\sqrt{1+\delta_{M,1}}}\|\bA^\top_{T^*-\tilde T^l}\bA_{T^*\cup T^{l-1}}(\bx^{l-1}_r)_{T^*\cup T^{l-1}}\|_2\\
	\ge&\frac{1}{\sqrt{1+\delta_{M,1}}}\left(\|\bA^\top_{T^*-\tilde T^l}\bA_{T^*-\tilde T^l}(\bx^{l-1}_r)_{T^*-\tilde T^l}\|_2-\|\bA^\top_{T^*-\tilde T^l}\bA_{(T^*\cap \wT^{l})\cup T^{l-1}}(\bx^{l-1}_r)_{(T^*\cap \wT^{l})\cup T^{l-1}}\|_2\right),
\end{align*}
where (\ref{eq.proj.ineq2}) is used in the first inequality and relations $T^*\cup T^{l-1}=(T^*-\wT^{l})\cup((T^*\cap \wT^{l})\cup T^{l-1})$ and $(T^*-\wT^{l})\cap((T^*\cap \wT^{l})\cup T^{l-1})=\varnothing$ are used in the last inequality.

Since
\begin{align*}
	&\|\bA^\top_{T^*-\tilde T^l}\bA_{T^*-\tilde T^l}(\bx^{l-1}_r)_{T^*-\tilde T^l}\|_2\geq  (1-\delta_{M,k})\|(\bx^{l-1}_r)_{T^*-\tilde T^l}\|_2
\end{align*}
according to \eqref{eq_BRIP_condition} and 
\begin{align*}
	&\|\bA^\top_{T^*-\tilde T^l}\bA_{(T^*\cap \wT^{l})\cup T^{l-1}}(\bx^{l-1}_r)_{(T^*\cap \wT^{l})\cup T^{l-1}}\|_2\leq \delta_{M,2k}\|(\bx^{l-1}_r)_{(T^*\cap \wT^{l})\cup T^{l-1}}\|_2 \leq \delta_{M,2k}\|\bx^{l-1}_r\|_2
\end{align*}
according to Lemma \ref{lemma_inner}, we have
\begin{align}
	&\sqrt{\sum_{g\in T^*-\tilde T^l}\|\text{proj}(\by^{l-1}_r,\bF_g)\|^2_2}\geq \frac{1}{\sqrt{1+\delta_{M,1}}}\left((1-\delta_{M,k})\|(\bx^{l-1}_r)_{T^*-\tilde T^l}\|_2-\delta_{M,2k}\|\bx^{l-1}_r\|_2\right).
	\label{eq_lower}
\end{align}

Combining \eqref{eq_5} with \eqref{eq_upper}  and \eqref{eq_lower} gives rise to
\begin{align}
	\frac{\delta_{M,2k}}{1-\delta_{M,k}}\left(\sqrt{\frac{2(1+\delta_{M,1})}{1-\delta_{M,1}}}+1\right)\|\bx_r^{l-1}\|_2\geq \|(\bx^{l-1}_r)_{T^*-\tilde T^l}\|_2.\label{eq_combined_ineq}
\end{align}

Note that 
\begin{align}
	\|\bx_r^{l-1}\|_2 \leq  \|\bc_{T^*-T^{l-1}}^*\|_2+\|\bc_{p,T^{l-1}}^*\|_2.
	\label{eq.xr.decom}
\end{align}
To derive an upper bound of $\|\bx_r^{l-1}\|_2$ in terms of $\|\bc_{T^*-T^{l-1}}^*\|_2$, we first derive an upper bound of $\|\bc_{p,T^{l-1}}^*\|_2$:
\begin{align}
	\|\bc^*_{p,T^{l-1}}\|_2= &\|-(\bA^{\top}_{T^{l-1}}\bA_{T^{l-1}})^{-1}\bA^{\top}_{T^{l-1}}(\bA_{T^*-T^{l-1}}\bc_{T^*-T^{l-1}})\|_2 \nonumber\\
	\leq & \frac{1}{1-\delta_{M,k}}\|\bA^{\top}_{T^{l-1}}\bA_{T^*-T^{l-1}}\bc_{T^*-T^{l-1}}\|_2 
	\leq \frac{\delta_{M,2k}}{1-\delta_{M,k}}\|\bc^*_{T^*-T^{l-1}}\|_2,
	\label{eq.cpt}
\end{align}
where the first inequality uses (\ref{eq.proj.ineq}) and the second inequality uses Lemma~\ref{lemma_inner}. 

Substituting (\ref{eq.cpt}) into (\ref{eq.xr.decom}) gives rise to
\begin{align}
	\|\bx_r^{l-1}\|_2\leq\left(1+\frac{\delta_{M,2k}}{1-\delta_{M,k}}\right)\|\bc_{T^*-T^{l-1}}^*\|_2= \frac{1-\delta_{M,k}+\delta_{M,2k}}{1-\delta_{M,k}}\|\bc_{T^*-T^{l-1}}^*\|_2.
	\label{eq_result1}
\end{align}
Since $T^{l-1}\subset\wT^l$, we have
\begin{align}
	(\bx_r^{l-1})_{T^*-\wT^{l}}=\bc^*_{T^*-\wT^l}.\label{eq_result2}
\end{align}
Therefore
\begin{align*}
	\|\bc^*_{T^*-\wT^l}\|_2=&\|(\bx_r^{l-1})_{T^*-\wT^{l}}\|_2
	\leq \frac{\delta_{M,2k}}{1-\delta_{M,k}}\left(\sqrt{\frac{2(1+\delta_{M,1})}{1-\delta_{M,1}}}+1\right)\|\bx_r^{l-1}\|_2\\
	\leq &\frac{\delta_{M,2k}(1-\delta_{M,k}+\delta_{M,2k})}{(1-\delta_{M,k})^2}\left(\sqrt{\frac{2(1+\delta_{M,1})}{1-\delta_{M,1}}}+1\right)\|\bc_{T^*-T^{l-1}}^*\|_2,
\end{align*}
where the first inequality follows from \eqref{eq_combined_ineq} and the second inequality follows from \eqref{eq_result1}. The theorem is proved.

\end{proof}
\section{Proof of Theorem~\ref{thm2}}\label{app:thm2}
\begin{proof}[Proof of Theorem~\ref{thm2}]
Our proof is based on the following inequality
\begin{align}
	\|\bc^*_{T^*-T^l}\|_2&\leq \|\bc^*_{T^*\cap(\wT^l-T^l)}\|_2+\|\bc^*_{T^*-\wT^l}\|_2.
	\label{eq.bc.ine}
\end{align}
We will derive an upper bound of $\|\bc^*_{T^*\cap(\wT^l-T^l)}\|_2$ in terms of $\|\bc^*_{T^*-\wT^l}\|_2$.

By construction in Algorithm \ref{alg_GPSP}, we have
\begin{align*}
	\bx_p^l = \bA_{\wT^l}^\dagger\by &=\bA_{\wT^l}^\dagger\bA_{T^*}\bc^*_{T^*}=\bA_{\wT^l}^\dagger\bA_{T^*\cap\wT^l}\bc^*_{T^*\cap\wT^l}+\bA_{\wT^l}^\dagger\bA_{T^*-\wT^l}\bc^*_{T^*-\wT^l}\\
	&=\bA_{\wT^l}^\dagger\bA_{\wT^l}\bc^*_{\wT^l}+\bA_{\wT^l}^\dagger\bA_{T^*-\wT^l}\bc^*_{T^*-\wT^l}=\bc^*_{\wT^l}+\bA_{\wT^l}^\dagger\bA_{T^*-\wT^l}\bc^*_{T^*-\wT^l}.
\end{align*}
Denote $T^l_\Delta =\wT^l-T^{l-1}$. We have
\begin{align}
	\|\bx_p^l-\bc^*_{\wT^l}\|^2_2 =& \|\bA_{\wT^l}^\dagger\bA_{T^*-\wT^l}\bc^*_{T^*-\wT^l}\|^2_2 
	= \|(\bA_{\wT^l}^\top\bA_{\wT^l})^{-1}\bA_{\wT^l}^\top\bA_{T^*-\wT^l}\bc^*_{T^*-\wT^l}\|^2_2 \nonumber\\
	\leq& \frac{1}{(1-\delta_{M,2k})^2}\|\bA_{\wT^l}^\top\bA_{T^*-\wT^l}\bc^*_{T^*-\wT^l}\|^2_2 \nonumber\\
	\le&\frac{1}{(1-\delta_{M,2k})^2}\left( \|\bA_{T^l_{\Delta}}^{\top}\bA_{T^*-\wT^l}\bc^*_{T^*-\wT^l}\|^2_2+ \|\bA_{T^{l-1}}^{\top}\bA_{T^*-\wT^l}\bc^*_{T^*-\wT^l}\|^2_2\right) \nonumber\\
	\leq&\frac{2\delta_{M,2k}^2}{(1-\delta_{M, 2k})^2}\|\bc^*_{T^*-\wT^l}\|^2_2,
	\label{eq_temp6}
\end{align}
where the first inequality uses (\ref{eq.proj.ineq}), the second inequality follows from $\widetilde{T}=T_{\Delta}^l+T^{l-1}$ and $T_{\Delta}^l\cap T^{l-1}=\varnothing$, the last inequality uses Lemma \ref{lemma_inner}.
Therefore,
\begin{align}
	\|\bx_p^l-\bc^*_{\wT^l}\|_2\le\frac{\sqrt{2}\delta_{M,2k}}{1-\delta_{M, 2k}}\|\bc^*_{T^*-\wT^l}\|_2.
	\label{eq.xplc*.upper}
\end{align}
Denote a non-empty subset $T'\subset\wT^l$ such that $T^*\cap T'=\varnothing$ and $|T'|=k$.  We have
\begin{align}
	\left(\pull\left(\bx^l_p,\widetilde{T}^l\right)\right)_{T'}=\left(\pull\left(\bx^l_p,\widetilde{T}^l\right)-\bc^*\right)_{T'}+\bc^*_{T'}=\left(\pull\left(\bx^l_p,\widetilde{T}^l\right)-\bc^*\right)_{T'}
\end{align}
and thus 
\begin{align}
	\left\|\left(\pull\left(\bx^l_p,\widetilde{T}^l\right)\right)_{T'}\right\|_2\leq \left\|\bx^l_p-\bc^*_{\wT^l}\right\|_2.
	\label{eq.xplc*}
\end{align}

According to GP-SP, we have
\begin{align}
	&\left\|\left(\pull\left(\bx^l_p,\widetilde{T}^l\right)\right)_{\wT^l-T^l}\right\|^2_2=\sum_{g\in\wT^l-T^l}\left\|\left(\pull\left(\bx^l_p,\widetilde{T}^l\right)\right)_{g}\right\|^2_2 \nonumber\\
	\leq& \frac{1}{1-\delta_{M,1}}\sum_{g\in\wT^l-T^l}\left\|\bF_g\left(\pull\left(\bx^l_p,\widetilde{T}^l\right)\right)_{g}\right\|^2_2
	\leq \frac{1}{1-\delta_{M,1}}\sum_{g\in T'}\left\|\bF_g\left(\pull\left(\bx^l_p,\widetilde{T}^l\right)\right)_{g}\right\|^2_2 \nonumber\\
	\leq &\frac{1+\delta_{M,1}}{1-\delta_{M,1}}\sum_{g\in T'}\left\|\left(\pull\left(\bx^l_p,\widetilde{T}^l\right)\right)_{g}\right\|^2_2 
	=\frac{1+\delta_{M,1}}{1-\delta_{M,1}}\left\|\left(\pull\left(\bx^l_p,\widetilde{T}^l\right)\right)_{T'}\right\|_2^2 
	\leq  \frac{1+\delta_{M,1}}{1-\delta_{M,1}}\left\|\bx^l_p-\bc^*_{\wT^l}\right\|_2^2,
	\label{eq_temp4}
\end{align}
where the first and third inequality comes from (\ref{eq_BRIP_condition}), the second inequality follows from the criterion in selecting $T^l$ in Algorithm \ref{alg_GPSP} line 7, the last inequality follows from (\ref{eq.xplc*}).

Again for $\left\|\left(\pull\left(\bx^l_p,\widetilde{T}^l\right)\right)_{\wT^l-T^l}\right\|_2$, we have
\begin{align}
	\left\|\left(\pull\left(\bx^l_p,\widetilde{T}^l\right)\right)_{\wT^l-T^l}\right\|_2&=\left\|\left(\pull\left(\bx^l_p-\bc^*_{\wT^l}+\bc^*_{\wT^l},\widetilde{T}^l\right)\right)_{\wT^l-T^l}\right\|_2 \nonumber\\
	&=\left\|\left(\pull\left(\bx^l_p-\bc^*_{\wT^l},\widetilde{T}^l\right)\right)_{\wT^l-T^l}+\bc^*_{\wT^l-T^l}\right\|_2 \nonumber\\
	&\geq \|\bc^*_{\wT^l-T^l}\|_2-\left\|\left(\pull\left(\bx^l_p-\bc^*_{\wT^l},\widetilde{T}^l\right)\right)_{\wT^l-T^l}\right\|_2 \nonumber\\
	&\geq \|\bc^*_{\wT^l-T^l}\|_2-\|\bx_p^l-\bc^*_{\wT^l}\|_2.
	\label{eq.xpl.lower}
\end{align}
Putting \eqref{eq_temp4} and (\ref{eq.xpl.lower}) together, we get
\begin{align*}
	(\sqrt{1+\delta_{M,1}}+\sqrt{1-\delta_{M,1}})\|\bx_p^l-\bc^*_{\wT^l}\|_2\geq \sqrt{1-\delta_{M,1}}\|\bc^*_{\wT^l-T^l}\|_2.
\end{align*}
Furthermore, we have
\begin{align}
	\|\bc^*_{T^*\cap(\wT^l-T^l)}\|_2\leq\|\bc^*_{(\wT^l-T^l)}\|_2\leq \left(1+\sqrt{\frac{1+\delta_{M,1}}{1-\delta_{M,1}}}\right)\|\bx_p^l-\bc_{\wT^l}^*\|_2.
	\label{eq_temp5}
\end{align}
Substituting (\ref{eq_temp5}) into (\ref{eq.bc.ine}) gives
\begin{align}
	\|\bc^*_{T^*-T^l}\|_2&\leq \|\bc^*_{T^*\cap(\wT^l-T^l)}\|_2+\|\bc^*_{T^*-\wT^l}\|_2 \nonumber\\
	&\leq \left(1+\sqrt{\frac{1+\delta_{M,1}}{1-\delta_{M,1}}}\right)\|\bx_p^l-\bc_{\wT^l}^*\|_2+\|\bc^*_{T^*-\wT^l}\|_2 \nonumber\\
	&\leq\left(1+ \frac{\sqrt{2}\delta_{M,2k}}{1-\delta_{M, 2k}}\left(1+\sqrt{\frac{1+\delta_{M,1}}{1-\delta_{M,1}}}\right)\right)\|\bc^*_{T^*-\wT^l}\|_2,
\end{align}
where the last inequality follows from (\ref{eq.xplc*.upper}).
\end{proof}
\section{Proof of Theorem~\ref{thm3}}\label{app:thm3}
\begin{proof}[Proof of Theorem~\ref{thm3}]
For $\|\by_r^l\|_2$, according to Algorithm \ref{alg_GPSP} (line 8),
\begin{align}
	\|\by_r^l\|_2&=\|\text{resid}(\by,\bA_{T^l})\|_2 \nonumber\\
	&=\|\text{resid}(\bA_{T^*-T^l}\bc^*_{T^*-T^l},\bA_{T^l})+\text{resid}(\bA_{T^*\cap T^l}\bc^*_{T^*\cap T^l},\bA_{T^l})\|_2 \nonumber\\
	&=\|\text{resid}(\bA_{T^*-T^l}\bc^*_{T^*-T^l},\bA_{T^l})\|_2 
	\leq \|\bA_{T^*-T^l}\bc^*_{T^*-T^l}\|_2 \nonumber\\
	&\leq\sqrt{1+\delta_{M,k}}\|\bc^*_{T^*-T^l}\|_2 
	\leq\mu_k\sqrt{1+\delta_{M,k}}\|\bc^*_{T^*-\wT^l}\|_2 \leq \mu_k\beta_k\sqrt{1+\delta_{M,k}}\|\bc^*_{T^*-T^{l-1}}\|_2,
	\label{eq_temp7}
\end{align}
where the second inequality comes from \eqref{eq_BRIP_condition}, the third and fourth inequality follows from Theorem \ref{thm2} and Theorem \ref{thm1}, respectively.

For $\|\by_r^{l-1}\|_2$, we deduce
\begin{align}
	\|\by_r^{l-1}\|_2&=\|\text{resid}(\by,\bA_{T^{l-1}})\|_2 =\|\text{resid}(\bA_{T^*-T^{l-1}}\bc^*_{T^*-T^{l-1}},\bA_{T^{l-1}})\|_2 \nonumber\\
	&\geq\frac{1-\delta_{M,k}-\delta_{M,2k}}{1-\delta_{M,k}}\|\bA_{T^*-T^{l-1}}\bc^*_{T^*-T^{l-1}}\|_2 \geq \frac{1-\delta_{M,k}-\delta_{M,2k}}{\sqrt{1-\delta_{M,k}}}\|\bc^*_{T^*-T^{l-1}}\|_2,
	\label{eq.yl-1}
\end{align}
where the first inequality follows from Lemma \ref{lemma_proj}, the second inequality follows from \eqref{eq_BRIP_condition}.

Putting \eqref{eq_temp7} and (\ref{eq.yl-1}) together gives rise to
\begin{align*}
	\|\by_r^l\|_2\leq\frac{\mu_k\beta_k\sqrt{1-\delta^2_{M,k}} }{1-\delta_{M,k}-\delta_{M,2k}}\|\by_r^{l-1}\|_2.
\end{align*}
\end{proof}

\section{Proof of Theorem \ref{thm.inaccurate}}\label{app:inacc}

\begin{proof}[Proof of Theorem \ref{thm.inaccurate}]
We deduce
\begin{align*}
	\|\bc^*-\widehat{\bc}\|_2&\leq \|\bc^*_{\widehat{T}}-\bA^\dagger_{\widehat{T}}\by\|_2+\|\bc^*_{T-\widehat{T}}\|_2\\
	&\leq\|\bc^*_{\widehat{T}}-\bA^\dagger_{\widehat{T}}(\bA_T\bc^*_T+\be)\|_2+\|\bc^*_{T-\widehat{T}}\|_2\\
	&\leq\|\bc^*_{\widehat{T}}-\bA^\dagger_{\widehat{T}}\bA_T\bc^*_T\|_2+\|\bA_{\widehat{T}}^\dagger\be\|_2+\|\bc^*_{T-\widehat{T}}\|_2\\
	&=\|\bc^*_{\widehat{T}}-\bA^\dagger_{\widehat{T}}\bA_T\bc^*_T\|_2+\|(\bA_{\widehat{T}}^{\top}\bA_{\widehat{T}})^{-1}\bA_{\widehat{T}}^{\top}\be\|_2+\|\bc^*_{T-\widehat{T}}\|_2\\
	&\overset{\eqref{eq_BRIP_condition}}{\leq}\|\bc^*_{\widehat{T}}-\bA^\dagger_{\widehat{T}}\bA_T\bc^*_T\|_2+\frac{1}{\sqrt{1-\delta_{M,k}}}\|\bA_{\widehat{T}}(\bA_{\widehat{T}}^{\top}\bA_{\widehat{T}})^{-1}\bA_{\widehat{T}}^{\top}\be\|_2+\|\bc^*_{T-\widehat{T}}\|_2\\
	&=\|\bc^*_{\widehat{T}}-\bA^\dagger_{\widehat{T}}\bA_T\bc^*_T\|_2+\frac{1}{\sqrt{1-\delta_{M,k}}}\|\text{proj}(\be,\bA_{\widehat{T}})\|_2+\|\bc^*_{T-\widehat{T}}\|_2\\
	&=\|(A^T_{\hat T}A_{\hat T})^{-1}A^T_{\hat T}(A_{\hat T}\bc^*_{\hat T}-A_T\bc^*_T)\|_2+\frac{1}{\sqrt{1-\delta_{M,k}}}\|\text{proj}(\be,\bA_{\widehat{T}})\|_2+\|\bc^*_{T-\widehat{T}}\|_2\\
	&=\|(A^T_{\hat T}A_{\hat T})^{-1}A^T_{\hat T}(A_{\hat T\cap T}\bc^*_{\hat T\cap T}-A_T\bc^*_T)\|_2+\frac{1}{\sqrt{1-\delta_{M,k}}}\|\text{proj}(\be,\bA_{\widehat{T}})\|_2+\|\bc^*_{T-\widehat{T}}\|_2\\
	&=\|(A^T_{\hat T}A_{\hat T})^{-1}A^T_{\hat T}A_{T-\hat T}\bc^*_{T-\hat T}\|_2+\frac{1}{\sqrt{1-\delta_{M,k}}}\|\text{proj}(\be,\bA_{\widehat{T}})\|_2+\|\bc^*_{T-\widehat{T}}\|_2\\
	&\le\frac{\delta_{M,2k}}{1-\delta_{M,k}}\|\bc^*_{T-\widehat{T}}\|_2+\frac{1}{\sqrt{1-\delta_{M,k}}}\|\text{proj}(\be,\bA_{\widehat{T}})\|_2+\|\bc^*_{T-\widehat{T}}\|_2\\
	&=\frac{1+\delta_{M,2k}-\delta_{M,k}}{1-\delta_{M,k}}\|\bc^*_{T-\widehat{T}}\|_2 + \frac{1}{\sqrt{1-\delta_{M,k}}}\|\text{proj}(\be,\bA_{\widehat{T}})\|_2,
\end{align*}
where the last inequality follows from (\ref{eq.proj.ineq}) and Lemma \ref{lemma_inner}.
\end{proof}
\section{Proof of Theorem~\ref{thm.inaccurate2}}\label{sec::proof_inaccurate2}
The proof of Theorem~\ref{thm.inaccurate2} relies on the following two lemmas:
\begin{lemma}
\label{lemma:1}
Under the conditions of Theorem~\ref{thm.inaccurate}, we have
\begin{align}
	\|\bc^*_{T^*-\widetilde{T}^l}\|_2\leq a\|\bc^*_{T^*-T^{l-1}}\|+b\|\be\|_2,\label{eq::claim1}
\end{align}
where 
\begin{equation}
	\label{eq:value_ab}
	\begin{cases}a&=\delta_{M,2k}\left(\sqrt{\frac{2}{1-\delta_{M,1}}}+\frac{1}{\sqrt{1+\delta_{M,1}}}\right)\frac{1-\delta_{M,k}+\delta_{M,2k}}{1-\delta_{M,k}}\frac{\sqrt{1+\delta_{M,1}}}{1-\delta_{M,k}},\\
		b&=2\sqrt{\frac{1+\delta_{M,k}}{1-\delta_{M,1}}}\frac{\sqrt{1+\delta_{M,1}}}{1-\delta_{M,k}}.
	\end{cases}
\end{equation}
\end{lemma}
\begin{proof}[Proof of Lemma \ref{lemma:1}]
Use the relation in (\ref{eq_4}) and define $\bx_r^{l-1}$ accordingly, we decompose $\by_r^{l-1}$ as
\begin{align*}
	\by_r^{l-1}=\bA_{T^*\cup T^{l-1}}\bx_r^{l-1}+\resid(\be, \bA_{T^{l-1}}).
\end{align*}
On one hand, we have
\begin{equation}
	\label{eq:left}
	\begin{split}
		&\sqrt{\sum_{g\in T^*-\widetilde{T}^l}\|\proj(\by_r^{l-1},\bF_g)\|_2^2}\\
		\geq& \sqrt{\sum_{g\in T^*-\widetilde{T}^l}\|\proj(\bA_{T^*\cup T^{l-1}}\bx_r^{l-1},\bF_g)\|_2^2}-\sqrt{\sum_{g\in T^*-\widetilde{T}^l}\|\proj(\resid(\be, \bA_{T^{l-1}}),\bF_g)\|_2^2}.
	\end{split}
\end{equation}

Similarly to the derivations in~\eqref{eq_upper}, we can deduce
\begin{align*}
	\sum_{g\in T^*-\widetilde{T}^l}\|\proj(\resid(\be, \bA_{T^{l-1}}),\bF_g)\|_2^2
	&\leq \frac{1}{1-\delta_{M,1}}\|\bA_{T^*-\widetilde{T}^l}^\top\resid(\be,\bA_{T^{l-1}})\|_2^2\leq \frac{1+\delta_{M,k}}{1-\delta_{M,1}}\|\be\|_2^2.
\end{align*}
On the other hand, denote $T^l_\Delta =\wT^l-T^{l-1}$. We have
\begin{equation}
	\label{eq:right}
	\begin{split}
		&\sqrt{\sum_{g\in T_\Delta^l-T^*}\|\proj(\by_r^{l-1},\bF_g)\|_2^2}\\
		\leq&\sqrt{\sum_{g\in T_\Delta^l-T^*}\|\proj(\bA_{T^*\cup T^{l-1}}\bx_r^{l-1},\bF_g)\|_2^2}+\sqrt{\sum_{g\in T_\Delta^l-T^*}\|\proj(\resid(\be, \bA_{T^{l-1}}),\bF_g)\|_2^2}\\
		\leq& \sqrt{\sum_{g\in T_\Delta^l-T^*}\|\proj(\bA_{T^*\cup T^{l-1}}\bx_r^{l-1},\bF_g)\|_2^2}+\sqrt{\frac{1+\delta_{M,k}}{1-\delta_{M,1}}}\|\be\|_2.
	\end{split}
\end{equation}
Combining \eqref{eq:left} and \eqref{eq:right},  we arrive at
\begin{align}
	\sqrt{\sum_{g\in T_\Delta^l-T^*}\|\proj(\bA_{T^*\cup T^{l-1}}\bx_r^{l-1},\bF_g)\|_2^2}+2\sqrt{\frac{1+\delta_{M,k}}{1-\delta_{M,1}}}\|\be\|_2\geq \sqrt{\sum_{g\in T^*-\widetilde{T}^l}\|\proj(\bA_{T^*\cup T^{l-1}}\bx_r^{l-1},\bF_g)\|_2^2}.
\end{align}
The derivation in Section~\ref{app:thm1} is applicable. In particular, we apply~\eqref{eq_upper} and~\eqref{eq_lower} and get
\begin{align}
	&\delta_{M,2k}\left(\sqrt{\frac{2}{1-\delta_{M,1}}}+\frac{1}{\sqrt{1+\delta_{M,1}}}\right)\|\bx_r^{l-1}\|_2+2\sqrt{\frac{1+\delta_{M,k}}{1-\delta_{M,1}}}\|\be\|_2\geq  \frac{1-\delta_{M,k}}{\sqrt{1+\delta_{M,1}}}\|(\bx^{l-1}_r)_{T^*-\tilde T^l}\|_2.
\end{align}
Now applying~\eqref{eq_result1} and~\eqref{eq_result2} gives
\begin{align}
	&\delta_{M,2k}\left(\sqrt{\frac{2}{1-\delta_{M,1}}}+\frac{1}{\sqrt{1+\delta_{M,1}}}\right)\frac{1-\delta_{M,k}+\delta_{M,2k}}{1-\delta_{M,k}}\|\bc_{T^*-T^{l-1}}^*\|_2+2\sqrt{\frac{1+\delta_{M,k}}{1-\delta_{M,1}}}\|\be\|_2\geq  \frac{1-\delta_{M,k}}{\sqrt{1+\delta_{M,1}}}\|\bc^*_{T^*-\widetilde{T^l}}\|_2.
\end{align}
Taking values in \eqref{eq:value_ab} 
proves the claim~\eqref{eq::claim1}.
\end{proof}

\begin{lemma}
\label{lemma:2}
Under the conditions of Theorem~\ref{thm.inaccurate}, we have
\begin{align}
	\|\bc^*_{T^*-T^l}\|_2\leq c\|\bc^*_{T^*-\widetilde{T}^l}\|_2+d\|\be\|_2\label{eq::claim2}
\end{align}
where 
\begin{equation}
	\label{eq:value_cd}
	\begin{cases}c&=\left(1+\sqrt{\frac{1+\delta_{M,1}}{1-\delta_{M,1}}}\right)\frac{\sqrt{2}\delta_{M,2k}}{1-\delta_{M,2k}},\\
		d&=\frac{1}{\sqrt{1-\delta_{M,2k}}}\left(1+\sqrt{\frac{1+\delta_{M,1}}{1-\delta_{M,1}}}\right).
	\end{cases}
\end{equation}
\end{lemma}
\begin{proof}[Proof of Lemma \ref{lemma:2}]
Note that we have the relation
$$\bx_p^l=\bA_{\widetilde{T}^l}^\dagger\by=\bA_{\wT^l}^{\dagger}(\bA_{T^*}\bc_{T^*}^*+\be)$$
and thus
\begin{align}
	\|\bx_p^l-\bc_{\wT^l}^*\|_2&\leq \|\bA_{\wT^l}^{\dagger}\bA_{T^*}\bc^*_{T^*}-\bc^*_{\wT^l}\|_2+\|\bA_{\wT^l}^\dagger\be\|_2 \leq \|\bA_{\wT^l}^\dagger\bA_{T^*-\wT^l}\bc^*_{T^*-\wT^l}\|_2+\frac{1}{\sqrt{1-\delta_{M,2k}}}\|\be\|_2 \nonumber\\
	&\leq \frac{\sqrt{2}\delta_{M,2k}}{1-\delta_{M,2k}}\|\bc_{T^*-\wT^l}^*\|_2+\frac{1}{\sqrt{1-\delta_{M,2k}}}\|\be\|_2, \label{eq.th7.xc}
\end{align}
where the last inequality is obtained similarly as in~\eqref{eq_temp6}. On the other hand, we still have  from the derivation for~\eqref{eq_temp5} that
\begin{align}
	\left(1+\sqrt{\frac{1+\delta_{M,1}}{1-\delta_{M,1}}}\right)\|\bx_p^l-\bc_{\wT^l}^*\|_2\geq  \|\bc^*_{T^*\cap(\wT^l-T^l)}\|_2.
	\label{eq.th7.xc2}
\end{align}
Combining (\ref{eq.th7.xc}) and (\ref{eq.th7.xc2}) gives
\begin{align*}
	\|\bc_{T^*-T^l}^*\|_2&\leq\|\bc^*_{T^*\cap(\wT^l-T^l)}\|_2+\|\bc^*_{T^*-\wT^l}\|_2\\
	&\leq\left(1+\sqrt{\frac{1+\delta_{M,1}}{1-\delta_{M,1}}}\right)\frac{\sqrt{2}\delta_{M,2k}}{1-\delta_{M,2k}}\|\bc_{T^*-\wT^l}^*\|_2+\frac{1}{\sqrt{1-\delta_{M,2k}}}\left(1+\sqrt{\frac{1+\delta_{M,1}}{1-\delta_{M,1}}}\right)\|\be\|_2,
\end{align*}
which proves our claim~\eqref{eq::claim2} by setting the values of $d$, $c$ as in \eqref{eq:value_cd}. 
\end{proof}

\begin{proof}[Proof of Theorem~\ref{thm.inaccurate2}]
Combining Lemma \ref{lemma:1} and Lemma \ref{lemma:2}, we have
\begin{align}
	\|\bc^*_{T^*-T^l}\|_2\leq ac\|\bc^*_{T^*-T^{l-1}}\|_2+(ad+b)\|\be\|_2.\label{eq_temp_final}
\end{align}
Note that we have the following upper bounds
\begin{align*}
	a &\leq \frac{6\delta_{M,2k}(1+\delta_{M,2k})}{(1-\delta_{M,2k})^3},\ 
	b\leq \frac{2(1+\delta_{M,2k})}{(1-\delta_{M,2k})^2}, \ 
	c\leq \frac{\sqrt{2}\delta_{M,2k}(2-\delta_{M,2k})}{(1-\delta_{M,2k})^2},\ 
	d\leq \frac{2-\delta_{M,2k}}{(1-\delta_{M,2k})^2}.
\end{align*}
Plugging these bounds into~\eqref{eq_temp_final} gives
\begin{align*}
	\|\bc^*_{T^*-T^l}\|_2
	\leq&\frac{6\sqrt{2}\delta^2_{M,2k}(1+\delta_{M,2k})(2-\delta_{M,2k})}{(1-\delta_{M,2k})^5}\|\bc^*_{T^*-T^{l-1}}\|_2\\
	&+\frac{12\delta_{M,2k}(1+\delta_{M,2k})^2+(2-\delta_{M,2k})(1-\delta_{M,2k})^3}{(1-\delta_{M,2k})^5}\|\be\|_2,
\end{align*}
which proves~\eqref{eq::thm_inacc2}.

Now we observe that
\begin{align*}
	\|\by_r^{l}\|_2&\leq\|\resid(\bA_{T^*-T^l}\bc^*_{T^*-T^l},\bA_{T^l})\|_2+\|\resid(\be,\bA_{T^l})\|_2\\
	&\leq\|\bA_{T^*-T^l}\bc^*_{T^*-T^l}\|_2+\|\be\|_2\\
	&\leq \sqrt{1+\delta_{M,k}}\|\bc^*_{T^*-T^l}\|_2+\|\be\|_2.
\end{align*}
Moreover, by Lemma~\ref{lemma_proj}, we have 
\begin{align*}
	\|\by_r^{l-1}\|_2&\geq \|\resid(\bA_{T^*-T^{l-1}}\bc^*_{T^*-T^{l-1}},\bA_{T^{l-1}})\|_2-\|\resid(\be,\bA_{T^{l-1}})\|_2\\
	&\geq \sqrt{1-\delta_{M,2k}}\left(1-\frac{\delta_{M,2k}}{1-\delta_{M,k}}\right)\|\bc^*_{T^*-T^{l-1}}\|_2-\|\be\|_2
\end{align*}
Hence, $\|\by_r^{l}\|_2<\|\by_r^{l-1}\|_2$ if
\begin{align*}
	\sqrt{1-\delta_{M,2k}}\left(1-\frac{\delta_{M,2k}}{1-\delta_{M,k}}\right)\|\bc^*_{T^*-T^{l-1}}\|_2>\sqrt{1+\delta_{M,k}}\|\bc^*_{T^*-T^l}\|_2+2\|\be\|_2.
\end{align*}
Applying~\eqref{eq::thm_inacc2}, we see that $\|\by_r^{l}\|_2<\|\by_r^{l-1}\|_2$ if
\begin{align*}
	\left(\sqrt{1-\delta_{M,2k}}\left(1-\frac{\delta_{M,2k}}{1-\delta_{M,k}}\right)-\sqrt{1+\delta_{M,k}}D_{M,k}\right)\|\bc^*_{T^*-T^{l-1}}\|_2>\left(\sqrt{1+\delta_{M,k}}E_{M,k}+2\right)\|\be\|_2.
\end{align*}
With the assumption  $\|\be\|_2\leq \delta_{M,2k}\|\bc^*_{T^*-T^{l-1}}\|_2$, we see that $\|\by_r^{l}\|_2<\|\by_r^{l-1}\|_2$ if
\begin{align*}
	1>\frac{\delta_{M,2k}}{1-\delta_{M,2k}}+\frac{\sqrt{1+\delta_{M,2k}}D_{M,k}+\delta_{M,2k}\left(\sqrt{1+\delta_{M,2k}}E_{M,k}+2\right)}{\sqrt{1-\delta_{M,2k}}},
\end{align*}
which completes the proof.
\end{proof}

\bibliographystyle{alpha}
\bibliography{sample}
\end{document}